\documentclass{lmcs} 
\pdfoutput=1

\usepackage{lastpage}
\lmcsdoi{21}{3}{23}
\lmcsheading{}{\pageref{LastPage}}{}{}%
{Feb.~09,~2024}{Aug.~29,~2025}{}


\usepackage{enumitem}
\usepackage[utf8]{inputenc}
\usepackage{tikz}

\newcommand\CC{\mathcal{C}}
\newcommand\mono[3]{#1\,:\,#2\succ\!\rightarrow#3}
\newcommand\map[3]{#1\,:\,#2\longrightarrow#3}
\newcommand\Set{\mathbf{Set}}
\newcommand\im[1]{\text{Im(}#1\text{)}}
\newcommand\Rel[1]{\mathbf{Rel}\text{(}#1\text{)}}
\newcommand\rcomp[2]{#1;#2}
\newcommand\pair[2]{\langle#1,#2\rangle}
\newcommand\id{\text{id}}
\newcommand\A{\mathcal{A}}
\newcommand\Map[1]{\mathbf{Map}\text{(}#1\text{)}}
\newcommand\Sha{\mathbf{Scha}}
\newcommand\Eff{\mathbf{Eff}}
\newcommand\pow[1]{\mathcal{P}#1}
\newcommand\Coal[1]{\mathbf{Coal}\text{(}#1\text{)}}
\newcommand\product[2]{#1\times#2}
\newcommand\splitg[1]{\pair{#1\pi_1}{#1\pi_2}}
\newcommand\splitf{\splitg{F}}
\newcommand\Bis[1]{\mathbf{Bis}\text{(}#1\text{)}}
\newcommand\epi[3]{#1\,:\,#2\longrightarrow\!\!\!\!\!\to#3}
\newcommand\belong[1]{\!\in_{#1}\!}
\newcommand\truth{\mathbb{T}}

\newcommand\terminal{\mathbf{1}}
\newcommand\distr[2]{\sigma_{#1,#2}}
\newcommand\natt[3]{#1\,:\,#2\Longrightarrow#3}
\newcommand\splitpf{\splitg{\pow{F}}}
\newcommand\nat{\mathbb{N}}
\newcommand\natb{\overline{\mathbb{N}}}
\newcommand\viet[1]{\mathcal{V}(#1)}
\newcommand\Stone{\mathbf{Stone}}
\newcommand\Vect[1]{#1\mathbf{Vect}}
\newcommand\belongsq[1]{\!\in_{#1}^2\!}
\newcommand\monop[1]{\text{mono}(#1)}
\newcommand\strength[2]{t_{#1,#2}}
\newcommand\costrength[2]{t'_{#1,#2}}
\newcommand\laxCoal[1]{\mathbf{Coal}_{\textbf{lax}}\text{(}#1\text{)}}
\newcommand\Sim[1]{\mathbf{Sim}\text{(}#1\text{)}}
\newcommand\leqp{\leq_{\pow{\!}}}
\newcommand\geqp{\geq_{\pow{\!}}}

\begin{document}

\title{Aczel-Mendler Bisimulations in a Regular Category}
\titlecomment{{\lsuper*}Extended version of \cite{dubut23}.
This work was partially done at the National Institute of Advanced Science 
and Technology, Tokyo, Japan}

\author[J.~Dubut]{J\'{e}r\'{e}my Dubut\lmcsorcid{0000-0002-2640-3065}}

\address{LIX, CNRS, École polytechnique, Institut Polytechnique de Paris, Palaiseau, France}	
\email{jeremy.dubut@polytechnique.edu}  

\begin{abstract}
	Aczel-Mendler bisimulations are a coalgebraic extension of a variety
	of computational relations between systems.
	It is usual to assume that the underlying category satisfies some
	form of the axiom of choice, so that the collection of bisimulations enjoys desirable
	properties, such as closure under composition.
	In this paper, we accommodate the definition in general regular
	categories and toposes.
	We show that this general definition:
	1) is closed under composition without using the axiom of choice,
	2) coincides with other types of coalgebraic formulations under milder
	conditions,
	3) coincides with the usual definition when the category satisfies the
	regular axiom of choice.
	In particular, the case of toposes heavily relies on power-objects,
	for which we recover some favourable properties along the way.
	Finally, we describe several examples in Stone spaces,
	toposes for name-passing, and modules over a ring.
\end{abstract}

\maketitle

\section*{Introduction}\label{S:one}

Bisimilarity is a way to describe that two states of two systems behave in
the same way. It formalises the fact that one can mimic any
execution starting from one state with an execution from the other state,
and vice versa.
In contrast to language equivalence, which requires one to consider entire
(possibly infinite) executions, bisimilarity is a local notion, focusing only on
the next step of the execution. As such, bisimilarity is often far more
tractable than the comparison of trace languages.

Since the seminal work by Park~\cite{park81}
on labelled transition systems,
a plethora of different notions of bisimilarity has arisen in various contexts:
for probabilistic~\cite{larsen91},
timed~\cite{wang90},
hybrid~\cite{girard05},
and truly concurrent~\cite{vanglabbeek91} systems, among others.
Although they deal with very different types of systems, these notions share
common ground: connections with logic, games, fixpoints, or even
some form of decidability that exhibits a similar flavour.
This has suggested that these theories could be abstracted into a
meta-theory that captures the essence of these shared foundations.

Categorical modelling is one such effort to abstract concrete theories
into purely mathematical ones, expressed in the language of category
theory. If an earlier success in computer science lies in the
denotational semantics of programming languages
(see, for example, the Curry-Howard-Lambek correspondence,
first published in~\cite{lambek88}),
a more recent achievement is the categorical modelling of bisimulations and
computational systems using coalgebras.
In this modelling, systems are represented as coalgebras—that is,
morphisms of the form $X\,\longrightarrow\,FX$, where $X$ is an object in some
category representing the state space of the system, and $F$ is an
endofunctor on this category, representing the type of allowed transitions.
By varying the underlying category and the functor $F$, one can capture
various (known and novel) types of systems. In this abstract view,
morphisms between coalgebras play an important role: they encompass the
intuition of bisimulation maps, that is, transformations of
systems that induce bisimulations. Building on this intuition, several
abstract notions of bisimilarity can be defined, all more or less equivalent
(see~\cite{jacobs16,staton11} for an overview).

In the present paper, we are particularly interested in Aczel-Mendler 
bisimilarity~\cite{aczel89}, 
which defines a bisimulation as an abstract relation
(that is, a subobject of a product) which
itself carries a coalgebra structure, from which the coalgebra structures
of the systems being compared can be recovered via projections.
This abstract notion has the advantage of being very close to the usual notions
of bisimulation in terms of relations, but this comes at the cost of being
overly set-flavoured. For instance, some basic properties
(such as closure under composition, or their relation to bisimulation maps)
only hold when the underlying category satisfies some form of the
axiom of choice.

These issues hinder the use of Aczel-Mendler bisimulations in
certain interesting categories. Regular categories-and in particular, toposes-
form a class of categories that
enjoy very desirable properties, notably a convenient
theory of relations, which is crucial for abstract bisimulations. However, they
do not satisfy the axiom of choice. This is the case, for example, with
the effective topos~\cite{hyland82}, which internalises concepts such as decidable sets and computable functions,
or the topos of nominal sets~\cite{lawvere89}, which models name-passing and,
more generally, infinite systems possessing some form of decidability.
Being able to abstract bisimulations in such categories thus becomes
essential, offering a potential route to general decidability results.

\subsection*{Outline}

The remainder of the paper is organised as follows. 
In Section~\ref{sec:relations}, we recall some necessary background 
on relations in a general category and allegories, with a particular 
focus on maps.
This includes the definition of relations, their basic constructions 
(diagonal, composition, converse, and intersection), the definition of 
an allegory, maps and tabulations, and
finally, the characterisations of relation maps and the tabularity of 
the allegory of relations.
In Section~\ref{sec:reg-cat}, we recall the definition of Aczel-Mendler 
bisimulations and some of their properties that
only hold under certain forms of the axiom of choice. 
We then extend them to regular AM-bisimulations,
which behave well in any regular category.
Section~\ref{sec:folklore} explores the power-object monad and 
some of its well-known properties that illuminate its
role in AM-bisimulations. We recover these properties in a purely 
relational way by observing that Kleisli
composition corresponds to the composition of relations.
In Section~\ref{sec:toposes}, we present a more elegant reformulation 
of regular AM-bisimulations in toposes, enabled by
the power-object monad.
Section~\ref{sec:simulations} extends this refined formulation to 
simulations.
Finally, in Section~\ref{sec:examples}, we explore examples of 
regular AM-bisimulations for Stone spaces,
toposes modelling name-passing, and linear weighted systems.

\subsection*{Contributions}

Our contributions may be summarised as follows:
\begin{enumerate} 
	\item An extension of the theory of Aczel-Mendler bisimulations 
	that works in any regular category,
	without relying on the axiom of choice. 
	In particular, we prove that closure under
	composition (Proposition~\ref{prop:composition-reg}) and
	coincidence with other notions of coalgebraic bisimulations
	(Theorem~\ref{theo:reg-equivalences}) do not require the axiom of 
	choice.
	\item An elementary and relational account of folklore properties 
	of power-objects, including  
	the fact that they yield a commutative  
	monad whose Kleisli category is isomorphic to  
	the category of relations (Theorem~\ref{thm:pow-monad}),  
	and that there are simple conditions for the existence of (weak)  
	distributive laws  
	with respect to it (Corollary~\ref{coro:distr}).
	\item A more elegant formulation of regular AM-bisimulations  
	in the case of toposes, enabled by the power-object monad,  
	with a connection to tabulations of coalgebra homomorphisms that 
	can be established  
	(Corollary~\ref{coro:tobi-spans}), again without assuming the 
	axiom of choice.
	\item An extension of this more refined formulation to simulations 
	in a topos (Section~\ref{sec:simulations}).
\end{enumerate}

\subsection*{Related work}

Section~\ref{sec:relations} provides a summary of the material 
required from the textbook~\cite{freyd90} on allegories, 
with a particular focus on allegories of relations.
Applications of allegories, and their extensions, 
to computer science include
fuzzy logic~\cite{winter07},
logic programme compilation~\cite{arias12},
and generic programming~\cite{backhouse99}.
Topos theory has a well-established literature covering a 
variety of aspects.
For a comprehensive reference on the subject, 
we recommend~\cite{johnstone02}.
Coalgebra theory—-particularly bisimulations for coalgebras—-has 
also seen substantial recent development.
Most of the results in this paper concerning bisimulations are grounded in concepts
discussed in the textbook~\cite{jacobs16}.
A detailed comparison of various notions of coalgebraic 
bisimilarity can be found in \cite{staton11}.
Aczel-Mendler bisimulations originate from~\cite{aczel89}.
Their connection to bisimulation and simulation maps within 
a categorical
framework lies at the heart of the theory of open 
maps~\cite{joyal96,wissmann19}.

\subsection*{Comparison with the CALCO 2023 paper}

In addition to the numerous complete proofs, 
this version adds Section~\ref{sec:folklore} about the 
power-object monad.

\subsection*{Notations}

Given two morphisms $\map{f}{X}{Y}$ and $\map{g}{X'}{Y'}$ in a category with binary product, 
we denote the pairing by $\map{\pair{f}{g}}{X}{Y\times Y'}$ (if $X = X'$), and the product by 
$\map{\product{f}{g}}{X\times X'}{Y\times Y'}$.

\section{Allegory of Relations}
\label{sec:relations}

In this section, we present the general notion of relations in a category,
focusing in particular on the fact that they form a tabular allegory.
Definitions, propositions, and proofs may be found in~\cite{freyd90}.
Our main motivations for introducing allegories in this paper are:
1) to highlight that regular categories provide the appropriate level of
abstraction for studying bisimulations; and
2) to introduce maps—that is, left adjoints in allegories—which we aim to
relate to coalgebra homomorphisms, in order to provide an abstract
justification for the idea that “coalgebra homomorphisms are bisimulation maps”.

\subsection{Subobjects and Factorisations}

In this paper, subobjects will play a crucial role throughout. 
Let us then
spend some time on their definition. Fix an object $A$ of $\CC$. 
There is a
preorder on the class of monos of the form $\mono{m}{X}{A}$ defined by
$\mono{m}{X}{A} \sqsubseteq \mono{m'}{X'}{A}$ if and only if there is a
morphism $\map{u}{X}{X'}$ such that $m'\cdot u =m$. In this case, $u$ is
unique and is a mono. A \emph{subobject of $A$} is then an equivalence
class of monos with $\mono{m}{X}{A} \equiv \mono{m'}{X'}{A}$ if
$m \sqsubseteq m'$ and
$m \sqsupseteq m'$, that is, there are $u$
and $u'$
such that $m'\cdot u =m$ and $m\cdot u' =m'$. In this case, $u$ and $u'$ are
inverses of each other. The preorder on the monos becomes a partial order
on subobjects, also denoted by $\sqsubseteq$.
Throughout the paper, when reasoning about subobjects, we will instead
reason using a representing mono. This is harmless when dealing with
notions such as pullbacks and factorisations that are unique only
up to isomorphism.

\begin{exa}
	In $\Set$, since monos are injective functions, 
	subobjects of a set are in
	bijection with its subsets. 
	The order $\sqsubseteq$ then corresponds to
	the usual inclusion $\subseteq$ of sets.
\end{exa}

Given a morphism $\map{f}{A}{B}$, there is a particular subobject of $B$
called the \emph{image of $f$}. In general, it is defined as the smallest
(for $\sqsubseteq$) subobject $\im{f}$ of $B$ such that $f$ can be
factorised as
$m\cdot e$, where $m$ is any representing mono.
The existence of the image is not guaranteed in general. It is, however, when
the category $\CC$ has a nice (epi, mono)-factorisation system, as is
the case for regular categories (and so for toposes). In a regular category,
every morphism $f$ can be uniquely (up to unique isomorphism) factorised as $m\cdot e$,
where $m$ is a mono and $e$ is a regular epi, and furthermore, this
factorisation is the image factorisation. In addition, this factorisation is
functorial and is preserved by pullbacks, meaning that if we have a
commutative diagram of the following form (outer rectangle):
\begin{center}
\begin{tikzpicture}[scale=1.5]
		
	\node (A) at (0,-0.5) {\scriptsize{$A$}};
	\node (imf) at (3,-0.5) {\scriptsize{\im{f}}};
	\node (B) at (6,-0.5) {\scriptsize{$B$}};
	\node (Ap) at (0,0.5) {\scriptsize{$A'$}};
	\node (imfp) at (3,0.5) {\scriptsize{\im{f'}}};
	\node (Bp) at (6,0.5) {\scriptsize{$B'$}};
	
	\path[->,font=\scriptsize]
		(A) edge node[left]{$g$} (Ap)
		(B) edge node[right]{$h$} (Bp);
		
	\path[->,bend right=15,font=\scriptsize]
		(A) edge node[below]{$f$} (B);
		
	\path[->,bend left=15,font=\scriptsize]
		(Ap) edge node[above]{$f'$} (Bp);
		
	\path[->>,font=\scriptsize]
		(A) edge node[below]{$e$} (imf)
		(Ap) edge node[above]{$e'$} (imfp);
		
	\path[>->,font=\scriptsize]
		(imf) edge node[below]{$m$} (B)
		(imfp) edge node[above]{$m'$} (Bp);
		
	\path[->,dotted,font=\scriptsize]
		(imf) edge node[right]{$k$} (imfp);
			
\end{tikzpicture}
\end{center}
there is a (dotted) morphism that makes the two squares commute, and if 
the outer rectangle is a pullback, then the rightmost square is also a pullback.
For a gentle overview of regular categories, an interested reader can 
look into~\cite{butz98}.
\begin{exa}
	In $\Set$, the image of a function is the usual notion of image, 
	that is, the subset ${f(a) \mid a \in A}$ of $B$. Since $\Set$ 
	is regular, and regular epis are surjective functions, the image 
	factorisation is given by the 
	(surjection, injection)-factorisation of the function $f$.
\end{exa}
\begin{rem}[Pullbacks vs. weak pullbacks]
	\label{rem:weak-pullbacks}
	In many places in this paper, where pullbacks would naturally
	play a role, they can be replaced by weak pullbacks, leading to
	laxer conditions. A \emph{weak pullback} of a cospan $\map{f}{X}{Z}$ 
	and $\map{g}{Y}{Z}$ is given by a commutative square 
	(as on the left):
	\begin{center}
		\begin{tikzpicture}[scale=1.5]
				
			\node (r1r2) at (2,1) {\scriptsize{$V$}};
			\node (r1) at (2,0) {\scriptsize{$X$}};
			\node (r2) at (4,1) {\scriptsize{$Y$}};
			\node (y) at (4,0) {\scriptsize{$Z$}};
			
			\path[->,font=\scriptsize]
				(r1r2) edge node[left]{$\alpha_1$} (r1)
				(r1r2) edge node[above]{$\alpha_2$} (r2)
				(r1) edge node[below]{$f$} (y)
				(r2) edge node[right]{$g$} (y);
				
			\draw[dashed] (2.2,0.55) -- (2.5,0.55) -- (2.5,0.8);
					
		\end{tikzpicture}
		\quad
		\begin{tikzpicture}[scale=1.5]
				
			\node (r1r2) at (2,1) {\scriptsize{$U$}};
			\node (r1) at (2,0) {\scriptsize{$X$}};
			\node (r2) at (4,1) {\scriptsize{$Y$}};
			\node (y) at (4,0) {\scriptsize{$Z$}};
			
			\path[->,font=\scriptsize]
				(r1r2) edge node[left]{$\beta_1$} (r1)
				(r1r2) edge node[above]{$\beta_2$} (r2)
				(r1) edge node[below]{$f$} (y)
				(r2) edge node[right]{$g$} (y);
					
		\end{tikzpicture}
		\quad
		\begin{tikzpicture}[scale=1.5]
				
			\node (r1r2) at (2,1) {\scriptsize{$W$}};
			\node (r1) at (2,0) {\scriptsize{$X$}};
			\node (r2) at (4,1) {\scriptsize{$Y$}};
			\node (y) at (4,0) {\scriptsize{$Z$}};
			
			\path[->,font=\scriptsize]
				(r1r2) edge node[left]{$\gamma_1$} (r1)
				(r1r2) edge node[above]{$\gamma_2$} (r2)
				(r1) edge node[below]{$f$} (y)
				(r2) edge node[right]{$g$} (y);
				
			\draw (2.2,0.55) -- (2.5,0.55) -- (2.5,0.8);
					
		\end{tikzpicture}
	\end{center}
	such that, for every other commutative square as in the middle,
	there is (not necessarily a unique) morphism $\map{\phi}{U}{V}$
	with $\alpha_i\cdot\phi = \beta_i$.
	We denote them by a dashed corner (while proper pullbacks are
	denoted by plain corners).
	If we are in a category
	where pullbacks exist, weak pullbacks can be equivalently 
	reformulated as the commutative squares as on the left, 
	such that, if the pullback of $f$ and $g$ is given as on the right, 
	then the unique morphism $\map{\psi}{V}{W}$ with 
	$\gamma_i\cdot\psi = \alpha_i$ is a split epi.

	As a first example of replacement of pullbacks by weak pullbacks,
	the preservation of images by pullbacks and the functoriality 
	also imply the preservation of images by weak pullbacks, 
	in the sense that, if the outer rectangle is a weak pullback, 
	then the rightmost square is also a weak pullback.
\end{rem}
 
\subsection{Relations in a Regular Category}

From now on, let us assume that the category $\CC$ is regular, 
that is, it has finite limits and a pullback-stable
(regular epi, mono)-factorisation as described in the previous section.
Everything in this section can be done in a locally regular category, 
but less conveniently. In general:

\begin{defi}
A \emph{relation from $X$ to $Y$} is a subobject of $X\times Y$.
\end{defi}
\noindent 
Objects of $\CC$ and relations between them form a category, denoted by
$\Rel{\CC}$. The composition is defined as follows. Let
$\mono{m_r}{R}{X\times Y}$ and $\mono{m_s}{S}{Y\times Z}$ be two monos,
representing two relations, $r$ from $X$ to $Y$ and $s$ from $Y$ to $Z$.
Form the following pullback and (regular epi, mono)-factorisation:
\begin{center}
\begin{tikzpicture}[scale=1.5]
		
	\node (r1r2) at (3,1) {\scriptsize{$R\star S$}};
	\node (r1) at (3,0) {\scriptsize{$R$}};
	\node (r2) at (6,1) {\scriptsize{$S$}};
	\node (y) at (6,0) {\scriptsize{$Y$}};
	
	\path[->,font=\scriptsize]
		(r1r2) edge node[left]{$\mu_1$} (r1)
		(r1r2) edge node[above]{$\mu_2$} (r2)
		(r1) edge node[below]{$\pi_2\cdot m_r$} (y)
		(r2) edge node[right]{$\pi_1\cdot m_s$} (y);
		
	\draw (3.2,0.55) -- (3.7,0.55) -- (3.7,0.8);
			
\end{tikzpicture}
\quad
\begin{tikzpicture}[scale=1.5]
		
	\node (r1r2) at (0,0) {\scriptsize{$R\star S$}};
	\node (xz) at (4,0) {\scriptsize{$X\times Z$}};
	\node (r1sqr2) at (2,-1) {\scriptsize{$\rcomp{R}{S}$}};
	
	\path[->,font=\scriptsize]
		(r1r2) edge node[above]{$\pair{\pi_1\cdot m_r \cdot \mu_1}{\pi_2\cdot m_s \cdot \mu_2}$} (xz);
		
	\path[->>,font=\scriptsize]
		(r1r2) edge node[below]{$e_{\rcomp{r}{s}}$} (r1sqr2);
		
	\path[>->,font=\scriptsize]
		(r1sqr2) edge node[below]{~$m_{\rcomp{r}{s}}$} (xz);
			
\end{tikzpicture}
\end{center}
The composition $\rcomp{r}{s}$ from $X$ to $Z$ is then the subobject represented 
by the mono part $m_{\rcomp{r}{s}}$.

\begin{rem}[Pullbacks vs. weak pullbacks, continued]
In the definition of the composition, we chose to form a pullback, because 
we know it exists. However, the definition is unchanged if we take 
any weak pullback instead. 
\end{rem}
The \emph{identity relation $\Delta_X$} is represented by the diagonal $\mono{\pair{\id}{\id}}{X}{X\times X}$.

\begin{prop}
$\Rel{\CC}$ is a category.
\end{prop}

\begin{exa}
In $\Set$, the composition of relations is the usual one:
\[
	R;S = \{(x,z) \in X\times Z \mid \exists y \in Y.\, 
		(x,y) \in R \wedge (y,z) \in S\},
\]
while the identity relation is the usual diagonal $\Delta_X = \{(x,x) \mid x \in X\}$.
\end{exa}

Of course, $\Rel{\CC}$ has much more structure. First, since subobjects are
naturally ordered by $\sqsubseteq$, and since this order is compatible with
the composition, $\Rel{\CC}$ has a structure of a locally ordered 2-category.
Furthermore, it comes equipped with an anti-involution
$\map{(\_)^\dagger}{\Rel{\CC}^{op}}{\Rel{\CC}}$ which makes it an I-category
in the sense of \cite{freyd90}. This involution is given by the converse of a
relation, as follows. If the relation $r$ is represented by the mono
$\mono{m_r}{R}{X\times Y}$, then $r^\dagger$ is represented by
$\mono{m_{r^\dagger} = \pair{\pi_2}{\pi_1}\cdot m_r}{R}{Y\times X}$.
Finally, the meet of two relations for the partial order $\sqsubseteq$ is defined
and is called the intersection. Given $\mono{m_r}{R}{X\times Y}$ and
$\mono{m_s}{S}{X\times Y}$ representing $r$ and $s$ respectively, the
intersection $r \cap s$ is then represented by the pullback of $m_r$ and
$m_s$. Altogether:

\begin{thm}
$\Rel{\CC}$ is an allegory, meaning that all this data satisfies the modular law:
\[
(\rcomp{R}{S})\cap T \sqsubseteq \rcomp{(R\cap (\rcomp{T}{S^\dagger}))}{S}.
\]
\end{thm}

\begin{exa}
\label{ex:dagger-modular-sets}
In $\Set$, $R^\dagger$ is the usual converse of the relation R:
$
	R^\dagger = \{(y,x) \mid (x,y) \in R\}.
$
The intersection $\cap$ is the intersection of relations as sets.
Let us show what the modular law means in $\Rel{\Set}$. 
The relation $(\rcomp{R}{S})\cap T$ is given by the set
\[
	\left\{(x,z)\in T \mid \exists y.\, (x,y)\in R \wedge (y,z) \in S\right\}.
\]	
Let $(x,z)$ be in this set and fix a witness $y$ as in the definition above.
This means in particular that $(y,z)\in S$ and so $(z,y)\in S^\dagger$.
Since $(x,z)\in T$, then $(x,y) \in \rcomp{T}{S^\dagger}$.
In summary,
$(x,y) \in R\cap (\rcomp{T}{S^\dagger})$ and 
$(x,z) \in \rcomp{(R\cap (\rcomp{T}{S^\dagger}))}{S}$.
Intuitively, the modular law is an algebraic law expressing how 
composition preserves intersection in a weak way.
More generally, this law is 
crucial to make adjoints in an allegory behave like direct/inverse images, 
(see the next section, and the Frobenius reciprocity \cite{lawvere70}).
\end{exa}

\subsection{Maps in Allegories}
\label{sec:maps-allegory}

From an allegory (intuitively of relations), it is possible to recover the 
morphisms of the original category through the notion of \emph{maps}.
In a general allegory $\A$, a map is a morphism which is a left adjoint 
(in the 2-categorical sense). 
Maps form 
a subcategory of $\A$ denoted by $\Map{\A}$. In the case of an allegory of 
relations:
\begin{thm}
\label{th:maps}
$\Map{\Rel{\CC}}$ is isomorphic to $\CC$.
\end{thm}
\noindent 
The reason for it is that maps (left adjoints) in $\Rel{\CC}$ are precisely 
the relations 
represented by a mono of the form $\pair{\id}{f}$ for some morphism $f$ of 
$\CC$, justifying the remark from Example~\ref{ex:dagger-modular-sets} that 
left adjoints in an allegory behave like direct images. 
Similarly, their right adjoints are relations represented by  $\pair{f}{\id}$, 
corresponding to inverse images.
This also implies that $\Rel{\CC}$ is \emph{tabular}, that is, 
it is generated by maps in the following sense. A \emph{tabulation} of a 
morphism $\map{\phi}{X}{Y}$ in an allegory is a pair of maps 
$\map{\psi}{Z}{X}$ and $\map{\xi}{Z}{Y}$ such that 
$\phi = \xi\cdot\psi^\dagger$ and 
$\psi^\dagger\cdot\psi \cap \xi^\dagger\cdot\xi = \id_Z$.
\begin{thm}
In an allegory of relations, the tabulations of a relation $R$ are 
exactly those 
pairs of relations $(S,T)$ represented by monos of the form $\pair{\id}{f}$ 
and $\pair{\id}{g}$ respectively, with $f$ and $g$ jointly monic, and such that 
$R = \rcomp{T^\dagger}{S}$. In particular, every relation has a tabulation, 
that is, $\Rel{\CC}$ is tabular.
\end{thm}
\noindent 
The intuition of this theorem is that relations are precisely  jointly 
monic spans.

\begin{exa}
In $\Set$, maps are graphs of functions, that is, relations of the form
$
	\{(x,f(x)) \mid x \in X\}
$
for some function $\map{f}{X}{Y}$.
Consequently, every relation $R$ is the same as the span of
$\map{f}{R}{X}~(x,y)\mapsto x$ and $\map{g}{R}{Y}~(x,y)\mapsto y$, that is,
$
	R = \{(f(r),g(r)) \mid r \in R\}.
$
\end{exa}

\section{Aczel-Mendler Bisimulations, in Regular Categories}
\label{sec:reg-cat}

We now start investigating our original problem: a nice general 
theory of bisimulations in terms of relations. The development of this 
section will start with the notion of Aczel-Mendler bisimulations \cite{aczel89}, where 
systems are described as coalgebras. We will witness that one bottleneck 
of this theory is the role of the axiom of choice that is necessary to prove 
even some basic properties of this notion of bisimulations. 
This prevents the use of this notion in most regular categories.
We will then 
show that we can fix this issue by a careful usage of relations.

\subsection{Systems as Coalgebras}

In this section, we will briefly recall coalgebras, and how to model 
systems with them. For a more complete introduction, see for example
\cite{jacobs16}.

Coalgebras require two ingredients:
\begin{itemize}
	\item 
	a category $\CC$ that describes the type of state spaces of 
	our systems; and
	\item 
	an endofunctor $F$ on $\CC$ that describes the type of allowed 
	transitions.
\end{itemize}
\noindent 
A \emph{coalgebra} is then a morphism of type $\map{\alpha}{X}{FX}$.
Intuitively, $X$ is the state space of the system and $\alpha$ maps a state to 
the collection of transitions from this state.

\begin{exa}
For example, deterministic transition systems labelled in the alphabet 
$\Sigma$ can be modelled with the $\Set$-functor 
$X~\mapsto~\Sigma\Rightarrow\,X$. A coalgebra for this functor is a 
function $X~\to~\Sigma\Rightarrow\,X$. It maps a state to a function from
$\Sigma$ to $X$, describing what the next state is after reading a 
particular letter. 
Non-deterministic labelled transition systems can be described using 
the functor $X\mapsto\pow{(\Sigma\times X)}$. A coalgebra then maps a state to a 
set of transitions, given by a letter and a state, describing the states 
we can reach from another state reading a particular letter. 
Another typical example is a probabilistic system, that can be described 
using the distribution functor $\mathcal{D}$. A transition for those systems 
is then a distribution on the states, describing what is the probability of 
reaching a given state in the next step.
\end{exa}

A \emph{morphism of coalgebras} from $\map{\alpha}{X}{FX}$ to 
$\map{\beta}{Y}{FY}$
is a morphism $\map{f}{X}{Y}$ of $\CC$ such that 
 the following 
diagram commutes:
\begin{center}
\begin{tikzpicture}[scale=1]
		
	\node (r1r2) at (3,1) {\scriptsize{$X$}};
	\node (r1) at (8,1) {\scriptsize{$Y$}};
	\node (r2) at (3,0) {\scriptsize{$FX$}};
	\node (y) at (8,0) {\scriptsize{$FY$}};
	
	\path[->,font=\scriptsize]
		(r1r2) edge node[above]{$f$} (r1)
		(r1r2) edge node[left]{$\alpha$} (r2)
		(r1) edge node[right]{$\beta$} (y)
		(r2) edge node[below]{$Ff$} (y);
			
\end{tikzpicture}
\end{center}
Coalgebras on $F$ and homomorphisms of coalgebras form a category, 
which we denote by $\Coal{F}$.

\subsection{Aczel-Mendler Bisimulations of Coalgebras}
\label{sec:Aczel-Mendler-bisimulations}

In this section, we follow closely the development of 
\cite{jacobs16}. We recall the definition of Aczel-Mendler bisimulations and give some 
of their properties.

\subsubsection{AM-Bisimulations}

\begin{defi}
We say that a relation is an \emph{Aczel-Mendler bisimulation} 
(AM-bisimulation for short) from the coalgebra $\map{\alpha}{X}{FX}$ 
to $\map{\beta}{Y}{FY}$, if for any mono $\mono{r}{R}{X\times Y}$ 
representing it, there is a morphism $\map{W}{R}{FR}$, \emph{a witness}, 
such that:
\begin{center}
\begin{tikzpicture}[scale=1.5]
		
	\node (r) at (0,0) {\scriptsize{$R$}};
	\node (xy) at (3,0.5) {\scriptsize{$X\times Y$}};
	\node (cr) at (3,-0.5) {\scriptsize{$FR$}};
	\node (fxfy) at (6,0.5) {\scriptsize{$F(X)\times F(Y)$}};
	\node (fxy) at (6,-0.5) {\scriptsize{$F(X\times Y)$}};
	
	\path[->,font=\scriptsize]
		(xy) edge node[above]{$\product{\alpha}{\beta}$} (fxfy)
		(r) edge node[below]{$W$} (cr)
		(fxy) edge node[right]{$\splitf$} (fxfy)
		(r) edge node[above]{$r$} (xy)
		(cr) edge node[below]{$Fr$} (fxy);
			
\end{tikzpicture}
\end{center}
\end{defi}

\begin{exa}
In the case of non-deterministic labelled transition systems, AM-bisimulations 
correspond to the usual strong bisimulations. The function 
$W$ maps a pair $(x,y)$ of states of $\alpha$ and $\beta$ to a subset of triples
$(a,x',y')$ such that $(x',y') \in R$. 
The commutation condition means that the set $c(x)$ of transitions from 
$x$ corresponds exactly to the set 
$\{(a,x') \mid \exists y'.\, (a,x',y')\in W(x,y)\}$, and similarly for $y$. 
This implies the characteristic property of a bisimulation: 
if there is a transition $(a,x')$ from 
$x$, then there exists a transition $(a,y')$ from $y$ 
such that $(x',y') \in R$; 
and vice versa.
\end{exa}

\subsubsection{I-Category of Bisimulations, under the Axiom of Choice}

We show now that AM-bisimulations behave well 
under the regular axiom of choice.
\begin{defi}
	A category has the regular axiom of choice if every 
	regular epi is split.
\end{defi}
\begin{prop}
\label{prop:bisim-i-cat}
Assume that $\CC$ has the regular axiom of choice and that $F$ 
preserves weak pullbacks.
Then the following is an I-category in the sense of \cite{freyd90}, denoted by 
$\Bis{F}$:
\begin{itemize}
	\item 
	objects are coalgebras on $F$,
	\item 
	morphisms are AM-bisimulations,
	\item 
	$\sqsubseteq$, identities, composition, and 
	$(\_)^\dagger$ are defined as in $\Rel{\CC}$.
\end{itemize}
\noindent 
That is, diagonals are AM-bisimulations, and AM-bisimulations are 
closed under composition and converse.
\end{prop}

\begin{proof}
It boils down to proving the following three facts:
\begin{itemize}
	\item \textbf{Diagonals are Aczel-Mendler bisimulations:} We have the following commutative diagram:
	\begin{center}
\begin{tikzpicture}[scale=1.5]
		
	\node (r) at (0,0) {\scriptsize{$X$}};
	\node (xy) at (3,0.5) {\scriptsize{$X\times X$}};
	\node (cr) at (3,-0.5) {\scriptsize{$FX$}};
	\node (fxfy) at (6,0.5) {\scriptsize{$F(X)\times F(X)$}};
	\node (fxy) at (6,-0.5) {\scriptsize{$F(X\times X)$}};
	
	\path[->,font=\scriptsize]
		(xy) edge node[above]{$\product{\alpha}{\alpha}$} (fxfy)
		(r) edge node[below]{$\alpha$} (cr)
		(fxy) edge node[right]{$\splitf$} (fxfy)
		(r) edge node[above]{$\pair{\id}{\id}$} (xy)
		(cr) edge node[below]{$F\pair{\id}{\id}$} (fxy);
			
\end{tikzpicture}
\end{center}

	\item \textbf{Aczel-Mendler bisimulations are closed under converse:} Assume 
	given a witness for $r$:
	\begin{center}
\begin{tikzpicture}[scale=1.5]
		
	\node (r) at (0,0) {\scriptsize{$R$}};
	\node (xy) at (3,0.5) {\scriptsize{$X\times Y$}};
	\node (cr) at (3,-0.5) {\scriptsize{$FR$}};
	\node (fxfy) at (6,0.5) {\scriptsize{$F(X)\times F(Y)$}};
	\node (fxy) at (6,-0.5) {\scriptsize{$F(X\times Y)$}};
	
	\path[->,font=\scriptsize]
		(xy) edge node[above]{$\product{\alpha}{\beta}$} (fxfy)
		(r) edge node[below]{$W$} (cr)
		(fxy) edge node[right]{$\splitf$} (fxfy)
		(r) edge node[above]{$r$} (xy)
		(cr) edge node[below]{$Fr$} (fxy);
			
\end{tikzpicture}
\end{center}
	Then it is also a witness for $r^\dagger=\pair{\pi_2}{\pi_1}\cdot r$:
	\begin{center}
\begin{tikzpicture}[scale=1.5]
		
	\node (r) at (0,0) {\scriptsize{$R$}};
	\node (xy) at (3,0.5) {\scriptsize{$Y\times X$}};
	\node (cr) at (3,-0.5) {\scriptsize{$FR$}};
	\node (fxfy) at (6,0.5) {\scriptsize{$F(Y)\times F(X)$}};
	\node (fxy) at (6,-0.5) {\scriptsize{$F(Y\times X)$}};
	
	\path[->,font=\scriptsize]
		(xy) edge node[above]{$\product{\beta}{\alpha}$} (fxfy)
		(r) edge node[below]{$W$} (cr)
		(fxy) edge node[right]{$\splitf$} (fxfy)
		(r) edge node[above]{$r^\dagger$} (xy)
		(cr) edge node[below]{$F(r^\dagger)$} (fxy);
			
\end{tikzpicture}
\end{center}
	\item \textbf{Aczel-Mendler bisimulations are closed under composition:}
We then have two witnesses:
\begin{center}
\begin{tikzpicture}[scale=1.1]
		
	\node (r) at (0,0) {\scriptsize{$R_1$}};
	\node (xy) at (1.5,0.5) {\scriptsize{$X\times Y$}};
	\node (cr) at (1.5,-0.5) {\scriptsize{$FR_1$}};
	\node (fxfy) at (4,0.5) {\scriptsize{$F(X)\times F(Y)$}};
	\node (fxy) at (4,-0.5) {\scriptsize{$F(X\times Y)$}};
	
	\path[->,font=\scriptsize]
		(xy) edge node[above]{$\product{\alpha}{\beta}$} (fxfy)
		(r) edge node[below]{$W_1$} (cr)
		(fxy) edge node[right]{$\splitf$} (fxfy)
		(r) edge node[above]{$r_1$} (xy)
		(cr) edge node[below]{$Fr_1$} (fxy);

	\node (r) at (7.3,0) {\scriptsize{$R_2$}};
	\node (xy) at (8.8,0.5) {\scriptsize{$Y\times Z$}};
	\node (cr) at (8.8,-0.5) {\scriptsize{$FR_2$}};
	\node (fxfy) at (11.3,0.5) {\scriptsize{$F(Y)\times F(Z)$}};
	\node (fxy) at (11.3,-0.5) {\scriptsize{$F(Y\times Z)$}};
	
	\path[->,font=\scriptsize]
		(xy) edge node[above]{$\product{\beta}{\gamma}$} (fxfy)
		(r) edge node[below]{$W_2$} (cr)
		(fxy) edge node[right]{$\splitf$} (fxfy)
		(r) edge node[above]{$r_2$} (xy)
		(cr) edge node[below]{$Fr_2$} (fxy);
			
\end{tikzpicture}
\end{center}
We then want to construct a morphism $\map{W}{\rcomp{R_1}{R_2}}{F(\rcomp{R_1}{R_2})}$ such that
\begin{center}
\begin{tikzpicture}[scale=1.5]
		
	\node (r) at (0,0) {\scriptsize{$\rcomp{R_1}{R_2}$}};
	\node (xy) at (3,0.5) {\scriptsize{$X\times Z$}};
	\node (cr) at (3,-0.5) {\scriptsize{$F(\rcomp{R_1}{R_2})$}};
	\node (fxfy) at (6,0.5) {\scriptsize{$F(X)\times F(Z)$}};
	\node (fxy) at (6,-0.5) {\scriptsize{$F(X\times Z)$}};
	
	\path[->,font=\scriptsize]
		(xy) edge node[above]{$\product{\alpha}{\gamma}$} (fxfy)
		(r) edge node[below]{$W$} (cr)
		(fxy) edge node[right]{$\splitf$} (fxfy)
		(r) edge node[above]{$\rcomp{r_1}{r_2}$} (xy)
		(cr) edge node[below]{$F(\rcomp{r_1}{r_2})$} (fxy);
			
\end{tikzpicture}
\end{center}
Since $F$ preserves weak pullbacks and by definition of the composition, we
have the following weak pullback and (regular epi, mono)-factorisation:
\begin{center}
\begin{tikzpicture}[scale=1.1]
		
	\node (r1r2) at (3,1) {\scriptsize{$F(R_1\star R_2)$}};
	\node (r1) at (3,0) {\scriptsize{$FR_1$}};
	\node (r2) at (5.5,1) {\scriptsize{$FR_2$}};
	\node (y) at (5.5,0) {\scriptsize{$FY$}};
	
	\path[->,font=\scriptsize]
		(r1r2) edge node[left]{$F\mu_1$} (r1)
		(r1r2) edge node[above]{$F\mu_2$} (r2)
		(r1) edge node[below]{$F(\pi_2\cdot r_1)$} (y)
		(r2) edge node[right]{$F(\pi_1\cdot r_2)$} (y);

	\draw[dashed] (3.2,0.55) -- (3.45,0.55) -- (3.45,0.8);

	\node (r1r2) at (9,1) {\scriptsize{$R_1\star R_2$}};
	\node (xz) at (13,1) {\scriptsize{$X\times Z$}};
	\node (r1sqr2) at (11,0) {\scriptsize{$\rcomp{R_1}{R_2}$}};
	
	\path[->,font=\scriptsize]
		(r1r2) edge node[above]{$\pair{\pi_1\cdot r_1 \cdot \mu_1}{\pi_2\cdot r_2 \cdot \mu_2}$} (xz);
		
	\path[->>,font=\scriptsize]
		(r1r2) edge node[below]{$e_{\rcomp{r_1}{r_2}}$} (r1sqr2);

	\path[>->,font=\scriptsize,bend left=50]
		(r1sqr2) edge node[left]{$s$} (r1r2);
		
	\path[>->,font=\scriptsize]
		(r1sqr2) edge node[below]{$\rcomp{r_1}{r_2}$} (xz);
			
\end{tikzpicture}
\end{center}
Denote by $s$ a section of $e_{\rcomp{r_1}{r_2}}$, which exists by the 
regular axiom of choice. Then we have the following:
\begin{center}
	\begin{tabular}{rclr}
		$F(\pi_1\cdot\,r_2)\cdot\,W_2\cdot\mu_2\cdot\,s$
		& $=$ & $\beta\cdot\pi_1\cdot\,r_2\cdot\mu_2\cdot\,s$
		& \hfill ($r_2$ is AM-bisimulation)\\
		& $=$ & $\beta\cdot\pi_2\cdot\,r_1\cdot\mu_1\cdot\,s$
		& \hfill  (definition of $\mu_i$)\\
		& $=$ & $F(\pi_2\cdot\,r_1)\cdot\,W_1\cdot\mu_1\cdot\,s$
		& \hfill ($r_1$ is AM-bisimulation)
	\end{tabular}
	\end{center}
By the universal property of weak pullbacks, 
we have $\map{\phi}{\rcomp{R_1}{R_2}}{F(R_1\star R_2)}$, such that
\begin{center}
\begin{tikzpicture}[scale=1.5]
		
	\node (r1r2) at (3,1) {\scriptsize{$F(R_1\star R_2)$}};
	\node (r1) at (3,0) {\scriptsize{$FR_1$}};
	\node (r2) at (6,1) {\scriptsize{$FR_2$}};
	\node (y) at (6,0) {\scriptsize{$FY$}};
	\node (z) at (1.8,1.7) {\scriptsize{$\rcomp{R_1}{R_2}$}};
	
	\path[->,font=\scriptsize]
		(r1r2) edge node[left]{$F\mu_1$} (r1)
		(r1r2) edge node[above]{$F\mu_2$} (r2)
		(r1) edge node[below]{$F(\pi_2\cdot r_1)$} (y)
		(r2) edge node[right]{$F(\pi_1\cdot r_2)$} (y);
		
	\path[->,font=\scriptsize, bend right =20]
		(z) edge node[left]{$W_1\cdot\mu_1\cdot s$} (r1);
		
	\path[->,font=\scriptsize, bend left =20]
		(z) edge node[right]{$W_2\cdot\mu_2\cdot s$} (r2);
		
	\path[->,font=\scriptsize, dotted]
		(z) edge node[above]{$\phi$} (r1r2);
			
\end{tikzpicture}
\end{center}
Now $W = Fe_{\rcomp{r_1}{r_2}}\cdot\phi$ is the expected witness:
	\begin{center}
\begin{align*} 
     & \splitf\cdot F(\rcomp{r_1}{r_2})\cdot W &\\
     = \quad & \splitf\cdot F(\rcomp{r_1}{r_2})\cdot Fe_{\rcomp{r_1}{r_2}}\cdot\phi
    & \hfill \text{(definition of $W$)}&\\
     = \quad & \splitf\cdot F\pair{\pi_1\cdot r_1\cdot\mu_1}{\pi_2\cdot r_2\cdot\mu_2}\cdot\phi
    & \hfill  \text{(definition of $\rcomp{r_1}{r_2}$)}&\\
     = \quad & \pair{F(\pi_1\cdot r_1\cdot \mu_1)}{F(\pi_2\cdot r_2\cdot\mu_2)}\cdot \phi
    & \hfill \text{(computation on products)}&\\
     = \quad & \product{F(\pi_1\cdot r_1)}{F(\pi_2\cdot r_2)}\cdot\pair{F(\mu_1)\cdot\phi}{F(\mu_2)\cdot\phi}
    & \hfill  \text{(computation on products)}&\\
     = \quad & \product{F(\pi_1\cdot r_1)}{F(\pi_2\cdot r_2)}\cdot\pair{W_1\cdot\mu_1\cdot s}{W_2\cdot\mu_2\cdot s}
    & \hfill  \text{(definition of $\phi$)}&\\
     = \quad & \pair{\alpha\cdot\pi_1\cdot r_1\cdot\mu_1\cdot s}{\gamma\cdot\pi_2\cdot r_2 \cdot\mu_2\cdot s}
    & \hfill  \text{(definition of the $W_i$)}&\\
     = \quad & \product{\alpha}{\gamma}\cdot\pair{\pi_1\cdot r_1\cdot\mu_1}{\pi_2\cdot r_2\cdot\mu_2}\cdot s
    & \hfill  \text{(computation on products)}&\\
     = \quad & \product{\alpha}{\gamma}\cdot(\rcomp{r_1}{r_2})
    & \hfill  \text{(definition of $s$)} &\qedhere
\end{align*}
\end{center}
\end{itemize}
\end{proof}

\begin{rem}
	As we have already seen, the preservation of weak pullbacks is a crucial property
	for a functor related to relations. More surprisingly, the reliance on the
	axiom of choice is necessary to prove closure under composition.
	This was already noted in~\cite{jacobs16,staton11}.
	Sometimes, this proposition is stated under the assumption that $F$ preserves
	pullbacks. When pullbacks exist, since any functor preserves split
	epis, it follows from Remark~\ref{rem:weak-pullbacks}
	that if a functor preserves pullbacks, then it also preserves weak pullbacks.
\end{rem}

In the proof, we rely on the regular axiom of choice in the following way:
we require that the epi part
$\epi{e_{\rcomp{r_1}{r_2}}}{R_1 \star R_2}{\rcomp{R_1}{R_2}}$
of a (regular epi, mono)-factorisation be split, that is,
that there exists a section $\mono{s}{\rcomp{R_1}{R_2}}{R_1 \star R_2}$.
In $\Set$, $R_1 \star R_2$ consists of triples $(x, y, z)$ such that $(x, y) \in R_1$ and $(y, z) \in R_2$.
The section then corresponds to making a choice of such an intermediate $y$ for every
pair $(x, z)$ in the composite relation.
This kind of choice is common, for instance, in the proof that strong bisimulations
are closed under composition: given a transition $(a, x')$ from $x$,
to show that a similar transition exists from $z$,
one picks an intermediate $y$, uses the assumption that $R_1$ is a bisimulation
to obtain a transition from $y$, and finally uses that $R_2$ is a bisimulation
to conclude the argument.

\subsubsection{Bisimulation Maps are Coalgebra Homomorphisms}

In this $I$-category of bisimulations, 
we can also discuss maps and tabulations, 
as we did in the context of relations. 
Moreover, since the 2-categorical structure of $\Bis{F}$ is 
inherited from that of $\Rel{\CC}$—specifically, 
because the local posets of bisimulations embed into the 
corresponding local posets of relations—we may apply 
results from Section~\ref{sec:maps-allegory} within this setting. 
In particular, we can establish the following:

\begin{thm}
\label{th:maps-coalgebra}
Under the assumptions of Proposition~\ref{prop:bisim-i-cat}, 
$\Map{\Bis{F}}$ is isomorphic to $\Coal{F}$.
\end{thm}
\noindent 
Using results from Section~\ref{sec:maps-allegory}, proving this theorem 
boils down to
proving that bisimulations that are maps are precisely graphs of 
coalgebra homomorphisms:
\begin{prop}
\label{prop:bisim-maps}
A morphism $\map{h}{X}{Y}$ of $\CC$ is a coalgebra homomorphism from 
$\alpha$ to $\beta$ if and only if the mono 
$\mono{\pair{\id}{h}}{X}{X\times Y}$ represents 
an AM-bisimulation from $\alpha$ to $\beta$.
\end{prop}

\begin{proof}
Let us prove both implications:
\begin{itemize}[leftmargin=1.5em]
\item[$\Rightarrow$] Assume given a coalgebra homomorphism 
$\map{h}{X}{Y}$ from the coalgebra $\map{\alpha}{X}{FX}$ to 
$\map{\beta}{Y}{FY}$, that is, with 
\[Fh\cdot\alpha = \beta\cdot h.\]
We then want $\map{W}{X}{FX}$ such that 
\[\product{\alpha}{\beta}\cdot\pair{\id}{h} = \splitf\cdot F\pair{\id}{h}\cdot W.\] 
Using $W = \alpha$ does the job:
\begin{center}
\begin{tabular}{rclcr}
    $\product{\alpha}{\beta}\cdot\pair{\id}{h}$ & $=$ & $\pair{\alpha}{\beta\cdot h}$
    & & \hfill (computation on products)\\
    & $=$ & $\pair{\alpha}{Fh\cdot\alpha}$
    & & \hfill ($h$ homomorphism)\\
    & $=$ & $\splitf\cdot F\pair{\id}{h}\cdot\alpha$
    & & \hfill (computation on products)
\end{tabular}
\end{center}

\item[$\Leftarrow$] Let us assume that we have a morphism $\map{W}{X}{FX}$ such that 
\[\splitf\cdot F\pair{\id}{h}\cdot W = \product{\alpha}{\beta}\cdot\pair{\id}{h}.\]
Then:
\begin{center}
\begin{tabular}{rclcr}
    $\alpha$ & $=$ & $\pi_1\cdot\product{\alpha}{\beta}\cdot \pair{\id}{h}$
    & & \hfill (computation on products)\\
    & $=$ & $F\pi_1\cdot F\pair{\id}{h}\cdot W$
    & & \hfill (definition of $W$) \\
    & $=$ & $W$
    & & \hfill (computation on products)\\
    $\beta\cdot h$ & $=$ & $\pi_2\cdot\product{\alpha}{\beta}\cdot\pair{\id}{h}$
    & & \hfill (computation on products)\\
    & $=$ & $F\pi_2\cdot F\pair{\id}{h}\cdot W$
    & & \hfill (definition of $W$) \\
    & $=$ & $Fh\cdot W$
    & & \hfill (computation on products)
\end{tabular}
\end{center}
Consequently, \[Fh\cdot\alpha = \beta\cdot h\] and $h$ is a coalgebra homomorphism.\qedhere
\end{itemize}
\end{proof}
\noindent 
Using this characterisation of maps for AM-bisimulations, 
and using the tabularity of the allegory of relations, 
we can prove that an AM-bisimulation can be described as a span of homomorphisms 
of coalgebras, under some form of the axiom of choice (see \cite{jacobs16}). 
We can formulate this in terms of 
tabulations:
\begin{prop}
\label{prop:lobi-tab}
If $U$ is an 
AM-bisimulation from $\alpha$ to $\beta$, 
and if $\map{f}{Z}{X}$, $\map{g}{Z}{Y}$ is a tabulation of $U$, 
then there is a coalgebra structure $\gamma$ on $Z$ such that $f$ is a 
coalgebra homomorphism from $\gamma$ to $\alpha$ and 
$g$ is a coalgebra homomorphism from $\gamma$ to $\beta$.
\end{prop}

\begin{proof}
The fact that $f$, $g$ is a tabulation of $U$ means that 
$\map{\pair{f}{g}}{Z}{X\times Y}$ is a 
mono and represents $U$. The fact that $U$ is a AM-bisimulation 
gives a witness which is a $F$-coalgebra structure on $Z$.
The commutativity of the diagram defining this witness implies that $f$ and 
$g$ are coalgebra homomorphisms.
\end{proof}

\begin{cor}
\label{coro:lobi-spans}
Assume $\CC$ has the regular axiom of choice.
Assume given two coalgebras $\map{\alpha}{X}{F(X)}$ and $\map{\beta}{Y}{F(Y)}$, 
and two points $\map{p}{\ast}{X}$ and $\map{q}{\ast}{Y}$. 
Then the following 
two statements are equivalent:
\begin{enumerate}
	\item 
	There is an AM-bisimulation $\mono{r}{R}{X\times Y}$ from 
	$\alpha$ to $\beta$, and a point $\map{c}{\ast}{R}$ such that 
	$r\cdot c = \pair{p}{q}$.
	\item 
	There is a span $X\,\xleftarrow{~f~}\,Z\,\xrightarrow{~g~}\,Y$ , 
	an $F$-coalgebra structure $\gamma$ 
	on $Z$ such that $f$ is a coalgebra homomorphism from $\gamma$ to 
	$\alpha$ and $g$ from $\gamma$ to $\beta$, 
	and a point $\map{w}{\ast}{Z}$ such that $f\cdot w = p$ and 
	$g\cdot w = q$.
\end{enumerate}
\end{cor}

\begin{proof}
Let us prove both implications:
\begin{itemize}
	\item $1 \Rightarrow 2)$ By Proposition~\ref{prop:lobi-tab},
	we obtain a tabulation $\map{f}{Z}{X}$ and $\map{g}{Z}{Y}$ 
	together with $\map{\gamma}{Z}{FZ}$ that makes $f$ and $g$ 
	coalgebra homomorphisms. In particular, $U$ is represented by 
	$\pair{f}{g}$. Since $r$ also represents $U$, there is an 
	iso $\phi$ such that $\pair{f}{g}\cdot\phi = r$.
	By taking $w = \phi\cdot\,c$, we have 
	\[f\cdot\,w = \pi_1\cdot\pair{f}{g}\cdot\phi\cdot\,c = 
		\pi_1\cdot\,r\cdot\,c = \pi_1\cdot\pair{p}{q} = p,\]
	and similarly $g\cdot\,w=q$.
	\item $2 \Rightarrow 1)$ Let us assume that we have a span of
	homomorphisms. Then, since $f$ and $g$ are coalgebra homomorphisms,
	$\pair{\id}{f}$ and $\pair{\id}{g}$ represent bisimulations
	by Proposition~\ref{prop:bisim-maps}. Since
	bisimulations are closed under converse,
	$\pair{\id}{f}^\dagger$ also represents a bisimulation.
	To conclude, we would like to prove that
	$\rcomp{\pair{\id}{f}^\dagger}{\pair{\id}{g}}$ is a bisimulation by using
	the closure under composition. However, the general closure under
	composition requires both the regular axiom of choice and that $F$
	preserves weak pullbacks. But since the relations we are
	composing are of special forms, namely that
	$\pair{\id}{f}^\dagger = \pair{f}{\id}$ is a right adjoint
	and $\pair{\id}{g}$ is a map, the construction
	in the proof of Proposition~\ref{prop:bisim-i-cat} does not need
	the preservation of weak pullbacks,
	and we can conclude with just the regular axiom of choice that
	$\rcomp{\pair{\id}{f}^\dagger}{\pair{\id}{g}}$ is an AM-bisimulation.
	By definition of the composition, this bisimulation is
	represented by the mono part $r$ of the (regular epi, mono)-factorisation:
	\begin{center}		
	\begin{tikzpicture}[scale=1.5]
			
		\node (r1r2) at (0,0) {\scriptsize{$Z$}};
		\node (xz) at (4,0) {\scriptsize{$X\times Y$}};
		\node (r1sqr2) at (2,-1) {\scriptsize{$R$}};
		
		\path[->,font=\scriptsize]
			(r1r2) edge node[above]{$\pair{f}{g}$} (xz);
			
		\path[->>,font=\scriptsize]
			(r1r2) edge node[below]{$e$} (r1sqr2);
			
		\path[>->,font=\scriptsize]
			(r1sqr2) edge node[below]{$r$} (xz);
				
	\end{tikzpicture}
	\end{center}
	Now, if we have $w$ as in 2), define $c = e\cdot\,w$.
	We have 
	\[r\cdot\,c = r\cdot\,e\cdot\,w=\pair{f}{g}\cdot\,w=\pair{p}{q}.\qedhere\]
\end{itemize}
\end{proof}

\begin{rem}
Here $\ast$ is usually the terminal object (since we 
are talking about points), but it can really be any object.
\end{rem}

\subsection{Picking vs. Collecting: AM-Bisimulations for Regular Categories}

We have seen that several results about AM-bisimulations  
depend on the regular axiom of choice, preventing its usage in more 
exotic toposes and regular categories. Actually, the only occurrences are of similar flavour: 
one wants to prove some property of elements $(x,z)$ in a composition of 
relations, and for that, one has to pick a witness $y$ in between. The main 
idea of our proposal is that, instead of picking a witness (which would 
require the axiom of choice), it is enough to collect all the witnesses, prove 
properties about all of them, and make sure that there is enough of them.
This can be done in any regular category as follows:

\begin{defi}
\label{def:regAMbi}
We say that a relation $R$ is a \emph{regular AM-bisimulation} from the coalgebra 
$\map{\alpha}{X}{FX}$ to $\map{\beta}{Y}{FY}$, if for any mono 
$\mono{r}{R}{X\times Y}$ representing it, there is another relation 
represented by  
$\mono{w}{W}{FR\times R}$ such that $\pi_2\cdot w$ is a regular epi and:
\begin{center}
\begin{tikzpicture}[scale=2]
	
	\node (frr) at (0,0) {\scriptsize{$W$}};	
	\node (r) at (1.5,0.5) {\scriptsize{$R$}};
	\node (xy) at (3.5,0.5) {\scriptsize{$X\times Y$}};
	\node (pfr) at (1.5,-0.5) {\scriptsize{$FR$}};
	\node (fxfy) at (6,0) {\scriptsize{$F(X)\times F(Y)$}};
	\node (pfxy) at (3.5,-0.5) {\scriptsize{$F(X\times Y)$}};
	\node (ai) at (0.6,0.45) {\scriptsize{$\pi_2\cdot w$}};
	\node (bai) at (0.6,-0.45) {\scriptsize{$\pi_1\cdot w$}};
	\node (oi) at (5.2,-0.5) {\scriptsize{$\splitf$}};
	\node (boi) at (4.8,0.5) {\scriptsize{$\product{\alpha}{\beta}$}};
	
	\path[->,font=\scriptsize]
		(frr) edge (r)
		(xy) edge (fxfy)
		(frr) edge (pfr) 
		(pfxy) edge (fxfy)
		(pfxy) edge (fxfy)
		(r) edge node[above]{$r$} (xy)
		(pfr) edge node[below]{$Fr$} (pfxy);
			
\end{tikzpicture}
\end{center}
\end{defi}

The intuition is as follows: $W$ collects witnesses that $R$ is a bisimulation.
In particular, for a given pair $(x,y)$ in $R$, there might be several witnesses.
The fact $\pi_2\cdot w$ is a regular epi guarantees that every pair in $R$ has at least one 
witness.
Of course, we have to prove that this extends plain AM-bisimulations:

\begin{prop}
\label{prop:reg-bisim-instances}
If $\CC$ is a regular category, then
a AM-bisimulation is a regular AM-bisimulation.
If additionally $\CC$ satisfies the regular axiom of choice, then 
a regular AM-bisimulation is a AM-bisimulation.
\end{prop}

\begin{proof}
    \begin{itemize}

	\item Assume that we have a AM-bisimulation
	
	\begin{center}
        \begin{tikzpicture}[scale=1.5]
        		
        	\node (r) at (0,0) {\scriptsize{$R$}};
        	\node (xy) at (3,0.5) {\scriptsize{$X\times Y$}};
        	\node (cr) at (3,-0.5) {\scriptsize{$FR$}};
        	\node (fxfy) at (6,0.5) {\scriptsize{$F(X)\times F(Y)$}};
        	\node (fxy) at (6,-0.5) {\scriptsize{$F(X\times Y)$}};
        	
        	\path[->,font=\scriptsize]
        		(xy) edge node[above]{$\product{\alpha}{\beta}$} (fxfy)
        		(r) edge node[below]{$w$} (cr)
        		(fxy) edge node[right]{$\splitf$} (fxfy)
        		(r) edge node[above]{$r$} (xy)
        		(cr) edge node[below]{$Fr$} (fxy);
        			
        \end{tikzpicture}
        \end{center}
        Then 
    	\begin{center}
    \begin{tikzpicture}[scale=1.5]
    	
    	\node (frr) at (0,0) {\scriptsize{$R$}};	
    	\node (r) at (1.5,0.7) {\scriptsize{$R$}};
    	\node (xy) at (3.5,0.7) {\scriptsize{$X\times Y$}};
    	\node (pfr) at (1.5,-0.7) {\scriptsize{$FR$}};
    	\node (fxfy) at (6,0) {\scriptsize{$F(X)\times F(Y)$}};
    	\node (pfxy) at (3.5,-0.7) {\scriptsize{$F(X\times Y)$}};
    	\node (ai) at (0.3,0.55) {\scriptsize{$\pi_2\cdot \langle w,\id\rangle$}};
    	\node (bai) at (0.3,-0.55) {\scriptsize{$\pi_1\cdot \langle w,\id\rangle$}};
    	\node (oi) at (5.2,-0.6) {\scriptsize{$\splitf$}};
    	\node (boi) at (4.8,0.6) {\scriptsize{$\product{\alpha}{\beta}$}};
    	
    	\path[->,font=\scriptsize]
    		(frr) edge (r)
    		(xy) edge (fxfy)
    		(frr) edge (pfr) 
    		(pfxy) edge (fxfy)
    		(pfxy) edge (fxfy)
    		(r) edge node[above]{$r$} (xy)
    		(pfr) edge node[below]{$Fr$} (pfxy);
    			
    \end{tikzpicture}
    \end{center}
    witnesses $r$ as a regular AM-bisimulation.
    	\item Assume that $\CC$ has the regular axiom of choice
		and that we have a regular AM-bisimulation
    	\begin{center}
    \begin{tikzpicture}[scale=1.5]
    	
    	\node (frr) at (0,0) {\scriptsize{$W$}};	
    	\node (r) at (1.5,0.7) {\scriptsize{$R$}};
    	\node (xy) at (3.5,0.7) {\scriptsize{$X\times Y$}};
    	\node (pfr) at (1.5,-0.7) {\scriptsize{$FR$}};
    	\node (fxfy) at (6,0) {\scriptsize{$F(X)\times F(Y)$}};
    	\node (pfxy) at (3.5,-0.7) {\scriptsize{$F(X\times Y)$}};
    	\node (ai) at (0.6,0.55) {\scriptsize{$\pi_2\cdot w$}};
    	\node (bai) at (0.6,-0.55) {\scriptsize{$\pi_1\cdot w$}};
    	\node (oi) at (5.2,-0.6) {\scriptsize{$\splitf$}};
    	\node (boi) at (4.8,0.6) {\scriptsize{$\product{\alpha}{\beta}$}};
    	
    	\path[->,font=\scriptsize]
    		(frr) edge (r)
    		(xy) edge (fxfy)
    		(frr) edge (pfr) 
    		(pfxy) edge (fxfy)
    		(pfxy) edge (fxfy)
    		(r) edge node[above]{$r$} (xy)
    		(pfr) edge node[below]{$Fr$} (pfxy);
		\path[->, bend left, font=\scriptsize, dashed]
			(r) edge node[below]{$s$} (frr);
    			
    \end{tikzpicture}
    \end{center}
Since $\pi_2\cdot w$ is regular epi so is a split epi by the regular axiom of choice, 
	there is $\map{s}{R}{W}$ such that $\pi_2\cdot w\cdot s = \id$. Now let us prove that
		\begin{center}
			\begin{tikzpicture}[scale=1.5]

				\node (r) at (0,0) {\scriptsize{$R$}};
				\node (xy) at (3,0.5) {\scriptsize{$X\times Y$}};
				\node (cr) at (3,-0.5) {\scriptsize{$FR$}};
				\node (fxfy) at (6,0.5) {\scriptsize{$F(X)\times F(Y)$}};
				\node (fxy) at (6,-0.5) {\scriptsize{$F(X\times Y)$}};

				\path[->,font=\scriptsize]
					(xy) edge node[above]{$\product{\alpha}{\beta}$} (fxfy)
					(r) edge node[below]{$\pi_1\cdot w\cdot s$} (cr)
					(fxy) edge node[right]{$\splitf$} (fxfy)
					(r) edge node[above]{$r$} (xy)
					(cr) edge node[below]{$Fr$} (fxy);

			\end{tikzpicture}
			\end{center}
			witnesses $r$ as an AM-bisimulation.
		\begin{align*} 
		\splitf\cdot Fr\cdot\pi_1\cdot w\cdot s &\mkern9mu =\mkern9mu
		(\product{\alpha}{\beta})\cdot r\cdot\pi_2\cdot w\cdot s
		  \hfill &\text{($r$ regular AM-bisimulation)}\\
		& \mkern9mu=\mkern9mu  (\product{\alpha}{\beta})\cdot r
		  \hfill &\text{(definition of $s$)} \tag*{\qedhere}
	\end{align*}
	\end{itemize}
	\end{proof}
\noindent 
Also, regular bisimulations are closed under composition. This requires a mild condition on $F$
as already observed in \cite{staton11}.

\begin{defi}
We say that $F$ \emph{covers pullbacks} if for every pair of pullbacks:\\
\begin{center}
\begin{tikzpicture}[scale=1.5]
		
	\node (r1r2) at (3,0.8) {\scriptsize{$R$}};
	\node (r1) at (3,0) {\scriptsize{$X$}};
	\node (r2) at (6,0.8) {\scriptsize{$Y$}};
	\node (y) at (6,0) {\scriptsize{$Z$}};
	
	\path[->,font=\scriptsize]
		(r1r2) edge node[left]{$u$} (r1)
		(r1r2) edge node[above]{$v$} (r2)
		(r1) edge node[below]{$f$} (y)
		(r2) edge node[right]{$g$} (y);
		
	\draw (3.2,0.45) -- (3.7,0.45) -- (3.7,0.7);
			
\end{tikzpicture}
\qquad\qquad
\begin{tikzpicture}[scale=1.5]
		
	\node (r1r2) at (3,0.8) {\scriptsize{$R'$}};
	\node (r1) at (3,0) {\scriptsize{$FX$}};
	\node (r2) at (6,0.8) {\scriptsize{$FY$}};
	\node (y) at (6,0) {\scriptsize{$FZ$}};
	
	\path[->,font=\scriptsize]
		(r1r2) edge node[left]{$u'$} (r1)
		(r1r2) edge node[above]{$v'$} (r2)
		(r1) edge node[below]{$Ff$} (y)
		(r2) edge node[right]{$Fg$} (y);
		
	\draw (3.2,0.45) -- (3.7,0.45) -- (3.7,0.7);
			
\end{tikzpicture}
\end{center}
the unique morphism $\map{\gamma}{FR}{R'}$ such that $u'\cdot\gamma=F u$ and $v'\cdot\gamma=F v$
is a regular epi.
\end{defi}

\begin{rem}
When $F$ preserves weak pullbacks, then $F$ covers pullbacks.
When $\CC$ has the regular axiom of choice, then both notions coincide.
\end{rem}

\begin{prop}
\label{prop:composition-reg}
When $F$ covers pullbacks, regular AM-bisimulations are closed under compositions.
\end{prop}

\begin{proof}
Assume that we have two regular AM-bisimulations
    	\begin{center}
    \begin{tikzpicture}[scale=1.5]
    	
    	\node (frr) at (0,0) {\scriptsize{$W_1$}};	
    	\node (r) at (1.5,0.7) {\scriptsize{$R_1$}};
    	\node (xy) at (3.5,0.7) {\scriptsize{$X\times Y$}};
    	\node (pfr) at (1.5,-0.7) {\scriptsize{$FR_1$}};
    	\node (fxfy) at (6,0) {\scriptsize{$F(X)\times F(Y)$}};
    	\node (pfxy) at (3.5,-0.7) {\scriptsize{$F(X\times Y)$}};
    	\node (ai) at (0.4,0.55) {\scriptsize{$\pi_2\cdot w_1$}};
    	\node (bai) at (0.4,-0.55) {\scriptsize{$\pi_1\cdot w_1$}};
    	\node (oi) at (5.2,-0.6) {\scriptsize{$\splitf$}};
    	\node (boi) at (4.8,0.6) {\scriptsize{$\product{\alpha}{\beta}$}};
    	
    	\path[->,font=\scriptsize]
    		(frr) edge (r)
    		(xy) edge (fxfy)
    		(frr) edge (pfr) 
    		(pfxy) edge (fxfy)
    		(pfxy) edge (fxfy)
    		(r) edge node[above]{$r_1$} (xy)
    		(pfr) edge node[below]{$Fr_1$} (pfxy);
    			
    \end{tikzpicture}
    \end{center}
        	\begin{center}
    \begin{tikzpicture}[scale=1.5]
    	
    	\node (frr) at (0,0) {\scriptsize{$W_2$}};	
    	\node (r) at (1.5,0.7) {\scriptsize{$R_2$}};
    	\node (xy) at (3.5,0.7) {\scriptsize{$Y\times Z$}};
    	\node (pfr) at (1.5,-0.7) {\scriptsize{$FR_2$}};
    	\node (fxfy) at (6,0) {\scriptsize{$F(Y)\times F(Z)$}};
    	\node (pfxy) at (3.5,-0.7) {\scriptsize{$F(Y\times Z)$}};
    	\node (ai) at (0.4,0.55) {\scriptsize{$\pi_2\cdot w_2$}};
    	\node (bai) at (0.4,-0.55) {\scriptsize{$\pi_1\cdot w_2$}};
    	\node (oi) at (5.2,-0.6) {\scriptsize{$\splitf$}};
    	\node (boi) at (4.8,0.6) {\scriptsize{$\product{\beta}{\gamma}$}};
    	
    	\path[->,font=\scriptsize]
    		(frr) edge (r)
    		(xy) edge (fxfy)
    		(frr) edge (pfr) 
    		(pfxy) edge (fxfy)
    		(pfxy) edge (fxfy)
    		(r) edge node[above]{$r_2$} (xy)
    		(pfr) edge node[below]{$Fr_2$} (pfxy);
    			
    \end{tikzpicture}
    \end{center}
    and we want to prove that the composition $r_1;r_2$ is also a regular AM-bisimulation.
    This composition is defined by the following pullback and (regular epi, mono)-factorisation:
\begin{center}
\begin{tikzpicture}[scale=1.5]
		
	\node (r1r2) at (3,1) {\scriptsize{$R_1\star R_2$}};
	\node (r1) at (3,0) {\scriptsize{$R_1$}};
	\node (r2) at (6,1) {\scriptsize{$R_2$}};
	\node (y) at (6,0) {\scriptsize{$Y$}};
	
	\path[->,font=\scriptsize]
		(r1r2) edge node[left]{$\mu_1$} (r1)
		(r1r2) edge node[above]{$\mu_2$} (r2)
		(r1) edge node[below]{$\pi_2\cdot r_1$} (y)
		(r2) edge node[right]{$\pi_1\cdot r_2$} (y);
		
	\draw (3.2,0.55) -- (3.7,0.55) -- (3.7,0.8);
			
\end{tikzpicture}
\qquad
\begin{tikzpicture}[scale=1.5]
		
	\node (r1r2) at (0,0) {\scriptsize{$R_1\star R_2$}};
	\node (xz) at (4,0) {\scriptsize{$X\times Z$}};
	\node (r1sqr2) at (2,-1) {\scriptsize{$\rcomp{R_1}{R_2}$}};
	
	\path[->,font=\scriptsize]
		(r1r2) edge node[above]{$\pair{\pi_1\cdot r_1 \cdot \mu_1}{\pi_2\cdot r_2 \cdot \mu_2}$} (xz);
		
	\path[->>,font=\scriptsize]
		(r1r2) edge node[below]{$e_{\rcomp{r_1}{r_2}}$} (r1sqr2);
		
	\path[>->,font=\scriptsize]
		(r1sqr2) edge node[below]{$\rcomp{r_1}{r_2}$} (xz);
			
\end{tikzpicture}
\end{center}
Let us form the following pullback
\begin{center}
\begin{tikzpicture}[scale=1.5]
		
	\node (r1r2) at (3,1) {\scriptsize{$P$}};
	\node (r1) at (3,0) {\scriptsize{$FR_1$}};
	\node (r2) at (6,1) {\scriptsize{$FR_2$}};
	\node (y) at (6,0) {\scriptsize{$FY$}};
	
	\path[->,font=\scriptsize]
		(r1r2) edge node[left]{$\nu_1$} (r1)
		(r1r2) edge node[above]{$\nu_2$} (r2)
		(r1) edge node[below]{$F(\pi_2\cdot r_1)$} (y)
		(r2) edge node[right]{$F(\pi_1\cdot r_2)$} (y);
		
	\draw (3.2,0.55) -- (3.7,0.55) -- (3.7,0.8);
			
\end{tikzpicture}
\end{center}
and since $F$ covers pullbacks, there is a regular epi 
$\epi{e}{F(R_1\star R_2)}{P}$ such that
\[
	\nu_1\cdot e = F\mu_1 \quad \text{and} \quad \nu_2\cdot e = F\mu_2.
\]
Now, form the following three pullbacks:
\begin{center}
\begin{tikzpicture}[scale=1.5]
		
	\node (efx) at (5.5,0) {\scriptsize{$R_1$}};
	\node (fex) at (3,0) {\scriptsize{$W_1$}};
	\node (efpx) at (8,1) {\scriptsize{$R_2$}};
	\node (fepx) at (8,2) {\scriptsize{$W_2$}};
	\node (u) at (5.5,1) {\scriptsize{$R_1\star R_2$}};
	\node (ana) at (5.5,2) {\scriptsize{$\bullet$}};
	\node (anl) at (3,1) {\scriptsize{$\bullet$}};
	\node (v) at (3,2) {\scriptsize{$Q$}};
	
	\path[->>,font=\scriptsize]
		(fex) edge node[below]{$\pi_2\cdot w_1$} (efx)
		(fepx) edge node[right]{$\pi_2\cdot w_2$} (efpx)
		(v) edge (ana)
		(v) edge (anl)
		(ana) edge (u)
		(anl) edge (u);
		
	\path[->,font=\scriptsize]
		(u) edge node[left]{$\mu_1$} (efx)
		(u) edge node[above]{$\mu_2$} (efpx)
		(ana) edge (fepx)
		(anl) edge (fex);
		
	\path[->,font=\scriptsize,bend right = 30]
		(v) edge node[left]{$\kappa_1$} (fex);
		
	\path[->,font=\scriptsize,bend left = 20]
		(v) edge node[above]{$\kappa_2$} (fepx);
		
	\path[->,dashed, font=\scriptsize]
		(v) edge node[above]{$e'$} (u);

	\draw (3.2,0.55) -- (3.7,0.55) -- (3.7,0.8);
	\draw (3.2,1.55) -- (3.7,1.55) -- (3.7,1.8);
	\draw (5.7,1.55) -- (5.95,1.55) -- (5.95,1.8);
			
\end{tikzpicture}
\end{center}
Since $\pi_2\cdot w_1$ and $\pi_2\cdot w_2$ are regular epis, and 
regular epis are closed under pullbacks in a regular category,
the four morphisms forming the top-left pullback are regular epis.
Let us call $\epi{e'}{Q}{R_1\star R_2}$ the composition of those regular epis, 
so that $e'$ is also a regular epi.
Now, the following square commutes:
\begin{center}
\begin{tikzpicture}[scale=1.5]
		
	\node (r1r2) at (3,1) {\scriptsize{$Q$}};
	\node (r1) at (3,0) {\scriptsize{$F(R_1)$}};
	\node (r2) at (6,1) {\scriptsize{$F(R_2)$}};
	\node (y) at (6,0) {\scriptsize{$FY$}};
	
	\path[->,font=\scriptsize]
		(r1r2) edge node[left]{$\pi_1\cdot w_1\cdot \kappa_1$} (r1)
		(r1r2) edge node[above]{$\pi_1\cdot w_2\cdot \kappa_2$} (r2)
		(r1) edge node[below]{$F(\pi_2\cdot r_1)$} (y)
		(r2) edge node[right]{$F(\pi_1\cdot r_2)$} (y);
			
\end{tikzpicture}
\end{center}

Indeed, 
\begin{center}
\begin{tabular}{rclcr}
    $F(\pi_2\cdot r_1)\cdot\pi_1\cdot w_1\cdot\kappa_1$ & $=$ & $\beta\cdot\pi_2\cdot r_1\cdot\pi_2\cdot w_1\cdot\kappa_1$\
    & & \hfill ($r_1$ regular AM-bisimulation)\\
    & $=$ & $\beta\cdot\pi_1\cdot r_2\cdot\pi_2\cdot w_2\cdot\kappa_2$
    & & \hfill (various pullbacks)\\
    & $=$ & $F(\pi_1\cdot r_2)\cdot\pi_1\cdot w_2\cdot\kappa_2$
    & & \hfill ($r_2$ regular AM-bisimulation)
\end{tabular}
\end{center}

Then, by universality of $P$, there is a unique morphism
$\map{u}{Q}{P}$ such that 
\[
	\pi_1\cdot w_1\cdot \kappa_1 = \nu_1\cdot u \quad \text{and} \quad \pi_1\cdot w_2\cdot \kappa_2 = \nu_2\cdot u.
\]

Finally, form the following pullback
\begin{center}
\begin{tikzpicture}[scale=1.5]
		
	\node (r1r2) at (3,1) {\scriptsize{$W'$}};
	\node (r1) at (3,0) {\scriptsize{$F(R_1\star R_2)$}};
	\node (r2) at (6,1) {\scriptsize{$Q$}};
	\node (y) at (6,0) {\scriptsize{$P$}};
	
	\path[->,font=\scriptsize]
		(r1r2) edge node[left]{$v$} (r1)
		(r2) edge node[right]{$u$} (y);
	
	\path[->>,font=\scriptsize]
		(r1r2) edge node[above]{$e''$} (r2)
		(r1) edge node[below]{$e$} (y);
		
	\draw (3.2,0.55) -- (3.7,0.55) -- (3.7,0.8);
			
\end{tikzpicture}
\end{center}
Since $e$ is a regular epi and regular epi are closed under pullbacks in a regular category, 
then $e''$ is also a regular epi.

Now define 
$\map{w'=\langle F(e_{\rcomp{r_1}{r_2}})\cdot v,e_{\rcomp{r_1}{r_2}}\cdot e'\cdot e''\rangle}{W'}{F(\rcomp{R_1}{R_2})\times(\rcomp{R_1}{R_2})}$.
Observe in particular that $\pi_2\cdot w'$ is a regular epi as the composition of regular epis.
Now, take the (regular epi, mono)-factorisation of $w'$, that is, we have $\map{\rho}{W'}{W}$ regular epi and 
$\map{w}{W}{F(\rcomp{R_1}{R_2})\times(\rcomp{R_1}{R_2})}$ mono, such that, $w\cdot\rho = w'$.
Observe that $\pi_2\cdot w\cdot \rho = \pi_2\cdot w'$, and since $\pi_2\cdot w'$ is regular epi, then 
$\pi_2\cdot w$ is regular epi (this is a usual property of regular epis). It remains to prove that the following 
diagram commutes:
        	\begin{center}
    \begin{tikzpicture}[scale=1.5]
    	
    	\node (frr) at (-1,0) {\scriptsize{$W$}};	
    	\node (r) at (0.5,0.7) {\scriptsize{$\rcomp{R_1}{R_2}$}};
    	\node (xy) at (3.5,0.7) {\scriptsize{$X\times Z$}};
    	\node (pfr) at (0.5,-0.7) {\scriptsize{$F(\rcomp{R_1}{R_2})$}};
    	\node (fxfy) at (6,0) {\scriptsize{$F(X)\times F(Z)$}};
    	\node (pfxy) at (3.5,-0.7) {\scriptsize{$F(X\times Z)$}};
    	\node (ai) at (-0.6,0.45) {\scriptsize{$\pi_2\cdot w$}};
    	\node (bai) at (-0.6,-0.45) {\scriptsize{$\pi_1\cdot w$}};
    	\node (oi) at (5.2,-0.6) {\scriptsize{$\splitf$}};
    	\node (boi) at (4.8,0.6) {\scriptsize{$\product{\alpha}{\gamma}$}};
    	
    	\path[->,font=\scriptsize]
    		(frr) edge (r)
    		(xy) edge (fxfy)
    		(frr) edge (pfr) 
    		(pfxy) edge (fxfy)
    		(pfxy) edge (fxfy)
    		(r) edge node[above]{$\rcomp{r_1}{r_2}$} (xy)
    		(pfr) edge node[below]{$F(\rcomp{r_1}{r_2})$} (pfxy);
    			
    \end{tikzpicture}
    \end{center}
   Let us prove 
   \[
   	\alpha\cdot\pi_1\cdot\rcomp{r_1}{r_2}\cdot\pi_2\cdot w = 
	F(\pi_1\cdot\rcomp{r_1}{r_2})\cdot\pi_1\cdot w,
   \]
   the other side is similar.
Since $\rho$ is epi, it is then enough to prove that
\[
	\alpha\cdot\pi_1\cdot\rcomp{r_1}{r_2}\cdot\pi_2\cdot w' =
	\alpha\cdot\pi_1\cdot\rcomp{r_1}{r_2}\cdot\pi_2\cdot w\cdot\rho = 
	F(\pi_1\cdot\rcomp{r_1}{r_2})\cdot\pi_1\cdot w\cdot\rho =
	F(\pi_1\cdot\rcomp{r_1}{r_2})\cdot\pi_1\cdot w'.
\]
Indeed,
\begin{center}
\begin{align*} 
     & \alpha\cdot\pi_1\cdot\rcomp{r_1}{r_2}\cdot\pi_2\cdot w' \\
     =\quad & \alpha\cdot\pi_1\cdot\rcomp{r_1}{r_2}\cdot e_{\rcomp{r_1}{r_2}}\cdot e'\cdot e''
     & \hfill \text{(definition of $w'$)}&\\
     =\quad & \alpha\cdot\pi_1\cdot\pair{\pi_1\cdot r_1 \cdot \mu_1}{\pi_2\cdot r_2 \cdot \mu_2}\cdot e'\cdot e''
     & \hfill \text{(definition of $\rcomp{r_1}{r_2}$)}&\\
     =\quad & \alpha\cdot\pi_1\cdot r_1 \cdot \mu_1\cdot e'\cdot e''
     & \hfill \text{(computation)}&\\
     = \quad & \alpha\cdot\pi_1\cdot r_1 \cdot \pi_2\cdot w_1\cdot\kappa_1\cdot e''
     & \hfill \text{(definition of $e'$ and $\kappa_1$)}&\\
     = \quad & F(\pi_1\cdot r_1)\cdot\pi_1\cdot w_1\cdot\kappa_1\cdot e''
     & \hfill \text{($r_1$ is regular AM-bisimulation)}&\\
     = \quad& F(\pi_1\cdot r_1)\cdot\nu_1\cdot u\cdot e''
     & \hfill \text{(definition of $u$)}&\\
     = \quad& F(\pi_1\cdot r_1)\cdot\nu_1\cdot e\cdot v
     & \hfill \text{(definition of $e''$ and $v$)}&\\
     = \quad& F(\pi_1\cdot r_1\cdot\mu_1)\cdot v
     & \hfill \text{(definition of $e$)}&\\
     = \quad& F(\pi_1\cdot \rcomp{r_1}{r_2}\cdot e_{\rcomp{r_1}{r_2}})\cdot v
     & \hfill \text{(definition of $\rcomp{r_1}{r_2}$)}&\\
     = \quad& F(\pi_1\cdot \rcomp{r_1}{r_2})\cdot\pi_1\cdot w'
     & \hfill \text{(definition of $w'$)} &\qedhere
\end{align*}
\end{center}
\end{proof}

In \cite{staton11}, Staton described conditions for several coalgebraic notions of bisimulations to coincide.
In this picture, AM-bisimulations were quite weak, as they would coincide with other notions
only under some form of the axiom of choice (again). Here, we will show that the picture is much nicer 
with regular AM-bisimulations.

Let us recall two notions with which we will compare regular AM-bisimulations.

\begin{defi}
A relation from $X$ to $Y$ is a \emph{Hermida-Jacobs bisimulation} (HJ-bisimulation for short)
from $\map{\alpha}{X}{FX}$ to $\map{\beta}{Y}{FY}$ if if there is a mono $\mono{r}{R}{X\times Y}$
representing it and a morphism $\map{w}{R}{\overline{F}R}$ where 
$\overline{F}R$ is obtained by the (epi, mono)-factorisation on the left, and such that the square on the 
right commutes:
\begin{center}
\begin{tikzpicture}[scale=1.3]
		
	\node (r1r2) at (0,0) {\scriptsize{$FR$}};
	\node (xz) at (6,0) {\scriptsize{$FX\times FY$}};
	\node (r1sqr2) at (3,-0.8) {\scriptsize{$\overline{F}R$}};
	
	\path[->,font=\scriptsize]
		(r1r2) edge node[above]{$\splitf\cdot Fr$} (xz);
		
	\path[->>,font=\scriptsize]
		(r1r2) edge node[below]{$e_r$} (r1sqr2);
		
	\path[>->,font=\scriptsize]
		(r1sqr2) edge node[below]{$m_r$} (xz);
			
\end{tikzpicture}
\quad\quad
\begin{tikzpicture}[scale=1.3]
		
	\node (r1r2) at (3,0.8) {\scriptsize{$R$}};
	\node (r1) at (3,0) {\scriptsize{$X\times Y$}};
	\node (r2) at (6,0.8) {\scriptsize{$\overline{F}R$}};
	\node (y) at (6,0) {\scriptsize{$FX\times FY$}};
	
	\path[->,font=\scriptsize]
		(r1r2) edge node[left]{$r$} (r1)
		(r1r2) edge node[above]{$w$} (r2)
		(r1) edge node[below]{$\product{\alpha}{\beta}$} (y)
		(r2) edge node[right]{$m_r$} (y);
			
\end{tikzpicture}
\end{center}

A relation is a \emph{behavioural equivalence} from 
$\map{\alpha}{X}{FX}$ to $\map{\beta}{Y}{FY}$ 
if it is represented by a pullback of 
coalgebra homomorphisms, that is, if there are a coalgebra 
$\map{\gamma}{Z}{FZ}$ and two coalgebra homomorphisms
$\map{f}{\alpha}{\gamma}$ and $\map{g}{\beta}{\gamma}$ such that 
the mono $\mono{\langle u,v\rangle}{R}{X\times Y}$ obtained from 
their pullback in $\CC$ represents it.
\begin{center}
\begin{tikzpicture}[scale=1.5]
		
	\node (r1r2) at (3,0.8) {\scriptsize{$R$}};
	\node (r1) at (3,0) {\scriptsize{$X$}};
	\node (r2) at (6,0.8) {\scriptsize{$Y$}};
	\node (y) at (6,0) {\scriptsize{$Z$}};
	
	\path[->,font=\scriptsize]
		(r1r2) edge node[left]{$u$} (r1)
		(r1r2) edge node[above]{$v$} (r2)
		(r1) edge node[below]{$f$} (y)
		(r2) edge node[right]{$g$} (y);
		
	\draw (3.2,0.45) -- (3.7,0.45) -- (3.7,0.7);
			
\end{tikzpicture}
\end{center}
\end{defi}

\begin{thm}
\label{theo:reg-equivalences}
Assume that $\CC$ is a regular category. Then:
\begin{itemize}
	\item a relation is a regular AM-bisimulation if and only if it is a HJ-bisimulation,
	\item if $\CC$ has pushouts, then a regular AM-bisimulation is included in 
		a behavioural equivalence,
	\item if $F$ covers pullbacks, then a behavioural equivalence is a regular AM-bisimulation.
\end{itemize}
\end{thm}

\begin{proof}
Let us prove that regular AM-bisimulations coincide with HJ-bisimulations.
\begin{itemize}
	\item Let us assume that we have a regular AM-bisimulation
    	\begin{center}
    \begin{tikzpicture}[scale=1.5]
    	
    	\node (frr) at (0,0) {\scriptsize{$W$}};	
    	\node (r) at (1.5,0.7) {\scriptsize{$R$}};
    	\node (xy) at (3.5,0.7) {\scriptsize{$X\times Y$}};
    	\node (pfr) at (1.5,-0.7) {\scriptsize{$FR$}};
    	\node (fxfy) at (6,0) {\scriptsize{$F(X)\times F(Y)$}};
    	\node (pfxy) at (3.5,-0.7) {\scriptsize{$F(X\times Y)$}};
    	\node (ai) at (0.4,0.55) {\scriptsize{$\pi_2\cdot w$}};
    	\node (bai) at (0.4,-0.55) {\scriptsize{$\pi_1\cdot w$}};
    	\node (oi) at (5.2,-0.6) {\scriptsize{$\splitf$}};
    	\node (boi) at (4.8,0.6) {\scriptsize{$\product{\alpha}{\beta}$}};
    	
    	\path[->,font=\scriptsize]
    		(frr) edge (r)
    		(xy) edge (fxfy)
    		(frr) edge (pfr) 
    		(pfxy) edge (fxfy)
    		(pfxy) edge (fxfy)
    		(r) edge node[above]{$r$} (xy)
    		(pfr) edge node[below]{$Fr$} (pfxy);
    			
    \end{tikzpicture}
    \end{center}
	Then the following diagram (outer rectangle) commutes:
\begin{center}
\begin{tikzpicture}[scale=1.5]
		
	\node (w) at (3,1) {\scriptsize{$W$}};
	\node (r) at (3,0) {\scriptsize{$R$}};
	\node (xy) at (3,-1) {\scriptsize{$X\times Y$}};
	\node (fr) at (6,1) {\scriptsize{$FR$}};
	\node (bfr) at (6,0) {\scriptsize{$\overline{F}R$}};
	\node (fxfy) at (6,-1) {\scriptsize{$FX\times FY$}};
	
	\path[->,font=\scriptsize]
		(w) edge node[above]{$\pi_1\cdot w$} (fr)
		(xy) edge node[below]{$\product{\alpha}{\beta}$} (fxfy);
	
	\path[->>,font=\scriptsize]
		(w) edge node[left]{$\pi_2\cdot w$} (r)
		(fr) edge node[right]{$e_r$} (bfr);
		
	\path[>->,font=\scriptsize]
		(r) edge node[left]{$r$} (xy)
		(bfr) edge node[right]{$m_r$} (fxfy);
		
	\path[->,dashed,font=\scriptsize]
		(r) edge node[above]{$w'$} (bfr);
		
	\path[->,bend left = 60,font=\scriptsize]
		(fr) edge node[right]{$\splitf\cdot Fr$} (fxfy);
			
\end{tikzpicture}
\end{center}
	since $r$ is a regular AM-bisimulation.
	Furthermore, by definition $\pi_2\cdot w$ and $e_r$ are regular epis, and $r$ and $m_r$ are monos.
	So by functoriality of the (regular epi, mono)-factorisation, 
	there is $\map{w'}{R}{\overline{F}R}$ as above (dashed).
	The lower square witnesses that $r$ is an HJ-bisimulation.
	\item Assume that $r$ is an HJ-bisimulation
\begin{center}
\begin{tikzpicture}[scale=1.5]
		
	\node (r1r2) at (3,1) {\scriptsize{$R$}};
	\node (r1) at (3,0) {\scriptsize{$X\times Y$}};
	\node (r2) at (8,1) {\scriptsize{$\overline{F}R$}};
	\node (y) at (8,0) {\scriptsize{$FX\times FY$}};
	
	\path[->,font=\scriptsize]
		(r1r2) edge node[left]{$r$} (r1)
		(r1r2) edge node[above]{$w$} (r2)
		(r1) edge node[below]{$\product{\alpha}{\beta}$} (y)
		(r2) edge node[right]{$m_r$} (y);
			
\end{tikzpicture}
\end{center}
Form the following pullback:
\begin{center}
\begin{tikzpicture}[scale=1.5]
		
	\node (r1r2) at (3,1) {\scriptsize{$W$}};
	\node (r1) at (3,0) {\scriptsize{$R$}};
	\node (r2) at (6,1) {\scriptsize{$FR$}};
	\node (y) at (6,0) {\scriptsize{$\overline{F}R$}};
	
	\path[->,font=\scriptsize]
		(r1r2) edge node[above]{$u$} (r2)
		(r1) edge node[below]{$w$} (y);
	
	\path[->>,font=\scriptsize]
		(r1r2) edge node[left]{$e$} (r1)
		(r2) edge node[right]{$e_r$} (y);
		
	\draw (3.2,0.55) -- (3.7,0.55) -- (3.7,0.8);
			
\end{tikzpicture}
\end{center}
Since $e_r$ is a regular epi and regular epis are closed under pullbacks in a regular category
then $e$ is a regular epi.
If we define $\map{w' = \langle u,e\rangle}{W}{FR\times R}$, then 
$w'$ is a mono. Indeed, if we fix $\map{\phi,\psi}{Z}{W}$, such 
that $w'\cdot\phi = w'\cdot\psi$, then $\phi$ and $\psi$ are 
morphisms of cones from 
$(Z,u\cdot\phi=u\cdot\psi,e\cdot\phi=e\cdot\psi)$ to 
$(W,u,e)$. By universality of the pullback, such a morphism of cones
is unique, so $\phi = \psi$.
It remains to prove that the following diagram commutes
    	\begin{center}
    \begin{tikzpicture}[scale=1.5]
    	
    	\node (frr) at (0,0) {\scriptsize{$W$}};	
    	\node (r) at (1.5,0.7) {\scriptsize{$R$}};
    	\node (xy) at (3.5,0.7) {\scriptsize{$X\times Y$}};
    	\node (pfr) at (1.5,-0.7) {\scriptsize{$FR$}};
    	\node (fxfy) at (6,0) {\scriptsize{$F(X)\times F(Y)$}};
    	\node (pfxy) at (3.5,-0.7) {\scriptsize{$F(X\times Y)$}};
    	\node (ai) at (0.6,0.55) {\scriptsize{$e$}};
    	\node (bai) at (0.6,-0.55) {\scriptsize{$u$}};
    	\node (oi) at (5.2,-0.6) {\scriptsize{$\splitf$}};
    	\node (boi) at (4.8,0.6) {\scriptsize{$\product{\alpha}{\beta}$}};
    	
    	\path[->,font=\scriptsize]
    		(frr) edge (r)
    		(xy) edge (fxfy)
    		(frr) edge (pfr) 
    		(pfxy) edge (fxfy)
    		(pfxy) edge (fxfy)
    		(r) edge node[above]{$r$} (xy)
    		(pfr) edge node[below]{$Fr$} (pfxy);
    			
    \end{tikzpicture}
    \end{center}
Let us do it for $\alpha$, 
\begin{center}
\begin{tabular}{rclcr}
    $\alpha\cdot\pi_1\cdot r\cdot e$ 
    & $=$ & $\pi_1\cdot m_r\cdot w\cdot e$
    & & \hfill ($r$ is HJ bisimulation)\\
    & $=$ & $\pi_1\cdot m_r\cdot e_r\cdot u$
    & & \hfill (definition of $e$ and $u$)\\
    & $=$ & $\pi_1\cdot \splitf\cdot Fr\cdot u$
    & & \hfill (definition of $e_r$ and $m_r$)\\
    & $=$ & $F(\pi_1\cdot r) \cdot Fr\cdot u$
    & & \hfill (computation)
\end{tabular}
\end{center}
\end{itemize}
\noindent 
At this point we could just invoke \cite{staton11} to conclude, 
but we provide dedicated proofs here.

Let us assume that $\CC$ has pushouts and assume that we have 
a regular AM-bisimulation
\begin{center}
\begin{tikzpicture}[scale=1.5]
	
	\node (frr) at (0,0) {\scriptsize{$W$}};	
	\node (r) at (1.5,0.7) {\scriptsize{$R$}};
	\node (xy) at (3.5,0.7) {\scriptsize{$X\times Y$}};
	\node (pfr) at (1.5,-0.7) {\scriptsize{$FR$}};
	\node (fxfy) at (6,0) {\scriptsize{$F(X)\times F(Y)$}};
	\node (pfxy) at (3.5,-0.7) {\scriptsize{$F(X\times Y)$}};
	\node (ai) at (0.6,0.55) {\scriptsize{$\pi_2\cdot w$}};
	\node (bai) at (0.6,-0.55) {\scriptsize{$\pi_1\cdot w$}};
	\node (oi) at (5.2,-0.6) {\scriptsize{$\splitf$}};
	\node (boi) at (4.8,0.6) {\scriptsize{$\product{\alpha}{\beta}$}};
	
	\path[->,font=\scriptsize]
		(frr) edge (r)
		(xy) edge (fxfy)
		(frr) edge (pfr) 
		(pfxy) edge (fxfy)
		(pfxy) edge (fxfy)
		(r) edge node[above]{$r$} (xy)
		(pfr) edge node[below]{$Fr$} (pfxy);
			
\end{tikzpicture}
\end{center}
Form the following pushout:
\begin{center}
\begin{tikzpicture}[scale=1.5]
		
	\node (r1r2) at (3,1) {\scriptsize{$R$}};
	\node (r1) at (3,0) {\scriptsize{$X$}};
	\node (r2) at (6,1) {\scriptsize{$Y$}};
	\node (y) at (6,0) {\scriptsize{$Z$}};
	
	\path[->,font=\scriptsize]
		(r1r2) edge node[above]{$\pi_2\cdot r$} (r2)
		(r1) edge node[below]{$f$} (y)
		(r1r2) edge node[left]{$\pi_1\cdot r$} (r1)
		(r2) edge node[right]{$g$} (y);
		
	\draw (5.8,0.45) -- (5.3,0.45) -- (5.3,0.2);
			
\end{tikzpicture}
\end{center}
Now, forming the pullback 
\begin{center}
\begin{tikzpicture}[scale=1.5]
		
	\node (r1r2) at (3,1) {\scriptsize{$R'$}};
	\node (r1) at (3,0) {\scriptsize{$X$}};
	\node (r2) at (6,1) {\scriptsize{$Y$}};
	\node (y) at (6,0) {\scriptsize{$Z$}};
	
	\path[->,font=\scriptsize]
		(r1r2) edge node[above]{$u$} (r2)
		(r1) edge node[below]{$f$} (y)
		(r1r2) edge node[left]{$v$} (r1)
		(r2) edge node[right]{$g$} (y);
		
	\draw (3.2,0.55) -- (3.7,0.55) -- (3.7,0.8);
			
\end{tikzpicture}
\end{center}
by universality of this pullback, there is a unique morphism $\map{\kappa}{R}{R'}$ such that
\[
	r = \langle u,v\rangle\cdot\kappa,
\]
witnessing that $r \leq \langle u,v\rangle$ as monos, that is, the relation represented by $r$
is included in the relation represented by $\langle u,v\rangle$. To conclude, 
it remains to prove that $\langle u,v\rangle$ represents a behavioural equivalence, that is, there exists 
a coalgebra structure $\map{\gamma}{Z}{FZ}$ making $f$ and $g$ coalgebra homomorphisms.
Let us prove that the following square commutes
\begin{center}
\begin{tikzpicture}[scale=1.5]
		
	\node (r1r2) at (3,1) {\scriptsize{$R$}};
	\node (r1) at (3,0) {\scriptsize{$X$}};
	\node (r2) at (6,1) {\scriptsize{$Y$}};
	\node (y) at (6,0) {\scriptsize{$FZ$}};
	
	\path[->,font=\scriptsize]
		(r1r2) edge node[above]{$\pi_2\cdot r$} (r2)
		(r1) edge node[below]{$Ff\cdot\alpha$} (y)
		(r1r2) edge node[left]{$\pi_1\cdot r$} (r1)
		(r2) edge node[right]{$Fg\cdot\beta$} (y);
			
\end{tikzpicture}
\end{center}
Since $\pi_2\cdot w$ is epi it is enough to prove that
\[
	Ff\cdot\alpha\cdot\pi_1\cdot r\cdot\pi_2\cdot w = 
	Fg\cdot\beta\cdot\pi_2\cdot r\cdot\pi_2\cdot w.
\]
Indeed,
\begin{center}
\begin{tabular}{rclcr}
    $Ff\cdot\alpha\cdot\pi_1\cdot r\cdot\pi_2\cdot w$ 
    & $=$ & $F(f\cdot\pi_1\cdot r)\cdot\pi_1\cdot w$
    & & \hfill ($r$ is AM-bisimulation)\\
    & $=$ & $F(g\cdot\pi_2\cdot r)\cdot\pi_1\cdot w$
    & & \hfill (definition of $f$ and $g$)\\
    & $=$ & $Fg\cdot\beta\cdot\pi_2\cdot r\cdot\pi_2\cdot w$
    & & \hfill ($r$ is AM-bisimulation)
\end{tabular}
\end{center}
By universality of $Z$ as a pushout, there is a unique 
$\map{\gamma}{Z}{FZ}$ such that
\[
	\gamma\cdot f = Ff\cdot \alpha \quad \text{and} \quad
	\gamma\cdot g = Fg\cdot\beta,
\]
that is $f$ and $g$ are coalgebra homomorphisms.

Finally, let us assume that $F$ covers pullbacks and that we have a behavioural equivalence
\begin{center}
\begin{tikzpicture}[scale=1.5]
		
	\node (r1r2) at (3,1) {\scriptsize{$R$}};
	\node (r1) at (3,0) {\scriptsize{$X$}};
	\node (r2) at (6,1) {\scriptsize{$Y$}};
	\node (y) at (6,0) {\scriptsize{$Z$}};
	
	\path[->,font=\scriptsize]
		(r1r2) edge node[left]{$u$} (r1)
		(r1r2) edge node[above]{$v$} (r2)
		(r1) edge node[below]{$f$} (y)
		(r2) edge node[right]{$g$} (y);
		
	\draw (3.2,0.55) -- (3.7,0.55) -- (3.7,0.8);
			
\end{tikzpicture}
\end{center}
with $\map{f}{\alpha}{\gamma}$ and $\map{g}{\beta}{\gamma}$ coalgebra homomorphisms.
Form the following pullback
\begin{center}
\begin{tikzpicture}[scale=1.5]
		
	\node (r1r2) at (3,1) {\scriptsize{$P$}};
	\node (r1) at (3,0) {\scriptsize{$FX$}};
	\node (r2) at (6,1) {\scriptsize{$FY$}};
	\node (y) at (6,0) {\scriptsize{$FZ$}};
	
	\path[->,font=\scriptsize]
		(r1r2) edge node[left]{$\mu$} (r1)
		(r1r2) edge node[above]{$\nu$} (r2)
		(r1) edge node[below]{$Ff$} (y)
		(r2) edge node[right]{$Fg$} (y);
		
	\draw (3.2,0.55) -- (3.7,0.55) -- (3.7,0.8);
			
\end{tikzpicture}
\end{center}
Since $F$ covers pullbacks, there is a regular epi $\epi{e}{FR}{P}$ such that
\[
	Fu = \mu\cdot e \quad \text{and} \quad Fv = \nu\cdot e.
\]
Now the following square commutes
\begin{center}
\begin{tikzpicture}[scale=1.5]
		
	\node (r1r2) at (3,1) {\scriptsize{$R$}};
	\node (r1) at (3,0) {\scriptsize{$FX$}};
	\node (r2) at (6,1) {\scriptsize{$FY$}};
	\node (y) at (6,0) {\scriptsize{$FZ$}};
	
	\path[->,font=\scriptsize]
		(r1r2) edge node[left]{$\alpha\cdot u$} (r1)
		(r1r2) edge node[above]{$\beta\cdot v$} (r2)
		(r1) edge node[below]{$Ff$} (y)
		(r2) edge node[right]{$Fg$} (y);
			
\end{tikzpicture}
\end{center}
Indeed,
\begin{center}
\begin{tabular}{rclcr}
    $Ff\cdot\alpha\cdot u$ 
    & $=$ & $\gamma\cdot f\cdot u$
    & & \hfill ($f$ is coalgebra homomorphism)\\
    & $=$ & $\gamma\cdot g\cdot v$
    & & \hfill (definition of $u$ and $v$)\\
    & $=$ & $Fg\cdot\beta\cdot v$
    & & \hfill ($g$ is coalgebra homomorphism)
\end{tabular}
\end{center}
By universality of $P$ as a pullback, there is a unique $\map{\theta}{R}{P}$ such that
\[
	\alpha\cdot u = \mu \cdot \theta \quad \text{and} \quad
	\beta\cdot v = \nu\cdot\theta.
\]
Then form the following pullback:
\begin{center}
\begin{tikzpicture}[scale=1.5]
		
	\node (r1r2) at (3,1) {\scriptsize{$W$}};
	\node (r1) at (3,0) {\scriptsize{$FR$}};
	\node (r2) at (6,1) {\scriptsize{$P$}};
	\node (y) at (6,0) {\scriptsize{$R$}};
	
	\path[->,font=\scriptsize]
		(r1r2) edge node[left]{$\theta'$} (r1)
		(r2) edge node[right]{$\theta$} (y);
		
	\path[->>,font=\scriptsize]
		(r1r2) edge node[above]{$e'$} (r2)
		(r1) edge node[below]{$e$} (y);
		
	\draw (3.2,0.55) -- (3.7,0.55) -- (3.7,0.8);
			
\end{tikzpicture}
\end{center}
Since $e$ is a regular epi, and regular epis are closed under pullbacks in a regular category,
$e'$ is also a regular epi.
So it remains to prove that the following diagram commutes
\begin{center}
\begin{tikzpicture}[scale=1.5]
	
	\node (frr) at (0,0) {\scriptsize{$W$}};	
	\node (r) at (1.5,0.7) {\scriptsize{$R$}};
	\node (xy) at (3.5,0.7) {\scriptsize{$X\times Y$}};
	\node (pfr) at (1.5,-0.7) {\scriptsize{$FR$}};
	\node (fxfy) at (6,0) {\scriptsize{$F(X)\times F(Y)$}};
	\node (pfxy) at (3.5,-0.7) {\scriptsize{$F(X\times Y)$}};
	\node (ai) at (0.6,0.55) {\scriptsize{$e'$}};
	\node (bai) at (0.6,-0.55) {\scriptsize{$\theta'$}};
	\node (oi) at (5.2,-0.6) {\scriptsize{$\splitf$}};
	\node (boi) at (4.8,0.6) {\scriptsize{$\product{\alpha}{\beta}$}};
	
	\path[->,font=\scriptsize]
		(frr) edge (r)
		(xy) edge (fxfy)
		(frr) edge (pfr) 
		(pfxy) edge (fxfy)
		(pfxy) edge (fxfy)
		(r) edge node[above]{$\langle u,v \rangle$} (xy)
		(pfr) edge node[below]{$F\langle u,v \rangle$} (pfxy);
			
\end{tikzpicture}
\end{center}
Let us prove it for $\alpha$ (the other side is similar):
\begin{center}
\begin{align*} 
    \alpha\cdot u\cdot e'
    & \quad=\quad  \mu\cdot\theta\cdot e'
     & \hfill \text{(definition of $\theta$)}\\
    & \quad=\quad  \mu\cdot e\cdot \theta'
     & \hfill \text{(definition of $\theta'$ and $e'$)}\\
    & \quad=\quad  Fu\cdot\theta'
     & \hfill \text{(definition of $e$)} \tag*{\qedhere}
\end{align*}
\end{center}
\end{proof}

\noindent 
In Section~\ref{sec:Aczel-Mendler-bisimulations}, we described that AM-bisimilarity coincides with 
the existence of a span of coalgebra homomorphisms. This can also be formulated in 
the context of regular AM-bisimulations. The witness $\mono{w}{W}{FR\times R}$ 
can be seen as a coalgebra in $\Rel{\CC}$ (although $F$ is technically not a functor on it). 
The coalgebra $\map{\alpha}{X}{FX}$ can also 
be seen as a coalgebra in $\Rel{\CC}$ as $\mono{\langle \alpha,\id\rangle}{X}{FX\times X}$. 
Then $\pi_1\cdot r$ can be seen as 
a coalgebra homomorphism from $w$ to $\alpha$, since the following diagram commutes
\begin{center}
\begin{tikzpicture}[scale=1.5]
		
	\node (r1r2) at (3,0.8) {\scriptsize{$W$}};
	\node (r1) at (3,0) {\scriptsize{$X$}};
	\node (r2) at (6,0.8) {\scriptsize{$FR\times R$}};
	\node (y) at (6,0) {\scriptsize{$FX\times X$}};
	
	\path[->,font=\scriptsize]
		(r1r2) edge node[left]{$\pi_1\cdot r \cdot\pi_2\cdot w$} (r1)
		(r2) edge node[right]{$\product{F(\pi_1\cdot r)}{\pi_1\cdot r}$} (y);
		
	\path[>->,font=\scriptsize]
		(r1) edge node[below]{$\langle\alpha,\id\rangle$} (y)
		(r1r2) edge node[above]{$w$} (r2);
			
\end{tikzpicture}
\end{center}
Regular AM-bisimulations can be interpreted as spans of coalgebra 
homomorphisms in $\Rel{\CC}$.

\section{The Relational Essence of Power-Objects in a Topos}
\label{sec:folklore}

In this section, we investigate toposes and their power-objects in a purely relational way.
The gain is that some ingredients of the proof, particularly the precise correspondence 
between composition of relations and Kleisli composition, will be used later on.
From this observation, we (re)prove that 1) power-objects form a commutative monad 
whose Kleisli category is isomorphic to the category of relations, 
2) power-objects behave well with epis, 
3) under some mild conditions on a monad in terms of weak pullbacks and epis, there is a (weak) distributive
law with respect to the power-object monad.
During the proofs, we will denote by $\monop{f}$ the mono part of 
the (epi, mono)-factorisation of $f$.

This section is mostly directed at coalgebraists who are not very familiar 
with toposes. The results here are known (sometimes folklore) but 
scattered in the rich literature. However, the proofs of the statements 
as presented in this section, which we call ``relational'' as they 
only rely on properties of relations, could not be found anywhere.
In total, this section should be seen an an advertisement that 
1) many things that are done in coalgebra in $\Set$ with the powerset
functor can be done automatically in any topos with the power-object 
functor, and 2) anyone intersted in toposes should invest in learning 
about the internal logic of a topos, as this makes the rather 
technical relational proofs much more concise.

\subsection{Toposes, as Relation Classifiers}

\begin{defi}
A topos is a finitely complete category with power-objects. The latter 
condition means that for every object $X$, there is a mono 
$\mono{\belong{X}}{E_X}{X\times\pow{X}}$ such that for every mono of the form 
$\mono{m}{R}{X\times Y}$ there is a unique morphism 
$\map{\xi_m}{Y}{\pow{X}}$ such that there is a pullback diagram of the form:
\begin{center}
\begin{tikzpicture}[scale=2]
		
	\node (r1r2) at (3,1) {\scriptsize{$R$}};
	\node (r1) at (3,0) {\scriptsize{$X\times Y$}};
	\node (r2) at (6,1) {\scriptsize{$E_X$}};
	\node (y) at (6,0) {\scriptsize{$X\times\pow{X}$}};
	
	\path[>->,font=\scriptsize]
		(r1r2) edge node[left]{$m$} (r1)
		(r2) edge node[right]{$\belong{X}$} (y);
		
	\path[->,font=\scriptsize]
		(r1r2) edge node[above]{$\theta_m$} (r2)
		(r1) edge node[below]{$\product{\id}{\xi_m}$} (y);
		
	\draw (3.2,0.55) -- (3.7,0.55) -- (3.7,0.8);
			
\end{tikzpicture}
\end{center}
\end{defi}
Here $\theta_m$ is not required to be unique, only $\xi_m$ is.
This formulation passes to relations since $\xi_m=\xi_{m'}$ if and only if $m$ and $m'$ represent the same 
relation $r$. In that case, we will write $\xi_r$ for $\xi_m=\xi_{m'}$.
Another formulation of toposes uses sub-object classifiers
which can be recovered as 
$\truth=~\belong{\terminal}\,:\,\terminal\simeq E_\terminal\to\terminal\times\pow{\terminal}\simeq\pow{\terminal}$. 
The formulation by power-objects implies that a topos is closed, which is not the case for the one by 
sub-object classifiers.
Conversely, $\pow{X}$ is equal to $\Omega^X$ and 
$\belong{X}$ is any mono corresponding to the evaluation morphism 
$X\times\Omega^X\to\Omega$ of the cartesian-closed structure.

\begin{exa}
In $\Set$, $\pow{X}$ is given by the usual power-set and $E_X$ is the 
subset of $X\times\pow{X}$ consisting of pairs $(x,U)$ such that $x \in U$.
In $\Sha$-the Schanuel topos $\Sha$ \cite{lawvere89}, equivalent to the 
category of nominal sets and equivariant 
functions-$\pow{X}$ is the nominal set of finitely supported subsets of $X$.
In $\Eff$-the effective topos~\cite{hyland82}, intuitively, the category of effective sets and computable 
functions-$\pow{X}$ is intuitively given by the set of decidable subsets of $X$ 
(although the formal description is much more abstract).
\end{exa}

\subsection{The Power-Object Monad}

The following is a folklore result about power-objects that can be proved, for example, by noticing 
that the proof in $\Set$ does not use either the law of excluded-middle nor the axiom of choice,
and the fact that any such statement is true in any topos:
\begin{thm} 
\label{thm:pow-monad}
In a topos $\CC$,
$\pow{\!}$ extends to a commutative monad whose Kleisli category is isomorphic to 
the category of relations $\Rel{\CC}$.
\end{thm} 
\noindent 
During the course of this section, we will give an elementary and relational proof of this statement.

Let us describe some parts of this statement that will be useful in the following discussion.
First, the structure of a \emph{covariant} functor (not to be confused with the 
contravariant structure that is also sometimes used) is given as follows.
Given a morphism $\map{f}{X}{Y}$, $\map{\pow{f}}{\pow{X}}{\pow{Y}}$ is 
defined as follows. Consider first the following (epi, mono)-factorisation:
\begin{center}
\begin{tikzpicture}[scale=1.5]
		
	\node (r1r2) at (0,0) {\scriptsize{$E_X$}};
	\node (xz) at (5,0) {\scriptsize{$Y\times\pow{X}$}};
	\node (r1sqr2) at (2.5,-1) {\scriptsize{$E_f$}};
	
	\path[->,font=\scriptsize]
		(r1r2) edge node[above]{$(\product{f}{\id})\cdot\belong{X}$} (xz);
		
	\path[->>,font=\scriptsize]
		(r1r2) edge node[below]{$e_f$} (r1sqr2);
		
	\path[>->,font=\scriptsize]
		(r1sqr2) edge node[below]{$m_f$} (xz);
			
\end{tikzpicture}
\end{center}
Then $\map{\pow{f}}{\pow{X}}{\pow{Y}}$ is the unique morphism corresponding to $m_f$.

The unit $\map{\eta_X}{X}{\pow{X}}$ is defined as $\xi_{\Delta_X}$, that is, the unique morphism such that there is a pullback of the form:
\begin{center}
\begin{tikzpicture}[scale=1.5]
		
	\node (r1r2) at (3,1) {\scriptsize{$X$}};
	\node (r1) at (3,0) {\scriptsize{$X\times X$}};
	\node (r2) at (8,1) {\scriptsize{$X$}};
	\node (y) at (8,0) {\scriptsize{$X\times\pow{X}$}};
	
	\path[->,font=\scriptsize]
		(r1r2) edge node[left]{$\pair{\id}{\id}$} (r1)
		(r1r2) edge node[above]{$\theta_X$} (r2)
		(r1) edge node[below]{$\product{\id}{\eta_X}$} (y)
		(r2) edge node[right]{$\belong{X}$} (y);
		
	\draw (3.2,0.55) -- (3.7,0.55) -- (3.7,0.8);
			
\end{tikzpicture}
\end{center}
for some $\theta_X$.
The multiplication $\map{\mu_X}{\pow{\pow{X}}}{\pow{X}}$ is  
defined as the unique morphism associated with the composition of 
relations $\rcomp{\belong{X}}{\belong{\pow{X}}}$. 
In diagrams, this means that 
we form a similar pattern of pullback followed by (epi, mono)-factorisation:
\begin{center}
\begin{tikzpicture}[scale=1.1]
		
	\node (r1r2) at (3,1) {\scriptsize{$E^3_X$}};
	\node (r1) at (3,0) {\scriptsize{$E_X$}};
	\node (r2) at (7,1) {\scriptsize{$E_{\pow{X}}$}};
	\node (y) at (7,0) {\scriptsize{$\pow{X}$}};
	
	\path[->,font=\scriptsize]
		(r1r2) edge node[left]{$\kappa_{1,X}$} (r1)
		(r1r2) edge node[above]{$\kappa_{2,X}$} (r2)
		(r1) edge node[below]{$\pi_2\cdot\belong{X}$} (y)
		(r2) edge node[right]{$\pi_1\cdot\belong{\pow{X}}$} (y);
		
	\draw (3.2,0.55) -- (3.7,0.55) -- (3.7,0.8);
			
\end{tikzpicture}
\qquad
\begin{tikzpicture}[scale=1.1]
		
	\node (r1r2) at (0,0) {\scriptsize{$E^3_X$}};
	\node (xz) at (6,0) {\scriptsize{$X\times\pow{\pow{X}}$}};
	\node (r1sqr2) at (3,-1) {\scriptsize{$E^2_X$}};
	
	\path[->,font=\scriptsize]
		(r1r2) edge node[above]{$\langle \pi_1\cdot\belong{X}\cdot\kappa_{1,X}, \pi_2\cdot\belong{\pow{X}}\cdot\kappa_{2,X}\rangle$} (xz);
		
	\path[->>,font=\scriptsize]
		(r1r2) edge node[below]{$\rho_X$} (r1sqr2);
		
	\path[>->,font=\scriptsize]
		(r1sqr2) edge node[below]{$\in^2_X$} (xz);
			
\end{tikzpicture}
\end{center}
and define $\mu_X$ as the unique morphism $\xi_{\in^2_X}$.

\subsection{The Kleisli Category is the Allegory of Relations}

The operator $\xi$ obtained from the definition connects a topos 
with the opposite of its category of relations. It maps a relation from $X$ to $Y$ 
to a morphism of the form $Y\to \pow{X}$, that is, a Kleisli morphism for 
$\pow{\!}$. The definition of a topos means that this is a 
one-to-one correspondence. To show that the Kleisli category and the opposite of the 
category of relations coincide, it is then enough that the composition and 
the identities are preserved by the operator $\xi$. For the identities, it is 
by design: the identities of the Kleisli category are given by the units, which 
are \emph{defined} as $\xi_{\Delta_X}$, and the diagonals are the identity 
relations. 

The only remaining part is then about compositions.
This is the main technical result of this section. In plain words, the following proposition 
means that $\xi$ maps the opposite of the composition of relations to the Kleisli 
composition:
\begin{prop}
\label{prop:composition}
Given two relations, $r$ from $X$ to $Y$ and $s$ from $Y$ to $Z$, 
$
	\xi_{\rcomp{r}{s}} = \mu_X\cdot\pow{\xi_r}\cdot\xi_s.
$
\end{prop}

The proof is quite technical and relies on a lot of diagram chasing.

\begin{proof}
The main trick is to prove that we have a pullback of the form
\begin{center}
\begin{tikzpicture}[scale=1.5]
		
	\node (r1r2) at (3,1) {\scriptsize{$\rcomp{R}{S}$}};
	\node (r1) at (3,0) {\scriptsize{$X\times Z$}};
	\node (r2) at (8,1) {\scriptsize{$E_X^2$}};
	\node (y) at (8,0) {\scriptsize{$X\times\pow{\pow{X}}$}};
	
	\path[->,font=\scriptsize]
		(r1r2) edge node[left]{$\rcomp{r}{s}$} (r1)
		(r1r2) edge node[above]{} (r2)
		(r1) edge node[below]{$\product{\id}{(\pow{\xi_r}\cdot\xi_s)}$} (y)
		(r2) edge node[right]{$\belongsq{X}$} (y);
		
	\draw (3.2,0.55) -- (3.7,0.55) -- (3.7,0.8);
			
\end{tikzpicture}
\end{center}
by using the preservation of the image by pullback on a suitable pullback.
 Then considering the following composition of pullbacks
\begin{center}
\begin{tikzpicture}[scale=1.5]
		
	\node (r1r2) at (2,1) {\scriptsize{$\rcomp{R}{S}$}};
	\node (r1) at (2,0) {\scriptsize{$X\times Z$}};
	\node (r2) at (5.5,1) {\scriptsize{$E_X^2$}};
	\node (y) at (5.5,0) {\scriptsize{$X\times\pow{\pow{X}}$}};
	\node (ez) at (9,1) {\scriptsize{$E_X$}};
	\node (zpz) at (9,0) {\scriptsize{$X\times\pow{X}$}};
	
	\path[->,font=\scriptsize]
		(r1r2) edge node[left]{$\rcomp{r}{s}$} (r1)
		(r1r2) edge  (r2)
		(r1) edge node[below]{$\product{\id}{(\pow{\xi_r}\cdot\xi_s)}$} (y)
		(r2) edge node[left]{$\belongsq{X}$} (y)
		(r2) edge (ez)
		(y) edge node[below]{$\product{\id}{\mu_X}$} (zpz)
		(ez) edge node[right]{$\belong{X}$} (zpz);
		
	\draw (2.2,0.55) -- (2.45,0.55) -- (2.45,0.8);
	\draw (5.7,0.55) -- (5.95,0.55) -- (5.95,0.8);
			
\end{tikzpicture}
\end{center}
does the job.

First, let us describe the pullbacks and the factorisations we have 
by assumption, to introduce notations. 
By definition of $\xi_r$ and $\xi_s$, we have the following two pullbacks:
\begin{center}
\begin{tikzpicture}[scale=1.5]
		
	\node (r1r2) at (3,1) {\scriptsize{$R$}};
	\node (r1) at (3,0) {\scriptsize{$X\times Y$}};
	\node (r2) at (5.5,1) {\scriptsize{$E_X$}};
	\node (y) at (5.5,0) {\scriptsize{$X\times\pow{X}$}};
	
	\path[->,font=\scriptsize]
		(r1r2) edge node[left]{$r$} (r1)
		(r1r2) edge node[above]{$\theta_r$} (r2)
		(r1) edge node[below]{$\product{\id}{\xi_r}$} (y)
		(r2) edge node[right]{$\belong{X}$} (y);
		
	\draw (3.2,0.55) -- (3.45,0.55) -- (3.45,0.8);
			
\end{tikzpicture}
\qquad\qquad
\begin{tikzpicture}[scale=1.5]
		
	\node (r1r2) at (3,1) {\scriptsize{$S$}};
	\node (r1) at (3,0) {\scriptsize{$Y\times Z$}};
	\node (r2) at (5.5,1) {\scriptsize{$E_Y$}};
	\node (y) at (5.5,0) {\scriptsize{$Y\times\pow{Y}$}};
	
	\path[->,font=\scriptsize]
		(r1r2) edge node[left]{$s$} (r1)
		(r1r2) edge node[above]{$\theta_s$} (r2)
		(r1) edge node[below]{$\product{\id}{\xi_s}$} (y)
		(r2) edge node[right]{$\belong{Y}$} (y);
		
	\draw (3.2,0.55) -- (3.45,0.55) -- (3.45,0.8);
			
\end{tikzpicture}
\end{center}
By definition of $\belongsq{X}$ we have the following pullback 
and factorisation:
\begin{center}
\begin{tikzpicture}[scale=1.5]
		
	\node (r1r2) at (3,1) {\scriptsize{$E^3_X$}};
	\node (r1) at (3,0) {\scriptsize{$E_X$}};
	\node (r2) at (5,1) {\scriptsize{$E_{\pow{X}}$}};
	\node (y) at (5,0) {\scriptsize{$\pow{X}$}};
	
	\path[->,font=\scriptsize]
		(r1r2) edge node[left]{$\kappa_1$} (r1)
		(r1r2) edge node[above]{$\kappa_2$} (r2)
		(r1) edge node[below]{$\pi_2\cdot\belong{X}$} (y)
		(r2) edge node[right]{$\pi_1\cdot\belong{\pow{X}}$} (y);
		
	\draw (3.2,0.55) -- (3.45,0.55) -- (3.45,0.8);
			
\end{tikzpicture}
\qquad\qquad
\begin{tikzpicture}[scale=1.5]
		
	\node (r1r2) at (0,0) {\scriptsize{$E^3_X$}};
	\node (xz) at (3.4,0) {\scriptsize{$X\times\pow{\pow{X}}$}};
	\node (r1sqr2) at (1.7,-1) {\scriptsize{$E^2_X$}};
	
	\path[->,font=\scriptsize]
		(r1r2) edge node[above]{$\pair{\pi_1\cdot\belong{X}\cdot\kappa_1}{\pi_2\cdot\belong{\pow{X}}\cdot\kappa_2}$} (xz);
		
	\path[->>,font=\scriptsize]
		(r1r2) edge node[below]{$\rho_X$} (r1sqr2);
		
	\path[>->,font=\scriptsize]
		(r1sqr2) edge node[below]{$\belongsq{X}$} (xz);
			
\end{tikzpicture}
\end{center}
By definition of $\pow{\xi_r}$, we have the following factorisation 
and the pullback:
\begin{center}
\begin{tikzpicture}[scale=1.5]
	
	\node (r1r2) at (0,0) {\scriptsize{$E_Y$}};
	\node (xz) at (3,0) {\scriptsize{$\pow{X}\times\pow{Y}$}};
	\node (r1sqr2) at (1.5,-1) {\scriptsize{$E_{\xi_r}$}};
	
	\path[->,font=\scriptsize]
		(r1r2) edge node[above]{$\product{\xi_r}{\id}\cdot\belong{Y}$} (xz);
		
	\path[->>,font=\scriptsize]
		(r1r2) edge node[below]{$e_{\xi_r}$} (r1sqr2);
		
	\path[>->,font=\scriptsize]
		(r1sqr2) edge node[below]{$m_{\xi_r}$} (xz);
			
\end{tikzpicture}
\qquad\qquad
\begin{tikzpicture}[scale=1.5]
		
	\node (r1r2) at (3,1) {\scriptsize{$E_{\xi_r}$}};
	\node (r1) at (3,0) {\scriptsize{$\pow{X}\times\pow{Y}$}};
	\node (r2) at (5.5,1) {\scriptsize{$E_{\pow{X}}$}};
	\node (y) at (5.5,0) {\scriptsize{$\pow{X}\times\pow{\pow{X}}$}};
	
	\path[->,font=\scriptsize]
		(r1r2) edge node[left]{$m_{\xi_r}$} (r1)
		(r1r2) edge node[above]{$\theta_{\xi_r}$} (r2)
		(r1) edge node[below]{$\product{\id}{\pow{\xi_r}}$} (y)
		(r2) edge node[right]{$\belong{\pow{X}}$} (y);
		
	\draw (3.2,0.55) -- (3.45,0.55) -- (3.45,0.8);
			
\end{tikzpicture}
\end{center}
Finally, by definition of $\rcomp{r}{s}$ we have the following pullback 
and factorisation:
\begin{center}
\begin{tikzpicture}[scale=1.5]
		
	\node (r1r2) at (3,1) {\scriptsize{$R\star S$}};
	\node (r1) at (3,0) {\scriptsize{$R$}};
	\node (r2) at (5,1) {\scriptsize{$S$}};
	\node (y) at (5,0) {\scriptsize{$Y$}};
	
	\path[->,font=\scriptsize]
		(r1r2) edge node[left]{$\mu_1$} (r1)
		(r1r2) edge node[above]{$\mu_2$} (r2)
		(r1) edge node[below]{$\pi_2\cdot r$} (y)
		(r2) edge node[right]{$\pi_1\cdot s$} (y);
		
	\draw (3.2,0.55) -- (3.45,0.55) -- (3.45,0.8);
			
\end{tikzpicture}
\qquad\qquad
\begin{tikzpicture}[scale=1.5]
	
	\node (r1r2) at (0,0) {\scriptsize{$R\star S$}};
	\node (xz) at (3,0) {\scriptsize{$X\times Z$}};
	\node (r1sqr2) at (1.5,-1) {\scriptsize{$\rcomp{R}{S}$}};
	
	\path[->,font=\scriptsize]
		(r1r2) edge node[above]{$\pair{\pi_1\cdot r\cdot\mu_1}{\pi_2\cdot s\cdot\mu_2}$} (xz);
		
	\path[->>,font=\scriptsize]
		(r1r2) edge node[below]{$\rho$} (r1sqr2);
		
	\path[>->,font=\scriptsize]
		(r1sqr2) edge node[below]{$\rcomp{r}{s}$} (xz);
			
\end{tikzpicture}
\end{center}

Now, let us describe the suitable pullback we want to look at. It is defined in several steps. First, form the following two pullbacks:
\begin{center}
\begin{tikzpicture}[scale=1.5]
		
	\node (r1r2) at (3,1) {\scriptsize{$\widehat{S}$}};
	\node (r1) at (3,0) {\scriptsize{$\pow{X}\times Z$}};
	\node (r2) at (5.5,1) {\scriptsize{$E_{\xi_r}$}};
	\node (y) at (5.5,0) {\scriptsize{$\pow{X}\times\pow{Y}$}};
	
	\path[->,font=\scriptsize]
		(r1r2) edge node[left]{$\widehat{s}$} (r1)
		(r1r2) edge node[above]{$\widehat{\theta_s}$} (r2)
		(r1) edge node[below]{$\product{\id}{\xi_s}$} (y)
		(r2) edge node[right]{$m_{\xi_r}$} (y);
		
	\draw (3.2,0.55) -- (3.45,0.55) -- (3.45,0.8);
			
\end{tikzpicture}
\qquad\qquad
\begin{tikzpicture}[scale=1.5]
		
	\node (r1r2) at (3,1) {\scriptsize{$R\square S$}};
	\node (r1) at (3,0) {\scriptsize{$E_X$}};
	\node (r2) at (5.5,1) {\scriptsize{$\widehat{S}$}};
	\node (y) at (5.5,0) {\scriptsize{$\pow{X}$}};
	
	\path[->,font=\scriptsize]
		(r1r2) edge node[left]{$\epsilon_1$} (r1)
		(r1r2) edge node[above]{$\epsilon_2$} (r2)
		(r1) edge node[below]{$\pi_2\cdot\belong{X}$} (y)
		(r2) edge node[right]{$\pi_1\cdot\widehat{s}$} (y);
		
	\draw (3.2,0.55) -- (3.45,0.55) -- (3.45,0.8);
			
\end{tikzpicture}
\end{center}
Our suitable pullback will have the following form:
\begin{center}
\begin{tikzpicture}[scale=1.5]
		
	\node (r1r2) at (2,1) {\scriptsize{$R\square S$}};
	\node (r1) at (2,0) {\scriptsize{$X\times Z$}};
	\node (r2) at (5.3,1) {\scriptsize{$E_3^X$}};
	\node (y) at (5.3,0) {\scriptsize{$X\times\pow{\pow{X}}$}};
	
	\path[->,font=\scriptsize]
		(r1r2) edge node[left]{$\pair{\pi_1\cdot\belong{X}\cdot\epsilon_1}{\pi_2\cdot\widehat{s}\cdot\epsilon_2}$} (r1)
		(r1r2) edge node[above]{$w$} (r2)
		(r1) edge node[below]{$\product{\id}{(\pow{\xi_r}\cdot\xi_s)}$} (y)
		(r2) edge node[right]{$\pair{\pi_1\cdot\belong{X}\cdot\kappa_1}{\pi_2\cdot\belong{\pow{X}}\cdot\kappa_2}$} (y);
		
	\draw (2.2,0.55) -- (2.45,0.55) -- (2.45,0.8);
			
\end{tikzpicture}
\end{center}
for some $w$ we describe now. 
We have the following commutative diagram:
\begin{center}
\begin{tikzpicture}[scale=1.5]
		
	\node (r1r2) at (3,1) {\scriptsize{$R\square S$}};
	\node (r1) at (3,0) {\scriptsize{$E_X$}};
	\node (r2) at (8,1) {\scriptsize{$E_{\pow{X}}$}};
	\node (y) at (8,0) {\scriptsize{$\pow{X}$}};
	
	\path[->,font=\scriptsize]
		(r1r2) edge node[left]{$\epsilon_1$} (r1)
		(r1r2) edge node[above]{$\theta_{\xi_r}\cdot\widehat{\theta_s}\cdot\epsilon_2$} (r2)
		(r1) edge node[below]{$\pi_2\cdot\belong{X}$} (y)
		(r2) edge node[right]{$\pi_1\cdot\belong{\pow{X}}$} (y);
			
\end{tikzpicture}
\end{center}
Indeed,
\begin{center}
\begin{tabular}{rclcr}
    $\pi_1\cdot\belong{\pow{X}}\cdot\theta_{\xi_r}\cdot\widehat{\theta_s}\cdot\epsilon_2$ & $=$ & $\pi_1\cdot m_{\xi_r}\cdot\widehat{\theta_s}\cdot\epsilon_2$
    & & \hfill (definition of $\pow{\xi_r}$)\\
    & $=$ & $\pi_1\cdot \widehat{s}\cdot\epsilon_2$
    & & \hfill (definition of $\widehat{S}$)\\
    & $=$ & $\pi_2\cdot\belong{X}\cdot\epsilon_1$
    & & \hfill (definition of $R\square S$)
\end{tabular}
\end{center}
So by the universal property of $E_X^3$, there is a unique morphism 
$\map{w}{R\square S}{E_X^3}$ such that 
\[\kappa_1\cdot w = \epsilon_1 
	\quad\text{and}\quad 
	\kappa_2\cdot w = \theta_{\xi_r}\cdot\widehat{\theta_s}\cdot\epsilon_2.\]

Let us prove that the suitable pullback is indeed a pullback. First it is a commutative diagram:
\begin{center}
\begin{tabular}{rclcr}
    $\pi_1\cdot\belong{X}\cdot\kappa_1\cdot w$ & $=$ & $\pi_1\cdot\belong{X}\cdot\epsilon_1$
    & & \hfill (definition of $w$)\\
    $\pi_2\cdot\belong{\pow{X}}\cdot\kappa_2\cdot w$ & $=$ & $\pi_2\cdot\belong{\pow{X}}\cdot \theta_{\xi_r}\cdot\widehat{\theta_s}
    \cdot\epsilon_2$
    & & \hfill (definition of $w$)\\
    & $=$ & $\pow{\xi_r}\cdot\pi_2\cdot m_{\xi_r}\cdot\widehat{\theta_s}\cdot\epsilon_2$
    & & \hfill (definition of $\pow{\xi_r}$)\\
    & $=$ & $\pow{\xi_r}\cdot\xi_s\cdot\pi_2\cdot\widehat{s}\cdot\epsilon_2$
    & & \hfill (definition of $\widehat{S}$)
\end{tabular}
\end{center}
Now, assume given another commutative diagram of the form:
\begin{center}
\begin{tikzpicture}[scale=1.5]
		
	\node (r1r2) at (3,1) {\scriptsize{$W$}};
	\node (r1) at (3,0) {\scriptsize{$X\times Z$}};
	\node (r2) at (7,1) {\scriptsize{$E_3^X$}};
	\node (y) at (7,0) {\scriptsize{$X\times\pow{\pow{X}}$}};
	
	\path[->,font=\scriptsize]
		(r1r2) edge node[left]{$\psi$} (r1)
		(r1r2) edge node[above]{$\phi$} (r2)
		(r1) edge node[below]{$\product{\id}{(\pow{\xi_r}\cdot\xi_s)}$} (y)
		(r2) edge node[right]{$\pair{\pi_1\cdot\belong{X}\cdot\kappa_1}{\pi_2\cdot\belong{\pow{X}}\cdot\kappa_2}$} (y);
			
\end{tikzpicture}
\end{center}
We construct a morphism $\map{\gamma}{W}{R\square S}$ 
using three universal properties of pullbacks as follows. 
First we have the following commutative diagram:
\begin{center}
\begin{tikzpicture}[scale=1.5]
		
	\node (r1r2) at (3,1) {\scriptsize{$W$}};
	\node (r1) at (3,0) {\scriptsize{$\pow{X}\times\pow{Y}$}};
	\node (r2) at (6.5,1) {\scriptsize{$E_{\pow{X}}$}};
	\node (y) at (6.5,0) {\scriptsize{$\pow{X}\times\pow{\pow{X}}$}};
	
	\path[->,font=\scriptsize]
		(r1r2) edge node[left]{$\pair{\pi_2\cdot\belong{X}\cdot\kappa_1\cdot\phi}{\xi_s\cdot\pi_2\cdot\psi}$} (r1)
		(r1r2) edge node[above]{$\kappa_2\cdot\phi$} (r2)
		(r1) edge node[below]{$\product{\id}{\pow{\xi_r}}$} (y)
		(r2) edge node[right]{$\belong{\pow{X}}$} (y);
			
\end{tikzpicture}
\end{center}
Indeed,
\begin{center}
\begin{tabular}{rclcr}
    $\pi_1\cdot\belong{\pow{X}}\cdot\kappa_2\cdot\phi$ & $=$ & $\pi_2\cdot \belong{X}\cdot\kappa_1\cdot\phi$
    & & \hfill (definition of $E^3_X$)\\
    $\pi_2\cdot\belong{\pow{X}}\cdot\kappa_2\cdot\phi$ & $=$ & $\pow{\xi_r}\cdot\xi_s\cdot\pi_2\cdot\psi$
    & & \hfill (assumption on $W$)
\end{tabular}
\end{center}
So by the universal property of $E_{\xi_r}$, there is a unique morphism 
$\map{\alpha}{W}{E_{\xi_r}}$ such that 
\[m_{\xi_r}\cdot\alpha = \pair{\pi_2\cdot\belong{X}\cdot\kappa_1\cdot\phi}{\xi_s\cdot\pi_2\cdot\psi}
\]
and
\[\theta_{\xi_r}\cdot\alpha = \kappa_2\cdot\phi.\]
Secondly, we have the following commutative diagram, by definition of $\alpha$:
\begin{center}
\begin{tikzpicture}[scale=1.5]
		
	\node (r1r2) at (3,1) {\scriptsize{$W$}};
	\node (r1) at (3,0) {\scriptsize{$\pow{X}\times Z$}};
	\node (r2) at (7,1) {\scriptsize{$E_{\xi_r}$}};
	\node (y) at (7,0) {\scriptsize{$\pow{X}\times\pow{Y}$}};
	
	\path[->,font=\scriptsize]
		(r1r2) edge node[left]{$\pair{\pi_2\cdot\belong{X}\cdot\kappa_1\cdot\phi}{\pi_2\cdot\psi}$} (r1)
		(r1r2) edge node[above]{$\alpha$} (r2)
		(r1) edge node[below]{$\product{\id}{\xi_s}$} (y)
		(r2) edge node[right]{$m_{\xi_r}$} (y);
			
\end{tikzpicture}
\end{center}
So by the universal property of $\widehat{S}$, there is a unique morphism 
$\map{\beta}{W}{\widehat{S}}$ such that 
\[\widehat{s}\cdot\beta = \pair{\pi_2\cdot\belong{X}\cdot\kappa_1\cdot\phi}{\pi_2\cdot\psi}\] 
and
\[\widehat{\theta_s}\cdot\beta = \alpha.\]
Finally, we have the following commutative diagram, by definition of 
$\beta$:
\begin{center}
\begin{tikzpicture}[scale=1.5]
		
	\node (r1r2) at (3,1) {\scriptsize{$W$}};
	\node (r1) at (3,0) {\scriptsize{$E_X$}};
	\node (r2) at (8,1) {\scriptsize{$\widehat{S}$}};
	\node (y) at (8,0) {\scriptsize{$\pow{X}$}};
	
	\path[->,font=\scriptsize]
		(r1r2) edge node[left]{$\kappa_1\cdot\phi$} (r1)
		(r1r2) edge node[above]{$\beta$} (r2)
		(r1) edge node[below]{$\pi_2\cdot\belong{X}$} (y)
		(r2) edge node[right]{$\pi_1\cdot\widehat{s}$} (y);
			
\end{tikzpicture}
\end{center}
So by the universal property of $R\square S$, there is a unique morphism 
$\map{\gamma}{W}{R\square S}$ such that 
\[\epsilon_1\cdot\gamma = \kappa_1\cdot\phi 
\quad\text{and}\quad 
\epsilon_2\cdot\gamma = \beta.\]

Let us prove that $\gamma$ is the unique morphism from $W$ to $R\square S$ such that 
\[w\cdot\gamma = \phi 
\quad\text{and}\quad 
\pair{\pi_1\cdot\belong{X}\cdot\epsilon_1}
	{\pi_2\cdot\widehat{s}\cdot\epsilon_2}\cdot\gamma = \psi.\]
First, it satisfies those conditions. For the first one, by the unicity of the 
pullback property of $E^3_X$, it is enough to prove the following
\begin{center}
\begin{tabular}{rclcl}
    $\kappa_1\cdot w\cdot\gamma$ & $=$ & $\epsilon_1\cdot\gamma$ & ~~~~~ & (definition of $w$)\\
    & $=$ & $\kappa_1\cdot\phi$ & ~~~~~ & (definition of $\gamma$)\\
    $\kappa_2\cdot w\cdot\gamma$ & $=$ & $\theta_{\xi_r}\cdot\widehat{\theta_s}\cdot\epsilon_2\cdot\gamma$ & ~~~~~ & (definition of $w$)\\
    & $=$ & $\theta_{\xi_r}\cdot\widehat{\theta_s}\cdot\beta$ & ~~~~~ & (definition of $\gamma$)\\
    & $=$ & $\theta_{\xi_r}\cdot\alpha$ & ~~~~~ & (definition of $\beta$)\\
    & $=$ & $\kappa_2\cdot\phi$ & ~~~~~ & (definition of $\alpha$)
\end{tabular}
\end{center}
For the second one:
\begin{center}
\begin{tabular}{rclcr}
    $\pi_1\cdot\belong{X}\cdot\epsilon_1\cdot\gamma$ & $=$ & $\pi_1\cdot\belong{X}\cdot\kappa_1\cdot\phi$
    & & \hfill (definition of $\gamma$)\\
    & $=$ & $\pi_1\cdot\psi$
    & & \hfill (assumption on $W$)\\
    $\pi_2\cdot\widehat{s}\cdot\epsilon_2\cdot\gamma$ & $=$ & $\pi_2\cdot\widehat{s}\cdot\beta$
    & & \hfill (definition of $\gamma$)\\
    & $=$ & $\pi_2\cdot\psi$
    & & \hfill (definition of $\beta$)
\end{tabular}
\end{center}

Now assume that there is another $\gamma'$ from $W$ to $R\square S$ such that 
\[w\cdot\gamma' = \phi 
\quad\text{and}\quad
\pair{\pi_1\cdot\belong{X}\cdot\epsilon_1}
	{\pi_2\cdot\widehat{s}\cdot\epsilon_2}\cdot \gamma' = \psi.\]
By the unicity properties of $\alpha$, $\beta$ and $\gamma$, 
it is enough to prove the following five equations:
\begin{center}
\begin{tabular}{rclcr}
    $\kappa_1\cdot\phi$ & $=$ & $\kappa_1\cdot w\cdot\gamma'$
    & & \hfill (assumption on $\gamma'$)\\
    & $=$ & $\epsilon_1\cdot\gamma'$
    & &\hfill (definition of $w$)\\
    $\pi_1\cdot\widehat{s}\cdot\epsilon_2\cdot\gamma'$ & $=$ & $\pi_2\cdot\belong{X}\cdot\epsilon_1\cdot\gamma'$
    & & \hfill (definition of $R\square S$)\\
    & $=$ & $\pi_2\cdot\belong{X}\cdot\kappa_1\cdot\phi$
    & & \hfill (assumption on $\gamma'$)\\
    $\pi_2\cdot\widehat{s}\cdot\epsilon_2\cdot\gamma'$ & $=$ & $\pi_2\cdot\psi$
    & & \hfill (assumption on $\gamma'$)\\
    $\kappa_2\cdot\phi$ & $=$ & $\kappa_2\cdot w\cdot\gamma'$
    & &  (assumption on $\gamma'$)\\
    & $=$ & $\theta_{\xi_r}\cdot\widehat{\theta_s}\cdot\epsilon_2\cdot\gamma'$
    & & (definition of $w$)\\
    $m_{\xi_r}\cdot\widehat{\theta_s}\cdot\epsilon_2\cdot\gamma'$ & $=$ & $\product{\id}{\xi_s}\cdot\widehat{s}\cdot\epsilon_2\cdot\gamma'$
    & & \hfill (definition of $\widehat{S}$)\\
    & $=$ & $\pair{pi_1\cdot\widehat{s}\cdot\epsilon_2\cdot\gamma'}{\xi_s\cdot\pi_2\cdot\psi}$
    & & \hfill (assumption on $\gamma'$)\\
    & $=$ & $\pair{\pi_2\cdot\belong{X}\cdot\epsilon_1\cdot\gamma'}{\xi_s\cdot\pi_2\cdot\psi}$
    & & \hfill (definition of $R\square S$)\\
    & $=$ & $\pair{\pi_2\cdot\belong{X}\cdot\kappa_1\cdot w\cdot\gamma'}{\xi_s\cdot\pi_2\cdot\psi}$
    & & \hfill (definition of $w$)\\
    & $=$ & $\pair{\pi_2\cdot\belong{X}\cdot\kappa_1\cdot \phi}{\xi_s\cdot\pi_2\cdot\psi}$
    & & \hfill (assumption on $\gamma'$)
\end{tabular}
\end{center}
from which we deduce that 
$\alpha = \widehat{\theta_s}\cdot\epsilon_2\cdot\gamma'$, 
then $\beta = \epsilon_2\cdot\gamma'$, 
and finally $\gamma = \gamma'$.

So we have our suitable pullback:
\begin{center}
\begin{tikzpicture}[scale=1.5]
		
	\node (r1r2) at (3,1) {\scriptsize{$R\square S$}};
	\node (r1) at (3,0) {\scriptsize{$X\times Z$}};
	\node (r2) at (5.3,1) {\scriptsize{$E_3^X$}};
	\node (y) at (5.3,0) {\scriptsize{$X\times\pow{\pow{X}}$}};
	
	\path[->,font=\scriptsize]
		(r1r2) edge node[left]{$\pair{\pi_1\cdot\belong{X}\cdot\epsilon_1}{\pi_2\cdot\widehat{s}\cdot\epsilon_2}$} (r1)
		(r1r2) edge node[above]{$w$} (r2)
		(r1) edge node[below]{$\product{\id}{(\pow{\xi_r}\cdot\xi_s)}$} (y)
		(r2) edge node[right]{$\pair{\pi_1\cdot\belong{X}\cdot\kappa_1}{\pi_2\cdot\belong{\pow{X}}\cdot\kappa_2}$} (y);
		
	\draw (3.2,0.55) -- (3.45,0.55) -- (3.45,0.8);
			
\end{tikzpicture}
\end{center}
To conclude with the preservation of the image by pullback, 
we have to prove that we have the correct 
(epi, mono)-factorisations, that is:
\begin{itemize}
	\item $\monop{\pair{\pi_1\cdot\belong{X}\cdot\kappa_1}{\pi_2\cdot\belong{\pow{X}}\cdot\kappa_2}} \equiv~\belongsq{X}$: this is the case by definition of $\belongsq{X}$.
	\item $\monop{\pair{\pi_1\cdot\belong{X}\cdot\epsilon_1}{\pi_2\cdot\widehat{s}\cdot\epsilon_2}} \equiv \rcomp{r}{s}$: 
	this part is much more complicated. We know, by construction, 
	that $\rcomp{r}{s} \equiv \monop{\pair{\pi_1\cdot r\cdot\mu_1}{\pi_2\cdot s\cdot\mu_2}}$, 
	so we need to compare those two morphisms. 
	We start by constructing a morphism 
	$\map{v}{R\star S}{R\square S}$, 
	by using two pullbacks properties as follows.
	
	First we have the following commutative diagram:
\begin{center}
\begin{tikzpicture}[scale=1.5]
		
	\node (r1r2) at (3,1) {\scriptsize{$R\star S$}};
	\node (r1) at (3,0) {\scriptsize{$\pow{X}\times Z$}};
	\node (r2) at (6.5,1) {\scriptsize{$E_{\xi_r}$}};
	\node (y) at (6.5,0) {\scriptsize{$\pow{X}\times\pow{Y}$}};
	
	\path[->,font=\scriptsize]
		(r1r2) edge node[left]{$\pair{\xi_r\cdot\pi_2\cdot r\cdot\mu_1}{\pi_2\cdot s\cdot\mu_2}$} (r1)
		(r1r2) edge node[above]{$e_{\xi_r}\cdot\theta_s\cdot\mu_2$} (r2)
		(r1) edge node[below]{$\product{\id}{\xi_s}$} (y)
		(r2) edge node[right]{$m_{\xi_r}$} (y);
			
\end{tikzpicture}
\end{center}
Indeed,
\begin{center}
\begin{tabular}{rclcr}
    $m_{\xi_r}\cdot e_{\xi_r}\cdot\theta_s\cdot\mu_2$ & $=$ & $\product{\xi_r}{\id}\cdot\belong{Y}\cdot\theta_s\cdot\mu_2$
    & & \hfill (definition of $E_{\xi_r}$)\\
    & $=$ & $\product{\xi_r}{\xi_s}\cdot s\cdot\mu_2$
    & & \hfill (definition of $\xi_s$)\\
    & $=$ & $\product{\id}{\xi_s}\cdot\pair{\xi_r\cdot\pi_1\cdot s\cdot\mu_2}{\pi_2\cdot s \cdot\mu_2}$
    & & \hfill (computation on products)\\
    & $=$ & $\product{\id}{\xi_s}\cdot\pair{\xi_r\cdot\pi_2\cdot r\cdot\mu_1}{\pi_2\cdot s \cdot\mu_2}$
    & & \hfill (definition of $R\star S$)
\end{tabular}
\end{center}
So by the universal property of $\widehat{S}$, there is a unique morphism 
$\map{u}{R\star S}{\widehat{S}}$ such that 
\[\widehat{\theta_s}\cdot u = e_{\xi_r}\cdot\theta_s\cdot\mu_2\] 
and
\[\widehat{s}\cdot u = \pair{\xi_r\cdot\pi_2\cdot r\cdot\mu_1}
	{\pi_2\cdot s\cdot\mu_2}.\]
	
	Next we have the following commutative diagram:
\begin{center}
\begin{tikzpicture}[scale=1.5]
		
	\node (r1r2) at (3,1) {\scriptsize{$R\star S$}};
	\node (r1) at (3,0) {\scriptsize{$E_X$}};
	\node (r2) at (8,1) {\scriptsize{$\widehat{S}$}};
	\node (y) at (8,0) {\scriptsize{$\pow{X}$}};
	
	\path[->,font=\scriptsize]
		(r1r2) edge node[left]{$\theta_r\cdot\mu_1$} (r1)
		(r1r2) edge node[above]{$u$} (r2)
		(r1) edge node[below]{$\pi_2\cdot\belong{X}$} (y)
		(r2) edge node[right]{$\pi_1\cdot\widehat{s}$} (y);
			
\end{tikzpicture}
\end{center}
Indeed,
\begin{center}
\begin{tabular}{rclcr}
    $\pi_2\cdot\belong{X}\cdot\theta_r\cdot\mu_1$ & $=$ & $\xi_r\cdot\pi_2\cdot r\cdot\mu_1$
    & & \hfill (definition of $\xi_r$)\\
    & $=$ & $\xi_r\cdot\pi_1\cdot s\cdot\mu_2$
    & & \hfill (definition of $R\star S$)\\
    & $=$ & $\xi_r\cdot\pi_1\cdot \belong{Y}\cdot\theta_s\cdot\mu_2$
    & & \hfill (definition of $\xi_s$)\\
    & $=$ & $\pi_1\cdot m_{\xi_r}\cdot e_{\xi_r}\cdot\theta_s\cdot\mu_2$
    & & \hfill (definition of $E_{\xi_r}$)\\
    & $=$ & $\pi_1\cdot m_{\xi_r}\cdot\widehat{\theta_s}\cdot u$
    & & \hfill (definition of $u$)\\
    & $=$ & $\pi_1\cdot \widehat{s}\cdot u$
    & & \hfill (definition of $\widehat{S}$)
\end{tabular}
\end{center}
So by the universal property of $R\square S$, there is a unique 
morphism $\map{v}{R\star S}{R\square S}$ such that 
\[\epsilon_1\cdot v = \theta_r\cdot\mu_1 \quad\text{and}\quad 
	\epsilon_2\cdot v = u.\]

Now, we can compare the two morphisms and their (epi, mono)-factorisations, since we have the following commutative diagram:
\begin{center}
\begin{tikzpicture}[scale=1.5]
		
	\node (ex) at (3,1) {\scriptsize{$R\star S$}};
	\node (exp) at (3,0) {\scriptsize{$R\square S$}};
	\node (ef) at (5.5,1) {\scriptsize{$\rcomp{R}{S}$}};
	\node (egf) at (5.5, 0) {\scriptsize{$T$}};
	\node (ypx) at (8,1) {\scriptsize{$X\times Z$}};
	\node (zpx) at (8,0) {\scriptsize{$X\times Z$}};
	
	\path[->,font=\scriptsize]
		(ex) edge node[left]{$v$} (exp)
		(ypx) edge node[right]{id} (zpx);
	
	\path[->>,font=\scriptsize]
		(ex) edge node[above]{$\rho$} (ef)
		(exp) edge (egf);
		
	\path[>->,font=\scriptsize]
		(ef) edge node[above]{$\rcomp{r}{s}$} (ypx)
		(egf) edge (zpx);
		
	\path[->,dotted,font=\scriptsize]
		(ef) edge (egf);
		
	\path[->,bend left = 30,font=\scriptsize]
		(ex) edge node[above]{$\pair{\pi_1\cdot r\cdot\mu_1}{\pi_2\cdot s\cdot\mu_2}$} (ypx);
		
	\path[->,bend right = 30,font=\scriptsize]
		(exp) edge node[below]{$\pair{\pi_1\cdot\belong{X}\cdot\epsilon_1}{\pi_2\cdot\widehat{s}\cdot\epsilon_2}$} (zpx);
			
\end{tikzpicture}
\end{center}
Indeed,
\begin{center}
\begin{tabular}{rclr}
    $\pi_1\cdot\belong{X}\cdot\epsilon_1\cdot v$ & $=$ & $\pi_1\cdot\belong{X}\cdot\theta_r\cdot\mu_1$ &   (definition of $v$)\\
    & $=$ & $\pi_1\cdot r\cdot\mu_1$ &   (definition of $\xi_r$)\\
    $\pi_2\cdot\widehat{s}\cdot\epsilon_2\cdot v$ & $=$ & $\pi_2\cdot\widehat{s}\cdot u$ &  (definition of $v$)\\
    & $=$ & $\pi_2\cdot s\cdot\mu_2$ &   (definition of $u$)
\end{tabular}
\end{center}
So by functoriality of the (epi, mono)-factorisation, we have the dotted morphism as above. To conclude, we need to prove that this is an iso. The right square tells us this is a mono. If we can prove that $v$ is an epi, then this dotted morphism would also be an epi, and since since we are in a topos, this would be an iso.

To prove that $v$ is an epi, we will use the fact that epis are closed under pullback in a topos. To this end, let us prove that the following square is a pullback:
\begin{center}
\begin{tikzpicture}[scale=1.5]
		
	\node (r1r2) at (3,1) {\scriptsize{$R\star S$}};
	\node (r1) at (3,0) {\scriptsize{$R\square S$}};
	\node (r2) at (8,1) {\scriptsize{$E_Y$}};
	\node (y) at (8,0) {\scriptsize{$E_{\xi_r}$}};
	
	\path[->,font=\scriptsize]
		(r1r2) edge node[left]{$v$} (r1)
		(r1r2) edge node[above]{$\theta_s\cdot\mu_2$} (r2)
		(r1) edge node[below]{$\widehat{\theta_s}\cdot\epsilon_2$} (y)
		(r2) edge node[right]{$e_{\xi_r}$} (y);
			
\end{tikzpicture}
\end{center}

First, it is a commutative square:
\begin{center}
\begin{tabular}{rclr}
    $\widehat{\theta_s}\cdot\epsilon_2\cdot v$ & $=$ & $\widehat{\theta_s}\cdot u$ & (definition of $v$)\\
    & $=$ & $e_{\xi_r}\cdot\theta_s\cdot\mu_2$ & (definition of $u$)
\end{tabular}
\end{center}

Now assume given another commutative diagram:
\begin{center}
\begin{tikzpicture}[scale=1.5]
		
	\node (r1r2) at (3,1) {\scriptsize{$W$}};
	\node (r1) at (3,0) {\scriptsize{$R\square S$}};
	\node (r2) at (8,1) {\scriptsize{$E_Y$}};
	\node (y) at (8,0) {\scriptsize{$E_{\xi_r}$}};
	
	\path[->,font=\scriptsize]
		(r1r2) edge node[left]{$\psi$} (r1)
		(r1r2) edge node[above]{$\phi$} (r2)
		(r1) edge node[below]{$\widehat{\theta_s}\cdot\epsilon_2$} (y)
		(r2) edge node[right]{$e_{\xi_r}$} (y);
			
\end{tikzpicture}
\end{center}
We want to construct a morphism $\map{\gamma}{W}{R\star S}$. 
This is done by using three pullback properties as follows.
First we have the following commutative diagram:
\begin{center}
\begin{tikzpicture}[scale=1.5]
		
	\node (r1r2) at (3,1) {\scriptsize{$W$}};
	\node (r1) at (3,0) {\scriptsize{$X\times Y$}};
	\node (r2) at (6.5,1) {\scriptsize{$E_X$}};
	\node (y) at (6.5,0) {\scriptsize{$X\times\pow{X}$}};
	
	\path[->,font=\scriptsize]
		(r1r2) edge node[left]{$\pair{\pi_1\cdot\belong{X}\cdot\epsilon_1\cdot\psi}{\pi_1\cdot\belong{Y}\cdot\phi}$} (r1)
		(r1r2) edge node[above]{$\epsilon_1\cdot\psi$} (r2)
		(r1) edge node[below]{$\product{\id}{\xi_r}$} (y)
		(r2) edge node[right]{$\in_{X}$} (y);
			
\end{tikzpicture}
\end{center}
Indeed, 
\begin{center}
\begin{tabular}{rclcr}
    $\pi_2\cdot\belong{X}\cdot\epsilon_1\cdot\psi$ & $=$ & $\pi_1\cdot\widehat{s}\cdot\epsilon_2\cdot\psi$
    & & \hfill (definition of $R\square S$)\\
    & $=$ & $\pi_1\cdot m_{\xi_r}\cdot\widehat{\theta_s}\cdot\epsilon_2\cdot\phi$
    & & \hfill (definition of $\widehat{S}$)\\
    & $=$ & $\pi_1\cdot m_{\xi_r}\cdot e_{\xi_r}\cdot\phi$
    & & \hfill (assumption on $W$)\\
    & $=$ & $\xi_r\cdot\pi_1\cdot\belong{Y}\cdot\phi$
    & & \hfill (definition of $E_{\xi_r}$)
\end{tabular}
\end{center}
So by the universal property of $R$, there is a unique morphism 
$\map{\alpha}{W}{R}$ such that 
\[r\cdot\alpha = \pair{\pi_1\cdot\belong{X}\cdot\epsilon_1\cdot\psi}
	{\pi_1\cdot\belong{Y}\cdot\phi}\]
and
\[\theta_r\cdot\alpha = \epsilon_1\cdot\psi.\]
Next we have the following commutative diagram:
\begin{center}
\begin{tikzpicture}[scale=1.5]
		
	\node (r1r2) at (3,1) {\scriptsize{$W$}};
	\node (r1) at (3,0) {\scriptsize{$Y\times Z$}};
	\node (r2) at (6.5,1) {\scriptsize{$E_Y$}};
	\node (y) at (6.5,0) {\scriptsize{$Y\times\pow{Y}$}};
	
	\path[->,font=\scriptsize]
		(r1r2) edge node[left]{$\pair{\pi_1\cdot\belong{Y}\cdot\phi}{\pi_2\cdot\widehat{s}\cdot\epsilon_2\cdot\psi}$} (r1)
		(r1r2) edge node[above]{$\phi$} (r2)
		(r1) edge node[below]{$\product{\id}{\xi_s}$} (y)
		(r2) edge node[right]{$\belong{Y}$} (y);
			
\end{tikzpicture}
\end{center}
Indeed, 
\begin{center}
\begin{tabular}{rclcr}
    $\xi_s\cdot\pi_2\cdot\widehat{s}\cdot\epsilon_2\cdot\psi$ & $=$ & $\pi_2\cdot m_{\xi_r}\cdot\widehat{\theta_s}\cdot\epsilon_2\cdot\psi$
    & & \hfill (definition of $\widehat{S}$)\\
    & $=$ & $\pi_2\cdot m_{\xi_r}\cdot e_{\xi_r}\cdot\phi$
    & & \hfill (assumption on $W$)\\
    & $=$ & $\pi_2\cdot \belong{Y}\cdot\phi$
    & & \hfill (definition of $E_{\xi_r}$)
\end{tabular}
\end{center}
So by the universal property of $S$, there is a unique morphism 
$\map{\beta}{W}{S}$ such that 
\[s\cdot\beta = \pair{\pi_1\cdot\belong{Y}\cdot\phi}{\pi_2\cdot\widehat{s}\cdot\epsilon_2\cdot\psi}\]
and
\[\theta_s\cdot\beta = \phi.\]
Finally we have the following commutative diagram:
\begin{center}
\begin{tikzpicture}[scale=1.5]
		
	\node (r1r2) at (3,1) {\scriptsize{$W$}};
	\node (r1) at (3,0) {\scriptsize{$R$}};
	\node (r2) at (8,1) {\scriptsize{$S$}};
	\node (y) at (8,0) {\scriptsize{$Y$}};
	
	\path[->,font=\scriptsize]
		(r1r2) edge node[left]{$\alpha$} (r1)
		(r1r2) edge node[above]{$\beta$} (r2)
		(r1) edge node[below]{$\pi_2\cdot r$} (y)
		(r2) edge node[right]{$\pi_1\cdot s$} (y);
			
\end{tikzpicture}
\end{center}
Indeed, 
\begin{center}
\begin{tabular}{rclcl}
    $\pi_2\cdot r\cdot\alpha$ & $=$ & $\pi_1\cdot\belong{Y}\cdot\phi$ & ~~~~~ & (definition of $\alpha$)\\
    & $=$ & $\pi_1\cdot s\cdot\beta$ & ~~~~~ & (definition of $\beta$)
\end{tabular}
\end{center}
So by the universal property of $R\star S$, there is a unique morphism 
$\map{\gamma}{W}{R\star S}$ such that 
\[\mu_1\cdot\gamma = \alpha \quad\text{and}\quad 
	\mu_2\cdot\gamma = \beta.\]

Let us prove that $\gamma$ is the unique morphism from 
$W$ to $R\star S$ such that 
\[v\cdot\gamma = \psi \quad\text{and}\quad 
	\theta_s\cdot\mu_2\cdot\gamma = \phi.\]
First, it satisfies those properties. For the first one, by unicity in the 
pullback property of $R\square S$ and the fact that $\widehat{s}$ 
is a mono, it is enough to prove:
\begin{center}
\begin{tabular}{rclcr}
    $\epsilon_1\cdot v\cdot\gamma$ & $=$ & $\theta_r\cdot \mu_1\cdot \gamma$
    & & \hfill (definition of $v$)\\
    & $=$ & $\theta_r\cdot\alpha$
    & & \hfill (definition of $\gamma$)\\
    & $=$ & $\epsilon_1\cdot\psi$
    & & \hfill (definition of $\alpha$)\\
    $\widehat{s}\cdot\epsilon_2\cdot v\cdot\gamma$ & $=$ & $\widehat{s}\cdot u \cdot \gamma$
    & & \hfill (definition of $v$)\\
    & $=$ & $\pair{\xi_r\cdot\pi_2\cdot r\cdot\mu_1\cdot\gamma}{\pi_2\cdot s\cdot\mu_2\cdot\gamma}$
    & & \hfill (definition of $u$)\\
    & $=$ & $\pair{\xi_r\cdot\pi_2\cdot r\cdot\alpha}{\pi_2\cdot s\cdot\beta}$
    & & \hfill (definition of $\gamma$)\\
    & $=$ & $\pair{\xi_r\cdot\pi_1\cdot\belong{Y}\cdot\phi}{\pi_2\cdot s\cdot\beta}$
    & & \hfill (definition of $\alpha$)\\
    & $=$ & $\pair{\xi_r\cdot\pi_1\cdot\belong{Y}\cdot\phi}{\pi_2\cdot\widehat{s}\cdot\epsilon_2\cdot\psi}$
    & & \hfill (definition of $\beta$)\\
    & $=$ & $\pair{\pi_1\cdot m_{\xi_r}\cdot e_{\xi_r}\cdot\phi}{\pi_2\cdot \widehat{s}\cdot\epsilon_2\cdot\psi}$
    & & \hfill (definition of $E_{\xi_r}$)\\
    & $=$ & $\pair{\pi_1\cdot m_{\xi_r}\cdot \widehat{\theta_s}\cdot\epsilon_2\cdot\psi}{\pi_2\cdot \widehat{s}\cdot\epsilon_2\cdot\psi}$
    & & \hfill (assumption on $W$)\\
    & $=$ & ${\pi_1\cdot \widehat{s}\cdot\epsilon_2\cdot\psi}{\pi_2\cdot \widehat{s}\cdot\epsilon_2\cdot\psi}$
    & & \hfill (definition of $\widehat{S}$)\\
    & $=$ & $\widehat{s}\cdot\epsilon_2\cdot\psi$
    & & \hfill (easy)
\end{tabular}
\end{center}
For the second one,
\[\theta_s\cdot\mu_2\cdot\gamma = \theta_s\cdot\beta = \phi.\]

Now assume that there is another $\gamma'$ from $W$ to $R\star S$ 
such that 
\[v\cdot\gamma' = \psi \quad\text{and}\quad
	 \theta_s\cdot\mu_2\cdot\gamma' = \phi.\]
Using the unicity of $\gamma$ it is enough to prove that 
$\mu_1\cdot\gamma' =\alpha$ and $\mu_2\cdot\gamma' = \beta$. 
For the first one, by unicity of $\alpha$ it enough to prove the following:
\begin{center}
\begin{tabular}{rclcr}
    $\theta_r\cdot\mu_1\cdot\gamma'$ & $=$ & $\epsilon_1\cdot v \cdot\gamma'$
    & & \hfill (definition of $v$)\\
    & $=$ & $\epsilon_1\cdot\psi$
    & & \hfill (assumption on $\gamma'$)\\
    $\pi_1\cdot r \cdot\mu_1\cdot\gamma'$ & $=$ & $\pi_1\cdot\belong{X}\cdot\theta_r\cdot\mu_1\cdot\gamma'$
    & & \hfill (definition of $\xi_r$)\\
    & $=$ & $\pi_1\cdot\belong{X}\cdot\epsilon_1\cdot\psi$
    & & \hfill (similar to the previous case)\\
    $\pi_2\cdot r \cdot\mu_1\cdot\gamma'$ & $=$ & $\pi_1\cdot s \cdot\mu_2\cdot\gamma'$
    & & \hfill (definition of $R\star S$)\\
    & $=$ & $\pi_1\cdot \belong{Y}\cdot\theta_s \cdot\mu_2\cdot\gamma'$
    & & \hfill (definition of $\xi_s$)\\
    & $=$ & $\pi_1\cdot \belong{Y}\cdot\phi$
    & & \hfill (assumption on $\gamma'$)
\end{tabular}
\end{center}

For the second one, by unicity of $\beta$, it is enough to prove 
the following:
\begin{center}
\begin{align*} 
    \theta_s\cdot\mu_2\cdot\gamma' & \mkern9mu=\mkern9mu  \phi
     & \hfill \text{(assumption on $\gamma'$)}&\\
    \pi_1\cdot s\cdot\mu_2\cdot\gamma'& \mkern9mu=\mkern9mu  \pi_1\cdot\belong{Y}\cdot\theta_s\cdot\mu_2\cdot\gamma'
     & \hfill \text{(definition of $\xi_s$)}&\\
    & \mkern9mu=\mkern9mu \pi_1\cdot\belong{Y}\cdot\phi
    & \hfill \text{(assumption on $\gamma'$)}&\\
    \pi_2\cdot s\cdot\mu_2\cdot\gamma'& \mkern9mu=\mkern9mu  \pi_2\cdot\widehat{s}\cdot u\cdot\gamma'
     & \hfill \text{(definition of $u$)}&\\
    & \mkern9mu=\mkern9mu  \pi_2\cdot\widehat{s}\cdot \epsilon_2\cdot v\cdot\gamma'
    &  \hfill \text{(definition of $v$)}&\\
    & \mkern9mu=\mkern9mu  \pi_2\cdot\widehat{s}\cdot \epsilon_2\cdot \psi
    &  \hfill \text{(assumption on $\gamma'$)}&\qedhere
\end{align*}
\end{center}\qedhere
\end{itemize}
\end{proof}
\noindent 
In addition, given a morphism $\map{f}{X}{Y}$ of the topos, we have a 
corresponding morphism in the Kleisli category with $\eta_Y\cdot f$. 
Through $\xi$, this morphism corresponds to the right adjoint $\pair{f}{\id}$.
Using Theorem~\ref{th:maps}, we obtain that this functor from the topos to 
the Kleisli category is in reality an 
embedding. In particular, this means:
\begin{lem}
For all $X$, $\eta_X$ is a mono.
\end{lem}
\noindent 
In this explanation, we can get rid of the ``opposite'', since an allegory is 
self-dual.

\subsection{Naturality and Coherence Axioms}

We are now all set to prove the first part of Theorem~\ref{thm:pow-monad}:
Both naturalities are also easy or consequence of 
Proposition~\ref{prop:composition}:
\begin{lem}
\label{lem:pow-mon-naturality}
$\eta$ and $\mu$ are natural. Furthermore, we have $\mu_X\cdot\eta_{\pow{X}} = \id_{\pow{X}}$,
	$\mu_X\cdot\pow{\eta_X} = \id_{\pow{X}}$, and
	$\mu_X\cdot\pow{\mu_X} = \mu_X\cdot\mu_{\pow{X}}$.
Consequently, $\pow{\!}$ is a monad whose Kleisli category is
the allegory of relations.
\end{lem}

\begin{proof}[Proof of Lemma~\ref{lem:pow-mon-naturality}]
\begin{itemize}
	\item \textbf{$\eta$ is natural:} Let $\map{f}{X}{Y}$. We have seen that 
	$\eta_Y\cdot f = \xi_{\pair{f}{\id}}$, and by unicity, it is enough to prove 
	that $\pow{f}\cdot\eta_X = \xi_{\pair{f}{\id}}$. We have the following composition of pullbacks:
\begin{center}
\begin{tikzpicture}[scale=1.5]
		
	\node (ef) at (3,1) {\scriptsize{$X$}};
	\node (ypx) at (3,0) {\scriptsize{$X\times X$}};
	\node (zpx) at (3,-1) {\scriptsize{$Y\times X$}};
	\node (ey) at (8,1) {\scriptsize{$X$}};
	\node (ypy) at (8,0) {\scriptsize{$X\times\pow{X}$}};
	\node (zpy) at (8,-1) {\scriptsize{$Y\times\pow{X}$}};
	
	\path[->,font=\scriptsize]
		(ef) edge node[left]{$\pair{\id}{\id}$} (ypx)
		(ypx) edge node[left]{$\product{f}{\id}$} (zpx)
		(ef) edge node[above]{$\theta_X$} (ey)
		(zpx) edge node[below]{$\product{\id}{\eta_X}$} (zpy)
		(ey) edge node[right]{$\belong{X}$} (ypy)
		(ypy) edge node[right]{$\product{f}{\id}$} (zpy)
		(ypx) edge node[above]{$\product{\id}{\eta_X}$} (ypy);
		
	\draw (3.2,0.55) -- (3.7,0.55) -- (3.7,0.8);
	\draw (3.2,-0.45) -- (3.7,-0.45) -- (3.7,-0.2);
			
\end{tikzpicture}
\end{center}
Since $\product{f}{\id}\cdot\pair{\id}{\id} = \pair{f}{\id}$ is a mono, and the (epi, mono)-factorisation of $\product{f}{\id}\cdot\belong{X}$ is given by $m_f\cdot e_f$, then by preservation of the image by pullback, we have the following composition of pullbacks:
\begin{center}
\begin{tikzpicture}[scale=1.5]
		
	\node (r1r2) at (3,1) {\scriptsize{$X$}};
	\node (r1) at (3,0) {\scriptsize{$Y\times X$}};
	\node (r2) at (5.5,1) {\scriptsize{$E_f$}};
	\node (y) at (5.5,0) {\scriptsize{$Y\times\pow{X}$}};
	\node (ez) at (8,1) {\scriptsize{$Y$}};
	\node (zpz) at (8,0) {\scriptsize{$Y\times\pow{Y}$}};
	
	\path[->,font=\scriptsize]
		(r1r2) edge node[left]{$\pair{f}{\id}$} (r1)
		(r1r2) edge node[above]{$\theta$} (r2)
		(r1) edge node[below]{$\product{\id}{\eta_X}$} (y)
		(r2) edge node[left]{$m_f$} (y)
		(r2) edge node[above]{$\theta_f$} (ez)
		(y) edge node[below]{$\product{\id}{\pow{f}}$} (zpz)
		(ez) edge node[right]{$\belong{Y}$} (zpz);
		
	\draw (3.2,0.55) -- (3.45,0.55) -- (3.45,0.8);
	\draw (5.7,0.55) -- (5.95,0.55) -- (5.95,0.8);
			
\end{tikzpicture}
\end{center}
for some $\theta$.
	\item \textbf{$\mu$ is natural:} Let $\map{f}{X}{Y}$. Observe that we have the 
	following monos representing the same relations:
	\[
		\rcomp{\pair{f}{\id}}{\belongsq{X}}\,
		\equiv \rcomp{\pair{f}{\id}}{\rcomp{\belong{X}}{\belong{\pow{X}}}}
		\equiv \rcomp{m_f}{\belong{\pow{X}}}.
	\]
	Then, by Proposition~\ref{prop:composition}, we have:
	\begin{itemize}
		\item $\xi_{\rcomp{\pair{f}{\id}}{\belongsq{X}}} 
		= \mu_Y\cdot\pow{\xi_{\pair{f}{\id}}}\cdot\xi_{\belongsq{X}}
		= \mu_Y\cdot\pow{(\eta_Y\cdot f)}\cdot\mu_X
		=\pow{f}\cdot\mu_X$, by a coherence axiom that we prove next.
		\item $\xi_{\rcomp{m_f}{\belong{\pow{X}}}}
		= \mu_Y\cdot\pow{\xi_{m_f}}\cdot\xi_{\belong{\pow{X}}}
		= \mu_Y\cdot\pow{\pow{f}}$, which uses the fact that 
		$\xi_{\belong{X}} = \id_{\pow{X}}$.
	\end{itemize}
\item \textbf{coherence axioms:}
\begin{itemize}
	\item $\id_{\pow{X}} 
	= \xi_{\belong{X}} 
	= \xi_{\rcomp{\belong{X}}{\Delta_{\pow{X}}}} 
	= \mu_X\cdot\pow{\xi_{\belong{X}}}\cdot\xi_{\Delta_{\pow{X}}} 
	= \mu_X\cdot\eta_{\pow{X}}$.
	\item $\id_{\pow{X}} 
	= \xi_{\belong{X}} 
	= \xi_{\rcomp{\Delta_X}{\belong{X}}} 
	= \mu_X\cdot\pow{\xi_{\Delta_X}}\cdot\xi_{\belong{X}} 
	= \mu_X\cdot\pow{\eta_X}$.
	\item $\xi_{\rcomp{\belong{X}}{\rcomp{\belong{\pow{X}}}{\belong{\pow{\pow{X}}}}}} 
	= \mu_X\cdot\pow{\xi_{\rcomp{\belong{X}}{\belong{\pow{X}}}}}\cdot\xi_{\belong{\pow{\pow{X}}}} 
	= \mu_X\cdot\pow{\xi_{\belong{X}^2}} 
	= \mu_X\cdot\pow{\mu_X}$ 
	and
	$\xi_{\rcomp{\belong{X}}{\rcomp{\belong{\pow{X}}}{\belong{\pow{\pow{X}}}}}} 
	= \mu_X\cdot\pow{\xi_{\belong{X}}}\cdot\xi_{\rcomp{\belong{\pow{X}}}{\belong{\pow{\pow{X}}}}} 
	= \mu_X\cdot\xi_{\belong{\pow{X}}^2} 
	= \mu_X\cdot\mu_{\pow{X}}$.\qedhere
\end{itemize}
\end{itemize}
\end{proof}
\noindent 
The remaining part of Theorem~\ref{thm:pow-monad} is about the strength of the monad. 
This will be explained as a particular case of 
proto-distributive laws later on.

\subsection{Pseudo-Inverse of a Morphism}

Let us continue this section with another useful consequence of Lemma~\ref{prop:composition}.
A morphism $\map{f}{X}{Y}$ induces a Kleisli morphism 
$\map{\eta_Y\cdot f}{X}{\pow{Y}}$, or a relation as the right adjoint 
$\pair{f}{\id}$.
This relation has a converse which is given by its left adjoint (or map) 
$\pair{\id}{f}$.
This relation then corresponds to a unique morphism 
$\map{f^\dagger = \xi_{\pair{\id}{f}}}{Y}{\pow{X}}$.
This pseudo-inverse has nice properties when $f$ is a mono or an epi:
\begin{prop}
\label{prop:pseudo-inverse}
When $f$ is a mono, $f^\dagger\cdot f = \eta_X$ and $\pow{f}$ is a split 
mono.
When $f$ is an epi, $\pow{f}\cdot f^\dagger = \eta_Y$ and $\pow{f}$ is a 
split epi.
\end{prop}
\begin{cor}
$\pow{\!}$ preserves epis and monos, and so (epi, mono)-factorisations.
\end{cor}
\noindent 
When translating the proof of Proposition~\ref{prop:composition-reg} to toposes, 
the main argument becomes
the fact that $\pow{\!}$ maps epis to split epis.

\begin{proof}[Proof of Proposition~\ref{prop:pseudo-inverse}]
We have the following composition of pullbacks:
\begin{center}
\begin{tikzpicture}[scale=1.5]
		
	\node (r1r2) at (3,1) {\scriptsize{$X$}};
	\node (r1) at (3,0) {\scriptsize{$X\times X$}};
	\node (r2) at (5.5,1) {\scriptsize{$X$}};
	\node (y) at (5.5,0) {\scriptsize{$X\times Y$}};
	\node (ez) at (8,1) {\scriptsize{$E_X$}};
	\node (zpz) at (8,0) {\scriptsize{$X\times\pow{X}$}};
	
	\path[->,font=\scriptsize]
		(r1r2) edge node[left]{$\pair{\id}{\id}$} (r1)
		(r1r2) edge node[above]{$\id$} (r2)
		(r1) edge node[below]{$\product{\id}{f}$} (y)
		(r2) edge node[left]{$\pair{\id}{f}$} (y)
		(r2) edge node[above]{$\theta_X$} (ez)
		(y) edge node[below]{$\product{\id}{f^\dagger}$} (zpz)
		(ez) edge node[right]{$\belong{X}$} (zpz);
		
	\draw (3.2,0.55) -- (3.45,0.55) -- (3.45,0.8);
	\draw (5.7,0.55) -- (5.95,0.55) -- (5.95,0.8);
			
\end{tikzpicture}
\end{center}
Observe that the left one is a pullback only when $f$ is a mono.
Then the equality holds by unicity of $\xi_{\pair{\id}{\id}} = \eta_X$.
From this equality, we deduce that 
\[(\mu_X\cdot\pow{f^\dagger})\cdot\pow{f} = \mu_X\cdot\pow{\eta_X} = \id,\]
and so that $\pow{f}$ is a split mono.

When $f$ is an epi, 
\[\rcomp{\pair{f}{\id}}{\pair{\id}{f}} = \pair{\id}{\id}.\]
Indeed, $\rcomp{\pair{f}{\id}}{\pair{\id}{f}}$ is reprsented by 
the mono part of the (epi, mono)-factorisation of $\pair{f}{f}$, 
which is given by $\pair{\id}{\id}\cdot\,f$ when $f$ is epi.
Consequently, from Proposition~\ref{prop:composition}:
\begin{align*}
\eta_Y & = \xi_{\pair{\id}{\id}}
	= \xi_{\rcomp{\pair{f}{\id}}{\pair{\id}{f}}}
	= \mu_Y\cdot\pow{\xi_{\pair{f}{\id}}}\cdot\xi_{\pair{\id}{f}}\\
	& = \mu_Y\cdot\pow{\eta_Y}\cdot\pow{f}\cdot f^\dagger
	= \pow{f}\cdot f^\dagger.
\end{align*}

From this equality, we deduce that 
\[\pow{f}\cdot(\mu_X\cdot\pow{f^\dagger}) 
	= \mu_Y\cdot\pow{(\pow{f}\cdot f^\dagger)} 
	= \mu_Y\cdot\pow{\eta_Y} 
	= \id,\]
and so that $\pow{f}$ is a split epi.
\end{proof}

\subsection{Proto-Distributive Laws}
\label{sec:distributive}

As a side remark, we can easily derive some candidates for (weak) distributive laws
for every functor, also called cross-operator in~\cite{demoor94}.
Everything written here already appears in some form in~\cite{goy21}, but proved in a purely
relational way.

The power-set monad (and more generally, the 
power-object monad) is often combined with other functors 
to model the non-determinism of a system. Having weak distributive laws then 
allows to simplify the analysis by transferring it from the original category 
to the Kleisli category (which we know well in the case of the 
power-object monad). See for example \cite{urabe18}.

The interesting observation behind the definition of power-objects is that 
there is a canonical way to define a candidate for a distributive law of 
$\pow{\!}$ over \emph{any functor $F$}. We will see that these canonical 
candidates give rise to well-known (weak) distributive laws in the literature.

Given a functor $F$ on the topos and any object $X$ of the topos, we
define $\map{\distr{F}{X}}{F\pow{X}}{\pow{FX}}$ as the usual pattern
(epi, mono)-factorisation followed by unique morphism from the definition.
In this case, we consider the (epi, mono)-factorisation of 
$\pair{F\pi_1}{F\pi_2}\cdot F\belong{X}$:
\begin{center}
\begin{tikzpicture}[scale=1.5]
		
	\node (r1r2) at (0,0) {\scriptsize{$FE_X$}};
	\node (xz) at (6,0) {\scriptsize{$FX\times F\pow X$}};
	\node (r1sqr2) at (3,-1) {\scriptsize{$E_{F,X}$}};
	
	\path[->,font=\scriptsize]
		(r1r2) edge node[above]{$\pair{F\pi_1}{F\pi_2}\cdot F\belong{X}$} (xz);
		
	\path[->>,font=\scriptsize]
		(r1r2) edge node[below]{$e_{F,X}$} (r1sqr2);
		
	\path[>->,font=\scriptsize]
		(r1sqr2) edge node[below]{$m_{F,X}$} (xz);
			
\end{tikzpicture}
\end{center}
and $\distr{F}{X}$ is defined as $\xi_{m_{F,X}}$.

Those proto-distributive laws are related to liftings of functors to the Kleisli 
category, here to the category of relations. In the case of $\Set$, several 
papers \cite{goy20,garner20} investigate this connection, and particularly, some 
conditions are given for the existence of (weak) distributive laws.
We can prove a similar theorem in any topos, as already stated in~\cite{goy21}:
\begin{prop}
\label{prop:distr_nat_1}
If $F$ preserves weak pullbacks and epis, 
then $\distr{F}{X}$ is natural in $X$. Furthermore, we have: 
$\distr{F}{X}\cdot F\eta_X = \eta_{FX}$ and 
$\distr{F}{X}\cdot F\mu_X = 
	\mu_{FX}\cdot\pow{\distr{F}{X}}\cdot\distr{F}{\pow{X}}$.
\end{prop}
\begin{rem}
In $\Set$, there is no need for the second condition, as any functor 
preserves epis: every epi is split in $\Set$ 
by the axiom of choice.
\end{rem}

\begin{proof}
Let $\map{f}{X}{Y}$. We want to prove that the following square commutes:
\begin{center}
\begin{tikzpicture}[scale=1.5]
		
	\node (r1r2) at (3,1) {\scriptsize{$F\pow{X}$}};
	\node (r1) at (3,0) {\scriptsize{$F\pow{Y}$}};
	\node (r2) at (8,1) {\scriptsize{$\pow{FX}$}};
	\node (y) at (8,0) {\scriptsize{$\pow{FY}$}};
	
	\path[->,font=\scriptsize]
		(r1r2) edge node[left]{$F\pow{f}$} (r1)
		(r1r2) edge node[above]{$\distr{F}{X}$} (r2)
		(r1) edge node[below]{$\distr{F}{Y}$} (y)
		(r2) edge node[right]{$\pow{Ff}$} (y);
			
\end{tikzpicture}
\end{center}
On one side, we have the following composition of pullbacks:
\begin{center}
\begin{tikzpicture}[scale=1.5]
		
	\node (efcx) at (3,1) {\scriptsize{$E_{F,X}$}};
	\node (fxfpx) at (3,0) {\scriptsize{$FX\times F\pow{X}$}};
	\node (efx) at (8,1) {\scriptsize{$E_{FX}$}};
	\node (fxpfx) at (8,0) {\scriptsize{$FX\times\pow{FX}$}};
	\node (fyfpx) at (3,-1) {\scriptsize{$FY\times F\pow{X}$}};
	\node (fypfx) at (8,-1) {\scriptsize{$FY\times\pow{FX}$}};
	
	\path[->,font=\scriptsize]
		(efcx) edge node[above]{$\theta_{F,X}$} (efx)
		(efcx) edge node[left]{$m_{F,X}$} (fxfpx)
		(fxfpx) edge node[left]{$\product{Ff}{\id}$} (fyfpx)
		(efx) edge node[right]{$\belong{FX}$} (fxpfx)
		(fxpfx) edge node[right]{$\product{Ff}{\id}$} (fypfx)
		(fxfpx) edge node[above]{$\product{\id}{\distr{F}{X}}$} (fxpfx)
		(fyfpx) edge node[below]{$\product{\id}{\distr{F}{X}}$} (fypfx);
		
	\draw (3.2,0.55) -- (3.45,0.55) -- (3.45,0.8);
	\draw (3.2,-0.45) -- (3.45,-0.45) -- (3.45,-0.2);
			
\end{tikzpicture}
\end{center}
Since $\monop{(\product{Ff}{\id})\cdot\belong{FX}} = m_{Ff}$, then 
by preservation of image by pullbacks, we have 
a pullback of the following shape:
\begin{center}
\begin{tikzpicture}[scale=1.5]
		
	\node (u) at (3,1) {\scriptsize{$U$}};
	\node (fyfpx) at (3,0) {\scriptsize{$FY\times F\pow{X}$}};
	\node (efy) at (8,1) {\scriptsize{$E_{FY}$}};
	\node (fypfy) at (8,0) {\scriptsize{$FY\times\pow{FY}$}};
	
	\path[->,font=\scriptsize]
		(u) edge (efy)
		(u) edge node[left]{$m_1$} (fyfpx)
		(efy) edge node[right]{$\belong{FY}$} (fypfy)
		(fyfpx) edge node[below]{$\product{\id}{(\pow{Ff}\cdot\distr{F}{X})}$} (fypfy);
		
	\draw (3.2,0.55) -- (3.45,0.55) -- (3.45,0.8);
			
\end{tikzpicture}
\end{center}
where $m_1 = \monop{(\product{Ff}{\id})\cdot m_{F,X}}$.
So, $\pow{Ff}\cdot\distr{F}{X}$ is the unique morphism associated to 
$m_1$, and it is enough to prove that $\distr{F}{Y}\cdot F\pow{f}$ is 
also associated to $m_1$ to conclude.

On the other side, we have the following composition of weak pullbacks:
\begin{center}
\begin{tikzpicture}[scale=1.5]
		
	\node (fef) at (3,1) {\scriptsize{$FE_f$}};
	\node (fypx) at (3,0) {\scriptsize{$F(Y\times \pow{X})$}};
	\node (fey) at (8,1) {\scriptsize{$FE_Y$}};
	\node (fypy) at (8,0) {\scriptsize{$F(Y\times\pow{Y})$}};
	\node (fyfypx) at (3,-1) {\scriptsize{$FY\times F(Y\times \pow{X})$}};
	\node (fyfypy) at (8,-1) {\scriptsize{$FY\times F(Y\times \pow{Y})$}};
	\node (fyfpx) at (3,-2) {\scriptsize{$FY\times F\pow{X}$}};
	\node (fyfpy) at (8,-2) {\scriptsize{$FY\times F\pow{Y}$}};
	
	\path[->,font=\scriptsize]
		(fef) edge node[above]{$F\theta_f$} (fey)
		(fypx) edge node[above]{$F(\product{\id}{\pow{f}})$} (fypy)
		(fyfypx) edge node[below]{$\product{\id}{F(\product{\id}{\pow{f}})}$} (fyfypy)
		(fyfpx) edge node[below]{$\product{\id}{F\pow{f}}$} (fyfpy)
		(fef) edge node[left]{$Fm_f$} (fypx)
		(fypx) edge node[left]{$\pair{F\pi_1}{\id}$} (fyfypx)
		(fyfypx) edge node[left]{$\product{\id}{F\pi_2}$} (fyfpx)
		(fey) edge node[right]{$F\belong{Y}$} (fypy)
		(fypy) edge node[right]{$\pair{F\pi_1}{\id}$} (fyfypy)
		(fyfypy) edge node[right]{$\product{\id}{F\pi_2}$} (fyfpy);
		
	\draw[dashed] (3.2,0.55) -- (3.45,0.55) -- (3.45,0.8);
	\draw (3.2,-0.45) -- (3.45,-0.45) -- (3.45,-0.2);
	\draw[dashed] (3.2,-1.45) -- (3.45,-1.45) -- (3.45,-1.2);
			
\end{tikzpicture}
\end{center}
The upper weak pullback comes from the preservation of weak pullbacks 
by $F$ and by the definition of $\pow{f}$. The middle pullback is easy.
The lower weak pullback comes from the preservation of weak pullbacks by 
$F$ and product functors.
Then by preservation of images by weak pullbacks, there is a 
weak pullback of the following shape:
\begin{center}
\begin{tikzpicture}[scale=1.5]
		
	\node (v) at (3,1) {\scriptsize{$V$}};
	\node (fyfpx) at (3,0) {\scriptsize{$FY\times F\pow{X}$}};
	\node (efy) at (8,1) {\scriptsize{$E_{FY}$}};
	\node (fypfy) at (8,0) {\scriptsize{$FY\times\pow{FY}$}};
	
	\path[->,font=\scriptsize]
		(v) edge (efy)
		(v) edge node[left]{$m_2$} (fyfpx)
		(efy) edge node[right]{$\belong{FY}$} (fypfy)
		(fyfpx) edge node[below]{$\product{\id}{(\distr{F}{Y}\cdot F\pow{f})}$} (fypfy);
		
	\draw[dashed] (3.2,0.55) -- (3.45,0.55) -- (3.45,0.8);
			
\end{tikzpicture}
\end{center}
where $m_2 = \monop{\pair{F\pi_1}{F\pi_2}\cdot Fm_f}$. To conclude, 
it is enough to prove that $m_1 \equiv m_2$:

\begin{center}
\begin{tabular}{rclcr}
    $m_2$ & $\equiv$ & $\monop{\pair{F\pi_1}{F\pi_2}\cdot Fm_f}$
    & & \hfill (definition)\\
    & $\equiv$ & $\monop{\pair{F\pi_1}{F\pi_2}\cdot Fm_f \cdot Fe_f}$
    & & \hfill ($F$ preserves epis)\\
    & $\equiv$ & $\monop{\pair{F\pi_1}{F\pi_2}\cdot F(\product{f}{\id})\cdot F\belong{X}}$
    & & \hfill (definition)\\
    & $\equiv$ & $\monop{(\product{Ff}{\id})\cdot\pair{F\pi_1}{F\pi_2}\cdot F\belong{X}}$
    & & \hfill (calculation)\\
    & $\equiv$ & $\monop{(\product{Ff}{\id})\cdot m_{F,X}\cdot e_{F,X}} $
    & & \hfill (definition)\\
    & $\equiv$ & $\monop{(\product{Ff}{\id})\cdot m_{F,X}}$
    & & \hfill ($e_{F,X}$ epi)\\
    & $\equiv$ & $m_1$
    & & \hfill (definition)
\end{tabular}
\end{center}

Now, let us prove the first coherence axiom:
\begin{center}
\begin{tikzpicture}[scale=1.5]
		
	\node (fpx) at (0,-0.25) {\scriptsize{$F\pow{X}$}};
	\node (pfx) at (5,-0.25) {\scriptsize{$\pow{FX}$}};
	\node (fx) at (2.5,-1.25) {\scriptsize{$FX$}};
	
	\path[->,font=\scriptsize]
		(fpx) edge node[above]{$\distr{F}{X}$} (pfx)
		(fx) edge node[below]{$F\eta_X$} (fpx)
		(fx) edge node[below]{$\eta_{FX}$} (pfx);
\end{tikzpicture}
\end{center}
Similarly to the above proof, we have the following composition of 
weak pullbacks:
\begin{center}
\begin{tikzpicture}[scale=1.5]
		
	\node (fef) at (3,1) {\scriptsize{$FX$}};
	\node (fypx) at (3,0) {\scriptsize{$F(X\times X)$}};
	\node (fey) at (8,1) {\scriptsize{$FE_X$}};
	\node (fypy) at (8,0) {\scriptsize{$F(X\times\pow{X})$}};
	\node (fyfypx) at (3,-1) {\scriptsize{$FX\times F(X\times X)$}};
	\node (fyfypy) at (8,-1) {\scriptsize{$FY\times F(X\times \pow{X})$}};
	\node (fyfpx) at (3,-2) {\scriptsize{$FX\times FX$}};
	\node (fyfpy) at (8,-2) {\scriptsize{$FX\times F\pow{X}$}};
	
	\path[->,font=\scriptsize]
		(fef) edge node[above]{$F\theta_X$} (fey)
		(fypx) edge node[above]{$F(\product{\id}{\eta_X})$} (fypy)
		(fyfypx) edge node[below]{$\product{\id}{F(\product{\id}{\eta_X})}$} (fyfypy)
		(fyfpx) edge node[below]{$\product{\id}{F\eta_X}$} (fyfpy)
		(fef) edge node[left]{$F\pair{\id}{\id}$} (fypx)
		(fypx) edge node[left]{$\pair{F\pi_1}{\id}$} (fyfypx)
		(fyfypx) edge node[left]{$\product{\id}{F\pi_2}$} (fyfpx)
		(fey) edge node[right]{$F\belong{X}$} (fypy)
		(fypy) edge node[right]{$\pair{F\pi_1}{\id}$} (fyfypy)
		(fyfypy) edge node[right]{$\product{\id}{F\pi_2}$} (fyfpy);
		
	\draw[dashed] (3.2,0.55) -- (3.45,0.55) -- (3.45,0.8);
	\draw (3.2,-0.45) -- (3.45,-0.45) -- (3.45,-0.2);
	\draw[dashed] (3.2,-1.45) -- (3.45,-1.45) -- (3.45,-1.2);
			
\end{tikzpicture}
\end{center}
By preservation of images by weak pullbacks, this implies that 
$\distr{F}{X}\cdot F\eta_X$ is the unique morphism associated to 
\[\monop{\pair{F\pi_1}{F\pi_2}\cdot F\pair{\id}{\id}} = \pair{\id_{FX}}{\id_{FX}},\] 
and so is $\eta_{FX}$.

Finally, let us prove the second coherence axiom:
\begin{center}
\begin{tikzpicture}[scale=1.5]

	\node (fppx) at (4,0) {\scriptsize{$F\pow{\pow{X}}$}};
	\node (fpx) at (9,0) {\scriptsize{$F\pow{X}$}};
	\node (pfpx) at (4,-0.75) {\scriptsize{$\pow{F\pow{X}}$}};
	\node (ppfx) at (4,-1.5) {\scriptsize{$\pow{\pow{FX}}$}};
	\node (pfx) at (9,-1.5) {\scriptsize{$\pow{FX}$}};
	
	\path[->,font=\scriptsize]
		(fppx) edge node[above]{$F\mu_X$} (fpx)
		(fppx) edge node[left]{$\distr{F}{\pow{X}}$} (pfpx)
		(pfpx) edge node[left]{$\pow{\distr{F}{X}}$} (ppfx)
		(fpx) edge node[right]{$\distr{F}{X}$} (pfx)
		(ppfx) edge node[below]{$\mu_{FX}$} (pfx);
			
\end{tikzpicture}
\end{center}
First, by Proposition~\ref{prop:composition}, 
$\mu_{FX}\cdot\pow{\distr{F}{X}}\cdot\distr{F}{\pow{X}}$
is the unique morphism associated with $\rcomp{m_{F,X}}{m_{F,\pow{X}}}$.
On the other side, with the same kind of composition of weak pullbacks, 
we have that $\distr{F}{X}\cdot F\mu_X$ is the unique morphism associated 
to 
\[\monop{\pair{F\pi_1}{F\pi_2}\cdot F\belongsq{X}},\] so it is enough to 
prove that both monos are the same. Using the fact that $F$ preserves 
epis, we can observe that 
\begin{equation}
\label{eq:kappa}
	\monop{\pair{F\pi_1}{F\pi_2}\cdot F\belongsq{X}} \equiv 
	\monop{\pair{F(\pi_1\cdot\belong{X}\cdot\kappa_1)}
		{F(\pi_2\cdot\belong{\pow{X}}\cdot\kappa_2)}},
\end{equation}
where $\kappa_1$ and $\kappa_2$ are obtained with the following 
pullback:
\begin{center}
\begin{tikzpicture}[scale=1.5]
		
	\node (r1r2) at (3,1) {\scriptsize{$E^3_X$}};
	\node (r1) at (3,0) {\scriptsize{$E_X$}};
	\node (r2) at (8,1) {\scriptsize{$E_{\pow{X}}$}};
	\node (y) at (8,0) {\scriptsize{$\pow{X}$}};
	
	\path[->,font=\scriptsize]
		(r1r2) edge node[left]{$\kappa_{1}$} (r1)
		(r1r2) edge node[above]{$\kappa_{2}$} (r2)
		(r1) edge node[below]{$\pi_2\cdot\belong{X}$} (y)
		(r2) edge node[right]{$\pi_1\cdot\belong{\pow{X}}$} (y);
		
	\draw (3.2,0.55) -- (3.7,0.55) -- (3.7,0.8);
			
\end{tikzpicture}
\end{center}
By preservation of weak pullbacks by $F$, the following is then a 
weak pullback:
\begin{center}
\begin{tikzpicture}[scale=1.5]
		
	\node (r1r2) at (3,1) {\scriptsize{$FE^3_X$}};
	\node (r1) at (3,0) {\scriptsize{$FE_X$}};
	\node (r2) at (8,1) {\scriptsize{$FE_{\pow{X}}$}};
	\node (y) at (8,0) {\scriptsize{$F\pow{X}$}};
	
	\path[->,font=\scriptsize]
		(r1r2) edge node[left]{$F\kappa_{1}$} (r1)
		(r1r2) edge node[above]{$F\kappa_{2}$} (r2)
		(r1) edge node[below]{$F(\pi_2\cdot\belong{X})$} (y)
		(r2) edge node[right]{$F(\pi_1\cdot\belong{\pow{X}})$} (y);
		
	\draw[dashed] (3.2,0.55) -- (3.7,0.55) -- (3.7,0.8);
			
\end{tikzpicture}
\end{center}
If we analyse the strict pullback of the same diagram, then we realise that 
it is also the limit of the following cospan:
\begin{center}
\begin{tikzpicture}[scale=1.5]
		
	\node (fpx) at (8,0) {\scriptsize{$F\pow{X}$}};
	\node (efx) at (5.5,0) {\scriptsize{$E_{F,X}$}};
	\node (fex) at (3,0) {\scriptsize{$FE_X$}};
	\node (efpx) at (8,1) {\scriptsize{$E_{F,\pow{X}}$}};
	\node (fepx) at (8,2) {\scriptsize{$FE_{\pow{X}}$}};
	
	\path[->>,font=\scriptsize]
		(fex) edge node[below]{$e_{F,X}$} (efx)
		(fepx) edge node[left]{$e_{F,\pow{X}}$} (efpx);
		
	\path[->,font=\scriptsize]
		(efx) edge node[below]{$\pi_2\cdot m_{F,X}$} (fpx)
		(efpx) edge node[left]{$\pi_1\cdot m_{F,\pow{X}}$} (fpx);
		
	\path[->,font=\scriptsize,bend left=40]
		(fepx) edge node[right]{$F(\pi_1\cdot\belong{\pow{X}})$} (fpx);
		
	\path[->,font=\scriptsize,bend right=30]
		(fex) edge node[below]{$F(\pi_2\cdot\belong{X})$} (fpx);
			
\end{tikzpicture}
\end{center}
Now, we can compute this pullback by computing four smaller pullbacks, 
which gives us the following situation, using the preservation of epis by 
pullbacks:
\begin{center}
\begin{tikzpicture}[scale=1.5]
		
	\node (fpx) at (8,0) {\scriptsize{$F\pow{X}$}};
	\node (efx) at (5.5,0) {\scriptsize{$E_{F,X}$}};
	\node (fex) at (3,0) {\scriptsize{$FE_X$}};
	\node (efpx) at (8,1) {\scriptsize{$E_{F,\pow{X}}$}};
	\node (fepx) at (8,2) {\scriptsize{$FE_{\pow{X}}$}};
	\node (u) at (5.5,1) {\scriptsize{$U$}};
	\node (ana) at (5.5,2) {\scriptsize{$\bullet$}};
	\node (anl) at (3,1) {\scriptsize{$\bullet$}};
	\node (v) at (3,2) {\scriptsize{$V$}};
	
	\path[->>,font=\scriptsize]
		(fex) edge node[below]{$e_{F,X}$} (efx)
		(fepx) edge node[right]{$e_{F,\pow{X}}$} (efpx)
		(v) edge (ana)
		(v) edge (anl)
		(ana) edge (u)
		(anl) edge (u);
		
	\path[->,font=\scriptsize]
		(efx) edge node[below]{$\pi_2\cdot m_{F,X}$} (fpx)
		(efpx) edge node[right]{$\pi_1\cdot m_{F,\pow{X}}$} (fpx)
		(u) edge node[left]{$\rho_1$} (efx)
		(u) edge node[above]{$\rho_2$} (efpx)
		(ana) edge (fepx)
		(anl) edge (fex);

	\draw (3.2,0.55) -- (3.7,0.55) -- (3.7,0.8);
	\draw (3.2,1.55) -- (3.7,1.55) -- (3.7,1.8);
	\draw (5.7,0.55) -- (5.95,0.55) -- (5.95,0.8);
	\draw (5.7,1.55) -- (5.95,1.55) -- (5.95,1.8);
			
\end{tikzpicture}
\end{center}
Since $FE^3_X$ is a weak pullback of this cospan, the unique morphism 
of cones from $FE^3_X$ to $V$ is a split epi. In total, this means that there 
is an epi $\epi{u}{FE^3_X}{U}$ such that 
\begin{equation}
\label{eq:u}
	\rho_1\cdot u = e_{F,X}\cdot F\kappa_1 \quad \text{and} \quad
	\rho_2\cdot u = e_{F,\pow{X}}\cdot F\kappa_2.
\end{equation}
Now, the lower-right pullback is the one used to define the composition
$\rcomp{m_{F,X}}{m_{F,\pow{X}}}$, which means that:
\[
	\rcomp{m_{F,X}}{m_{F,\pow{X}}} = 
	\monop{\pair{\pi_1\cdot m_{F,X}\cdot\rho_1}{\pi_2\cdot m_{F,\pow{X}}\cdot\rho_2}}.
\]
To conclude, it is enough to observe:
\begin{center}
\begin{align*}
     & \mkern9mu \monop{\pair{\pi_1\cdot m_{F,X}\cdot\rho_1}{\pi_2\cdot m_{F,\pow{X}}\cdot\rho_2}} \\
     \equiv \mkern9mu & \mkern9mu \monop{\pair{\pi_1\cdot m_{F,X}\cdot\rho_1}{\pi_2\cdot m_{F,\pow{X}}\cdot\rho_2}\cdot u}
     & \hfill \text{($u$ is epi)}\\
     \equiv \mkern9mu & \mkern9mu \monop{\pair{\pi_1\cdot m_{F,X}\cdot e_{F,X}\cdot F\kappa_1}{\pi_2\cdot m_{F,\pow{X}}\cdot e_{F,\pow{X}}\cdot F\kappa_2}}
     & \hfill \text{(by \eqref{eq:u})}\\
     \equiv \mkern9mu &\mkern9mu  \monop{\pair{F(\pi_1\cdot\belong{X})\cdot F\kappa_1}{F(\pi_2\cdot\belong{\pow{X}})\cdot F\kappa_2}}
     & \hfill \text{(definition)}\\
     \equiv \mkern9mu & \mkern9mu \monop{\pair{F\pi_1}{F\pi_2}\cdot F\belongsq{X}}
     & \hfill \text{(by \eqref{eq:kappa})} \tag*{\qedhere}
\end{align*}
\end{center}
\end{proof}
\noindent 
In \cite{garner20}, some conditions are also given to get (weak) distributive laws.
Those results can be encompassed in a result about naturality of 
$\distr{F}{X}$ with respect to $F$ in the following sense:
\begin{prop}
\label{prop:distr_nat_2}
Assume given a natural transformation $\natt{\tau}{F}{G}$ such that 
its naturality squares are weak pullbacks. Then the following diagram 
commutes for any $X$:
\begin{center}
\begin{tikzpicture}[scale=1.5]
		
	\node (r1r2) at (3,1) {\scriptsize{$F\pow{X}$}};
	\node (r1) at (3,0) {\scriptsize{$G\pow{X}$}};
	\node (r2) at (8,1) {\scriptsize{$\pow{FX}$}};
	\node (y) at (8,0) {\scriptsize{$\pow{GX}$}};
	
	\path[->,font=\scriptsize]
		(r1r2) edge node[left]{$\tau_{\pow{X}}$} (r1)
		(r1r2) edge node[above]{$\distr{F}{X}$} (r2)
		(r1) edge node[below]{$\distr{G}{X}$} (y)
		(r2) edge node[right]{$\pow{\tau_X}$} (y);
			
\end{tikzpicture}
\end{center}
\end{prop}

\begin{proof}
On one side, we have the following composition of pullbacks:
\begin{center}
\begin{tikzpicture}[scale=1.5]
		
	\node (efcx) at (3,1) {\scriptsize{$E_{F,X}$}};
	\node (fxfpx) at (3,0) {\scriptsize{$FX\times F\pow{X}$}};
	\node (efx) at (8,1) {\scriptsize{$E_{FX}$}};
	\node (fxpfx) at (8,0) {\scriptsize{$FX\times\pow{FX}$}};
	\node (gxfpx) at (3,-1) {\scriptsize{$GX\times F\pow{X}$}};
	\node (gxpfx) at (8,-1) {\scriptsize{$GX\times \pow{GX}$}};
	
	\path[->,font=\scriptsize]
		(efcx) edge node[above]{$\theta_{F,X}$} (efx)
		(fxfpx) edge node[above]{$\product{\id}{\distr{F}{X}}$} (fxpfx)
		(gxfpx) edge node[below]{$\product{\id}{\distr{F}{X}}$} (gxpfx)
		(efcx) edge node[left]{$m_{F,X}$} (fxfpx)
		(fxfpx) edge node[left]{$\product{\tau_X}{\id}$} (gxfpx)
		(efx) edge node[right]{$\belong{FX}$} (fxpfx)
		(fxpfx) edge node[right]{$\product{\tau_X}{\id}$} (gxpfx);
		
	\draw (3.2,0.55) -- (3.45,0.55) -- (3.45,0.8);
	\draw (3.2,-0.45) -- (3.45,-0.45) -- (3.45,-0.2);
			
\end{tikzpicture}
\end{center}
which implies that $\pow{\tau_X}\cdot\distr{F}{X}$ is the unique morphism
corresponding to $\monop{(\product{\tau_X}{\id})\cdot m_{F,X}}$.
On the other side, we have the following composition of weak pullbacks:
\begin{center}
\begin{tikzpicture}[scale=1.5]
		
	\node (fex) at (3,1) {\scriptsize{$FE_X$}};
	\node (fxpx) at (3,0) {\scriptsize{$F(X\times \pow{X})$}};
	\node (gex) at (8,1) {\scriptsize{$GE_X$}};
	\node (gxpx) at (8,0) {\scriptsize{$G(X\times\pow{X})$}};
	\node (gxfxpx) at (3,-1) {\scriptsize{$GX\times F(X\times \pow{X})$}};
	\node (gxgxpx) at (8,-1) {\scriptsize{$GX\times G(X\times \pow{X})$}};
	\node (gxfpx) at (3,-2) {\scriptsize{$GX\times F\pow{X}$}};
	\node (gxgpx) at (8,-2) {\scriptsize{$GX\times G\pow{X}$}};
	
	\path[->,font=\scriptsize]
		(fex) edge node[above]{$\tau_{E_X}$} (gex)
		(fxpx) edge node[above]{$\tau_{X\times\pow{X}}$} (gxpx)
		(gxfxpx) edge node[below]{$\product{\id}{\tau_{X\times\pow{X}}}$} (gxgxpx)
		(gxfpx) edge node[below]{$\product{\id}{\tau_{\pow{X}}}$} (gxgpx)
		(fex) edge node[left]{$F\belong{X}$} (fxpx)
		(fxpx) edge node[left]{$\pair{\tau_X\cdot F\pi_1}{\id}$} (gxfxpx)
		(gxfxpx) edge node[left]{$\product{\id}{F\pi_2}$} (gxfpx)
		(gex) edge node[right]{$G\belong{X}$} (gxpx)
		(gxpx) edge node[right]{$\pair{G\pi_1}{\id}$} (gxgxpx)
		(gxgxpx) edge node[right]{$\product{\id}{G\pi_2}$} (gxgpx);
		
	\draw[dashed] (3.2,0.55) -- (3.45,0.55) -- (3.45,0.8);
	\draw (3.2,-0.45) -- (3.45,-0.45) -- (3.45,-0.2);
	\draw[dashed] (3.2,-1.45) -- (3.45,-1.45) -- (3.45,-1.2);
			
\end{tikzpicture}
\end{center}
Indeed, the upper and lower weak pullbacks come from the naturality 
squares of $\tau$, and the middle pullback is easy.
This means that $\distr{G}{X}\cdot\tau_{\pow{X}}$ is the unique morphism 
corresponding to $\monop{\product{\id}{F\pi_2}\cdot\pair{\tau_X\cdot F\pi_1}{\id}\cdot F\belong{X}}$.
It is easy to check that both monos are the same.
\end{proof}

As stated in~\cite{goy21}:
\begin{cor}
\label{coro:distr}
If $(T,\mu^T,\eta^T)$ is a monad which preserves weak pullbacks and 
epis, and 
for which the naturality squares of $\mu^T$ are weak pullbacks, then 
$\distr{T}{X}$ is a weak distributive law. If the 
naturality squares of $\eta^T$ are also weak pullbacks, then 
$\distr{T}{X}$ is a distributive law.
\end{cor}
\noindent 
Before proving Corollary~\ref{coro:distr}, let us prove an easy lemma 
about $\distr{F}{X}$:
\begin{lem}
\label{lem:distr-func}
We have the following equalities:
\begin{itemize}
	\item $\distr{\text{Id}}{X} = \id_{\pow{X}}$,
	\item if $G$ preserves weak pullbacks and epis, then 
	$\distr{G\cdot F}{X} = \distr{G}{FX}\cdot G\distr{F}{X}$.
\end{itemize}
\end{lem}
\begin{proof}
\begin{itemize}
	\item By definition, $\distr{\text{Id}}{X}$ correspond to the mono 
	$\belong{X}$, which is also the case of $\id_{\pow{X}}$.
	\item By definition, $\distr{G\cdot F}{X}$ corresponds to the mono
	$m_{G\cdot F,X}$. Also, by the same kind of composition of 
	weak pullbacks as previous proofs, $\distr{G}{FX}\cdot G\distr{F}{X}$
	corresponds to the mono 
	$\monop{\pair{G\pi_1}{G\pi_2}\cdot Gm_{F,X}}$. Using the fact 
	that $G$ preserves epi and the definition of $m_{F,X}$, it is easy to 
	check that both monos are the same.\qedhere
\end{itemize}
\end{proof}

\begin{proof}[Proof of Corollary~\ref{coro:distr}]
For the first part, we can apply Proposition~\ref{prop:distr_nat_2} with 
$\mu^T$ which is a natural transformation from $TT$ to $T$. 
We then obtain:
\[
	\distr{T}{X}\cdot\mu^T_{\pow{X}} = \pow{\mu^T_X}\cdot\distr{TT}{X}.
\]
Then by Lemma~\ref{lem:distr-func}:
\[
	\distr{T}{X}\cdot\mu^T_{\pow{X}} = \pow{\mu^T_X}\cdot\distr{T}{TX}\cdot T\distr{T}{X}
\]
which is the coherence axiom to prove.
Similarly, the second part consists in using 
Proposition~\ref{prop:distr_nat_2} with $\eta^T$ and then the first point of 
Lemma~\ref{lem:distr-func}.
\end{proof}

\begin{exa}
\label{ex:distr}
In $\Set$, when $F$ is $\pow{\!}$ itself, 
$\distr{F}{X}$ is a weak distributive law, as described in 
\cite{goy20}, namely,
$
	\distr{\pow{\!}}{X}:~U \in \pow{\pow{X}}\mapsto 
	\{V \subseteq X \mid V \subseteq \bigcup U \wedge
		\forall W\in U.\, W\cap V \neq \emptyset\}.
$
A similar, analysis can be done for the distribution monad $\mathcal{D}$.
More generally (see \cite{goy21}), in any topos, $\pow{\!}$ satisfies the assumptions of  
Proposition~\ref{prop:distr_nat_1} and the first part of Corollary~\ref{coro:distr}, 
meaning that 
$\distr{\pow{\!}}{X}$ 
is a weak distributive law. 
However, it satisfies the second part only when the topos is trivial.
\end{exa}

\begin{proof}[Proof of Example~\ref{ex:distr}]
Here, we want to prove that in any topos $\distr{\pow{\!}}{X}$ satisfies the assumptions of
Proposition~\ref{prop:distr_nat_1}. We already know that $\pow{\!}$ preserves epis. Let us 
prove that it preserves weak pullbacks.

Assume given a weak pullback of the form:
\begin{center}
\begin{tikzpicture}[scale=1.5]
		
	\node (r1r2) at (3,1) {\scriptsize{$X$}};
	\node (r1) at (3,0) {\scriptsize{$Y_1$}};
	\node (r2) at (8,1) {\scriptsize{$Y_2$}};
	\node (y) at (8,0) {\scriptsize{$Z$}};
	
	\path[->,font=\scriptsize]
		(r1r2) edge node[left]{$\epsilon_1$} (r1)
		(r1r2) edge node[above]{$\epsilon_2$} (r2)
		(r1) edge node[below]{$\mu_1$} (y)
		(r2) edge node[right]{$\mu_2$} (y);
			
	\draw[dashed] (3.2,0.55) -- (3.7,0.55) -- (3.7,0.8);
\end{tikzpicture}
\end{center}
We want to prove that we have the following weak pullback:
\begin{center}
\begin{tikzpicture}[scale=1.5]
		
	\node (r1r2) at (3,1) {\scriptsize{$\pow{X}$}};
	\node (r1) at (3,0) {\scriptsize{$\pow{Y_1}$}};
	\node (r2) at (8,1) {\scriptsize{$\pow{Y_2}$}};
	\node (y) at (8,0) {\scriptsize{$\pow{Z}$}};
	
	\path[->,font=\scriptsize]
		(r1r2) edge node[left]{$\pow{\epsilon_1}$} (r1)
		(r1r2) edge node[above]{$\pow{\epsilon_2}$} (r2)
		(r1) edge node[below]{$\pow{\mu_1}$} (y)
		(r2) edge node[right]{$\pow{\mu_2}$} (y);
			
	\draw[dashed] (3.2,0.55) -- (3.7,0.55) -- (3.7,0.8);
\end{tikzpicture}
\end{center}
So we assume given another commutative square of the form:
\begin{center}
\begin{tikzpicture}[scale=1.5]
		
	\node (r1r2) at (3,1) {\scriptsize{$W$}};
	\node (r1) at (3,0) {\scriptsize{$\pow{Y_1}$}};
	\node (r2) at (8,1) {\scriptsize{$\pow{Y_2}$}};
	\node (y) at (8,0) {\scriptsize{$\pow{Z}$}};
	
	\path[->,font=\scriptsize]
		(r1r2) edge node[left]{$\phi_1$} (r1)
		(r1r2) edge node[above]{$\phi_2$} (r2)
		(r1) edge node[below]{$\pow{\mu_1}$} (y)
		(r2) edge node[right]{$\pow{\mu_2}$} (y);
			
\end{tikzpicture}
\end{center}

We want to construct a morphism $\map{\phi}{W}{\pow{X}}$, and the trick is to play with the correspondence with relations. First, let us form the following pullbacks:
\begin{center}
\begin{tikzpicture}[scale=1.5]
		
	\node (r1r2) at (3,1) {\scriptsize{$R_i$}};
	\node (r1) at (3,0) {\scriptsize{$Y_i\times W$}};
	\node (r2) at (5.3,1) {\scriptsize{$E_{Y_i}$}};
	\node (y) at (5.3,0) {\scriptsize{$Y_i\times\pow{Y_i}$}};
	
	\path[->,font=\scriptsize]
		(r1r2) edge node[left]{$r_i$} (r1)
		(r1r2) edge node[above]{$\theta_i$} (r2)
		(r1) edge node[below]{$\product{\id}{\phi_i}$} (y)
		(r2) edge node[right]{$\belong{Y_i}$} (y);
		
	\draw (3.2,0.55) -- (3.45,0.55) -- (3.45,0.8);
			
\end{tikzpicture}
\qquad\qquad
\begin{tikzpicture}[scale=1.5]
		
	\node (r1r2) at (3,1) {\scriptsize{$R$}};
	\node (r1) at (3,0) {\scriptsize{$R_1$}};
	\node (r2) at (5.3,1) {\scriptsize{$R_2$}};
	\node (y) at (5.3,0) {\scriptsize{$Z\times W$}};
	
	\path[->,font=\scriptsize]
		(r1r2) edge node[left]{$\rho_1$} (r1)
		(r1r2) edge node[above]{$\rho_2$} (r2)
		(r1) edge node[below]{$(\product{\mu_1}{\id})\cdot r_1$} (y)
		(r2) edge node[right]{$(\product{\mu_2}{\id})\cdot r_2$} (y);
		
	\draw (3.2,0.55) -- (3.45,0.55) -- (3.45,0.8);
			
\end{tikzpicture}
\end{center}

So by construction, we have the following commutative square:
\begin{center}
\begin{tikzpicture}[scale=1.5]
		
	\node (r1r2) at (3,1) {\scriptsize{$R$}};
	\node (r1) at (3,0) {\scriptsize{$Y_1$}};
	\node (r2) at (8,1) {\scriptsize{$Y_2$}};
	\node (y) at (8,0) {\scriptsize{$Z$}};
	
	\path[->,font=\scriptsize]
		(r1r2) edge node[left]{$\pi_1\cdot r_1\cdot\rho_1$} (r1)
		(r1r2) edge node[above]{$\pi_1\cdot r_2\cdot\rho_2$} (r2)
		(r1) edge node[below]{$\mu_1$} (y)
		(r2) edge node[right]{$\mu_2$} (y);
			
\end{tikzpicture}
\end{center}
and by the universal property of $X$, there is a (non necessarily unique) morphism $\map{\widehat{\phi}}{R}{X}$ such that $$\epsilon_i\cdot\widehat{\phi} = \pi_1\cdot r_i\cdot\rho_i.$$
Since we have $\pi_2\cdot r_1\cdot\rho_1 = \pi_2\cdot r_2\cdot \rho_2$, we have the following unique (epi, mono)-factorisation:
\begin{center}
\begin{tikzpicture}[scale=1.5]
		
	\node (r1r2) at (0,0) {\scriptsize{$R$}};
	\node (xz) at (6,0) {\scriptsize{$X\times W$}};
	\node (r1sqr2) at (3,-1) {\scriptsize{$\widehat{R}$}};
	
	\path[->,font=\scriptsize]
		(r1r2) edge node[above]{$\pair{\widehat{\phi}}{\pi_2\cdot r_i\cdot\rho_i}$} (xz);
		
	\path[->>,font=\scriptsize]
		(r1r2) edge node[below]{$e$} (r1sqr2);
		
	\path[>->,font=\scriptsize]
		(r1sqr2) edge node[below]{$m$} (xz);
			
\end{tikzpicture}
\end{center}

Define then $\phi = \xi_m$. To conclude, we need to prove that $\phi_i = \pow{\epsilon_i}\cdot\phi$. But we know that:
\begin{center}
\begin{tabular}{rclcl}
    $\phi_i$ & $=$ & $\xi_{r_i}$ & ~~~~~ & (definition of $\phi_i$)\\
    $\pow{\epsilon_i}\cdot\phi$& $=$ & $\mu_{Y_i}\cdot\pow{\eta_{Y_i}}\cdot\pow{\epsilon_i}\cdot\phi$ & ~~~~~ & (unit coherence axiom)\\
    & $=$ & $\mu_{Y_i}\cdot\pow{\xi_{\pair{\epsilon_i}{\id}}}\cdot\phi$& ~~~~~ & (calculation)\\
    & $=$ & $\mu_{Y_i}\cdot\pow{\xi_{\pair{\epsilon_i}{\id}}}\cdot\xi_m$& ~~~~~ & (definition of $m$)\\
    & $=$ & $\xi_{\rcomp{\pair{\epsilon_i}{\id}}{m}}$ & ~~~~~ & (Proposition~\ref{prop:composition})
\end{tabular}
\end{center}
So we need to prove that $r_i \equiv \rcomp{\pair{\epsilon_i}{\id}}{m}$. 
We know by definition of composition that $\rcomp{\pair{\epsilon_i}{\id}}{m}$
is $\monop{\pair{\epsilon_i\cdot\pi_1\cdot m}{\pi_2\cdot m}}$. 
Since $e$ is an epi, this is also 
$\monop{\pair{\epsilon_i\cdot\pi_1\cdot m}{\pi_2\cdot m}\cdot e} \equiv \monop{r_i\cdot\rho_i}$. 
So to conclude, it is enough to prove that $\rho_i$ is an epi.
By assumption, we know that 
$\pow{\mu_1}\cdot\phi_1 = \pow{\mu_2}\cdot\phi_2$. 
By using again the same trick, this implies that 
$\monop{\rcomp{\pair{\mu_1}{\id}}{r_1}} = \monop{\rcomp{\pair{\mu_2}{\id}}{r_2}}$.
Let us write $e_1$ and $e_2$ their corresponding epic parts. 
But we also know that we have the following pullback:
\begin{center}
\begin{tikzpicture}[scale=1.5]
		
	\node (r1r2) at (3,1) {\scriptsize{$R$}};
	\node (r1) at (3,0) {\scriptsize{$R_1$}};
	\node (r2) at (8,1) {\scriptsize{$R_2$}};
	\node (y) at (8,0) {\scriptsize{$Z\times W$}};
	
	\path[->,font=\scriptsize]
		(r1r2) edge node[left]{$\rho_1$} (r1)
		(r1r2) edge node[above]{$\rho_2$} (r2)
		(r1) edge node[below]{$(\product{\mu_1}{\id})\cdot r_1$} (y)
		(r2) edge node[right]{$(\product{\mu_2}{\id})\cdot r_2$} (y);
		
	\draw (3.2,0.55) -- (3.7,0.55) -- (3.7,0.8);
			
\end{tikzpicture}
\end{center}
which means we have the following pullback:
\begin{center}
\begin{tikzpicture}[scale=1.5]
		
	\node (r1r2) at (3,1) {\scriptsize{$R$}};
	\node (r1) at (3,0) {\scriptsize{$R_1$}};
	\node (r2) at (8,1) {\scriptsize{$R_2$}};
	\node (y) at (8,0) {\scriptsize{}};
	
	\path[->,font=\scriptsize]
		(r1r2) edge node[left]{$\rho_1$} (r1)
		(r1r2) edge node[above]{$\rho_2$} (r2)
		(r1) edge node[below]{$e_1$} (y)
		(r2) edge node[right]{$e_2$} (y);
		
	\draw (3.2,0.55) -- (3.7,0.55) -- (3.7,0.8);
			
\end{tikzpicture}
\end{center}
Since $e_i$ is epi, $\rho_i$ is epi by preservation of epis by pullbacks.
\end{proof}

As a consequence, let us look at the strength and costrength of $\pow{\!}$. 
Indeed, define the strength as:
$
	\map{\strength{X}{Y}=\distr{X\times\_}{Y}}
		{X\times\pow{Y}}{\pow{(X\times Y)}}.
$
\begin{prop}
$\strength{X}{Y}$ is the strength of $\pow{\!}$.
\end{prop}
\begin{proof}
Naturality in $X$ comes from Proposition~\ref{prop:distr_nat_1}.
Naturality in $Y$ comes from Proposition~\ref{prop:distr_nat_2}.
The coherence axioms are also consequences of either Proposition.
\end{proof}

Dually, the costrength can be defined as
$
	\map{\costrength{X}{Y} = \distr{\_\times Y}{X}}{\pow{X}\times Y}{\pow{(X\times Y)}}.
$
By the naturality of Proposition~\ref{prop:distr_nat_2}, we indeed have the
expected equality
$
	\costrength{X}{Y} = \pow{\lambda_{Y,X}}\cdot\strength{Y}{X}
		\cdot\lambda_{\pow{X},Y},
$
where $\map{\lambda_{X,Y}}{X\times Y}{Y\times X}$ is the symmetry of 
the product.
\begin{thm}
\label{th:com-str-mon}
$\pow{\!}$ is a commutative strong monad.
\end{thm}
\begin{proof}
The commutation axiom is a consequence of 
Proposition~\ref{prop:composition}.
\end{proof}

\section{AM-Bisimulations in a Topos}
\label{sec:toposes}
Since toposes are regular categories, the notion of regular AM-bisimulations makes sense.
We show here that they can be reformulated as follows.
\begin{defi}
\label{def:tobi}
We say that a relation is a \emph{toposal AM-bisimulation} from the coalgebra 
$\map{\alpha}{X}{FX}$ to $\map{\beta}{Y}{FY}$, if for any mono 
$\mono{r}{R}{X\times Y}$ representing it, there is a morphism 
$\map{W}{R}{\pow{FR}}$ such that:
\begin{center}
\begin{tikzpicture}[scale=1.5]
		
	\node (r) at (-2,0) {\scriptsize{$R$}};
	\node (xy) at (0.5,0.7) {\scriptsize{$X\times Y$}};
	\node (pfr) at (0.5,-0.7) {\scriptsize{$\pow{FR}$}};
	\node (fxfy) at (3.5,0.7) {\scriptsize{$F(X)\times F(Y)$}};
	\node (pfxy) at (3.5,-0.7) {\scriptsize{$\pow{F(X\times Y)}$}};
	\node (pfxfy) at (6,0) {\scriptsize{$\pow{F(X)}\times \pow{F(Y)}$}};
	\node (oi) at (5.6,-0.5) {\scriptsize{$\splitpf$}};
	\node (boi) at (5.6,0.5) {\scriptsize{$\product{\eta_{F(X)}}{\eta_{F(Y)}}$}};
	
	\path[->,font=\scriptsize]
		(xy) edge node[above]{$\product{\alpha}{\beta}$} (fxfy)
		(r) edge node[below]{$W$} (pfr)
		(pfxy) edge (pfxfy)
		(fxfy) edge (pfxfy)
		(r) edge node[above]{$r$} (xy)
		(pfr) edge node[below]{$\pow{Fr}$} (pfxy);
			
\end{tikzpicture}
\end{center}
\end{defi}
In other words, an $F$-toposal AM-bisimulation between $\alpha$ and $\beta$ is 
a $\pow{F}$-AM-bisimulation between $\eta\cdot\alpha$ and 
$\eta\cdot\beta$. Intuitively, this means that toposal bisimulations look at 
systems as non-deterministic. This allows us to \emph{collect} witnesses 
as a morphism $\map{W}{R}{\pow{FR}}$ instead of picking some, very much like regular AM-bisimulations.

We have to make sure that toposal and regular AM-bisimulations coincide.
\begin{prop}
\label{prop:lobi-tobi}
Assume that $\CC$ is a topos. 
Then for every relation $U$ from 
$X$ to $Y$, 
every coalgebra $\map{\alpha}{X}{FX}$ and $\map{\beta}{Y}{FY}$, 
$U$ is a toposal AM-bisimulation from $\alpha$ to $\beta$ if and only 
if it is a regular AM-bisimulation between them.
\end{prop}

\begin{proof}
Assume that $\CC$ is a topos.
\begin{itemize}
	\item Assume that we have a regular AM-bisimulation
    	\begin{center}
    \begin{tikzpicture}[scale=1.5]
    	
    	\node (frr) at (0,0) {\scriptsize{$W$}};	
    	\node (r) at (1.5,0.7) {\scriptsize{$R$}};
    	\node (xy) at (3.5,0.7) {\scriptsize{$X\times Y$}};
    	\node (pfr) at (1.5,-0.7) {\scriptsize{$FR$}};
    	\node (fxfy) at (6,0) {\scriptsize{$F(X)\times F(Y)$}};
    	\node (pfxy) at (3.5,-0.7) {\scriptsize{$F(X\times Y)$}};
    	\node (ai) at (0.6,0.55) {\scriptsize{$\pi_2\cdot w$}};
    	\node (bai) at (0.6,-0.55) {\scriptsize{$\pi_1\cdot w$}};
    	\node (oi) at (5.2,-0.6) {\scriptsize{$\splitf$}};
    	\node (boi) at (4.8,0.6) {\scriptsize{$\product{\alpha}{\beta}$}};
    	
    	\path[->,font=\scriptsize]
    		(frr) edge (r)
    		(xy) edge (fxfy)
    		(frr) edge (pfr) 
    		(pfxy) edge (fxfy)
    		(pfxy) edge (fxfy)
    		(r) edge node[above]{$r$} (xy)
    		(pfr) edge node[below]{$Fr$} (pfxy);
    			
    \end{tikzpicture}
    \end{center}
    The relation $\mono{w}{W}{FR\times R}$ uniquely corresponds to a morphism
    $\map{\xi_w}{R}{\pow{FR}}$. Let us prove that this witnesses $r$ as a toposal bisimulation
\begin{center}
\begin{tikzpicture}[scale=1.5]
		
	\node (r) at (0,0) {\scriptsize{$R$}};
	\node (xy) at (1.5,0.7) {\scriptsize{$X\times Y$}};
	\node (pfr) at (1.5,-0.7) {\scriptsize{$\pow{FR}$}};
	\node (fxfy) at (3.5,0.7) {\scriptsize{$F(X)\times F(Y)$}};
	\node (pfxy) at (3.5,-0.7) {\scriptsize{$\pow{F(X\times Y)}$}};
	\node (pfxfy) at (5,0) {\scriptsize{$\pow{F(X)}\times \pow{F(Y)}$}};
	\node (oi) at (5.2,-0.5) {\scriptsize{$\splitpf$}};
	\node (boi) at (5.2,0.5) {\scriptsize{$\product{\eta_{F(X)}}{\eta_{F(Y)}}$}};
	
	\path[->,font=\scriptsize]
		(xy) edge node[above]{$\product{\alpha}{\beta}$} (fxfy)
		(r) edge node[below]{$\xi_w$} (pfr)
		(pfxy) edge (pfxfy)
		(fxfy) edge (pfxfy)
		(r) edge node[above]{$r$} (xy)
		(pfr) edge node[below]{$\pow{Fr}$} (pfxy);
			
\end{tikzpicture}
\end{center}
Let us then prove that
\[
	\eta_{FX}\cdot\alpha\cdot\pi_1\cdot r = \pow{F(\pi_1\cdot r)}\cdot\xi_w,
\]
the statement for $Y$ and $\beta$ being similar.
To prove this equality, since they are both morphisms from $R$ to $\pow{FX}$, it is enough to prove they correspond to 
the same relation on $FX\times R$.
First,         
\begin{center}
\begin{tabular}{rclcr}
    $\pow{F(\pi_1\cdot r)}\cdot\xi_w$ & $=$ & $\mu_{FX}\cdot\pow{(\eta_{FX}\cdot F(\pi_1\cdot r))}\cdot\xi_w$
    & & \hfill (coherence axiom)\\
    & $=$ & $\mu_{FX}\cdot\pow{\xi_{\langle F(\pi_1\cdot r),\id\rangle}}\cdot\xi_w$
    & & \hfill (*)\\
    & $=$ & $\xi_{\langle F(\pi_1\cdot r),\id\rangle;w}$
    & & \hfill (Lemma~\ref{prop:composition})
\end{tabular}
\end{center}
Here $(*)$ comes from the fact we have the following composition of pullbacks:
		\begin{center}
			\begin{tikzpicture}[scale=1.4]
		
				\node (r1r2) at (2,1) {\scriptsize{$FR$}};
				\node (r1) at (2,0) {\scriptsize{$FX\times FR$}};
				\node (r2) at (6,1) {\scriptsize{$FX$}};
				\node (y) at (6,0) {\scriptsize{$FX\times FX$}};
				\node (r2p) at (10,1) {\scriptsize{$E_{FX}$}};
				\node (yp) at (10,0) {\scriptsize{$FX\times \pow{FX}$}};
	 
				\path[->,font=\scriptsize]
					(r1r2) edge node[left]{$\langle F(\pi_1\cdot r), \id\rangle$} (r1)
					(r1r2) edge node[above]{$F(\pi_1\cdot r)$} (r2)
					(r1) edge node[below]{$\product{\id}{F(\pi_1\cdot r)}$} (y)
					(r2) edge node[left]{$\langle\id,\id\rangle$} (y)
					(r2) edge node[above]{$\theta_{FX}$} (r2p)
					(y) edge node[below]{$\product{\id}{\eta_{FX}}$} (yp)
					(r2p) edge node[right]{$\belong{FX}$} (yp);
		
				\draw (2.2,0.55) -- (2.7,0.55) -- (2.7,0.8);
				\draw (6.2,0.55) -- (6.7,0.55) -- (6.7,0.8);
			
			\end{tikzpicture}
		\end{center}
		where the left pullback is by simple computation and the right one is by definition of $\eta_{FX}$.
		Now, by definition, the composition of relations $\langle F(\pi_1\cdot r),\id\rangle;w$ is given by 
		the monic part of the (epi, mono)-factorisation of 
		\[
			\langle F(\pi_1\cdot r)\cdot\pi_1\cdot w,\pi_2\cdot w\rangle = \langle \alpha\cdot\pi_1\cdot r\cdot\pi_2\cdot w,\pi_2\cdot w\rangle.
		\]
		Since $\pi_2\cdot w$ is epi, and $\langle\alpha\cdot\pi_1\cdot r,\id\rangle$ is mono, then
		the monic part of $\langle F(\pi_1\cdot r)\cdot\pi_1\cdot w,\pi_2\cdot w\rangle$ is 
		$\langle\alpha\cdot\pi_1\cdot r,\id\rangle$, which corresponds to the 
		morphism $\eta_{FX}\cdot\alpha\cdot\pi_1\cdot r$ (similarly to $*$).

	\item Now assume we have a toposal bisimulation
	\begin{center}
\begin{tikzpicture}[scale=1.5]
		
	\node (r) at (0,0) {\scriptsize{$R$}};
	\node (xy) at (1.5,0.7) {\scriptsize{$X\times Y$}};
	\node (pfr) at (1.5,-0.7) {\scriptsize{$\pow{FR}$}};
	\node (fxfy) at (3.5,0.7) {\scriptsize{$F(X)\times F(Y)$}};
	\node (pfxy) at (3.5,-0.7) {\scriptsize{$\pow{F(X\times Y)}$}};
	\node (pfxfy) at (5,0) {\scriptsize{$\pow{F(X)}\times \pow{F(Y)}$}};
	\node (oi) at (5.2,-0.5) {\scriptsize{$\splitpf$}};
	\node (boi) at (5.2,0.5) {\scriptsize{$\product{\eta_{F(X)}}{\eta_{F(Y)}}$}};
	
	\path[->,font=\scriptsize]
		(xy) edge node[above]{$\product{\alpha}{\beta}$} (fxfy)
		(r) edge node[below]{$w$} (pfr)
		(pfxy) edge (pfxfy)
		(fxfy) edge (pfxfy)
		(r) edge node[above]{$r$} (xy)
		(pfr) edge node[below]{$\pow{Fr}$} (pfxy);
			
\end{tikzpicture}
\end{center}
	Then $w$ corresponds to a unique relation represented by a mono
	$\mono{m_w}{W}{FR\times R}$. Let us prove that this witnesses $r$ as a 
	regular AM-bisimulation, that is, that the following diagram commutes
    	\begin{center}
    \begin{tikzpicture}[scale=1.5]
    	
    	\node (frr) at (0,0) {\scriptsize{$W$}};	
    	\node (r) at (1.5,0.7) {\scriptsize{$R$}};
    	\node (xy) at (3.5,0.7) {\scriptsize{$X\times Y$}};
    	\node (pfr) at (1.5,-0.7) {\scriptsize{$FR$}};
    	\node (fxfy) at (6,0) {\scriptsize{$F(X)\times F(Y)$}};
    	\node (pfxy) at (3.5,-0.7) {\scriptsize{$F(X\times Y)$}};
    	\node (ai) at (0.4,0.55) {\scriptsize{$\pi_2\cdot m_w$}};
    	\node (bai) at (0.4,-0.55) {\scriptsize{$\pi_1\cdot m_w$}};
    	\node (oi) at (5.2,-0.6) {\scriptsize{$\splitf$}};
    	\node (boi) at (4.8,0.6) {\scriptsize{$\product{\alpha}{\beta}$}};
    	
    	\path[->,font=\scriptsize]
    		(frr) edge (r)
    		(xy) edge (fxfy)
    		(frr) edge (pfr) 
    		(pfxy) edge (fxfy)
    		(pfxy) edge (fxfy)
    		(r) edge node[above]{$r$} (xy)
    		(pfr) edge node[below]{$Fr$} (pfxy);
    			
    \end{tikzpicture}
    \end{center}
    and that $\pi_2\cdot m_w$ is epi.
    Using the same calculation as the previous point, the diagram of $r$ being a toposal bisimulation
    can be translated in terms of relations as
    \[
    	\langle F(\pi_1\cdot r),\id\rangle;m_w = \langle\alpha\cdot\pi_1\cdot r,\id\rangle \quad \text{and} \quad
	\langle F(\pi_2\cdot r),\id\rangle;m_w = \langle\beta\cdot\pi_2\cdot r,\id\rangle.
    \]
	Let's concentrate on $\alpha$ ($\beta$ will be similar).
	The composition $\langle F(\pi_1\cdot r),\id\rangle;m_w$ is again given by the monic part of 
	$\langle F(\pi_1\cdot r)\cdot\pi_1\cdot m_w,\pi_2\cdot m_w\rangle$, which is equal to 
	$\langle\alpha\cdot\pi_1\cdot r,\id\rangle$. This means that there is an epi $e$ such that 
	\[
		\langle F(\pi_1\cdot r)\cdot\pi_1\cdot m_w,\pi_2\cdot m_w\rangle = 
		\langle\alpha\cdot\pi_1\cdot r,\id\rangle\cdot e.
	\]
	Consequently, $\pi_2\cdot m_w = e$ and $\pi_2\cdot m_w$ is an epi.
	Furthermore, 
	\[
		F(\pi_1\cdot r)\cdot\pi_1\cdot m_w = 
		\alpha\cdot\pi_1\cdot r\cdot e =
		\alpha\cdot\pi_1\cdot r\cdot\pi_2\cdot m_w.\qedhere
	\]
\end{itemize}
\end{proof}
\noindent 
This nicer formulation allows us to prove a much nicer tabularity property, 
which could only be informally described for regular AM-bisimulations:
\begin{prop}
\label{prop:tobi-i-cat} 
Assume that $\CC$ is a topos and that $F$ 
covers pullbacks.
Then the following is an I-category:
	objects are coalgebras on $F$,
	morphisms are toposal AM-bisimulations,
	$\sqsubseteq$, identities, composition, and 
	$(\_)^\dagger$ are defined as in $\Rel{\CC}$.
\end{prop}
\begin{rem}
Remark that this Proposition is similar to Proposition~\ref{prop:bisim-i-cat},
without the axiom of choice and assuming only that $F$ covers pullbacks, but 
by replacing plain AM-bisimulations by toposal AM-bisimulations.
\end{rem}

\begin{proof}
	We could directly conclude this from Propositions~\ref{prop:composition-reg} and~\ref{prop:lobi-tobi}, but let us 
	show that the proof of Proposition~\ref{prop:bisim-i-cat} can be adapted more easily in the case when $F$ 
	preserves weak pullbacks.
	
	The only thing to prove is that toposal bisimulations are closed under 
	composition, without using the regular axiom of choice.
	The proof starts the same way as Proposition~\ref{prop:bisim-i-cat}. We 
	have two witnesses $\map{W_i}{R_i}{\pow{F}R_i}$ and we want to 
	construct a witness 
	$\map{W}{\rcomp{R_1}{R_2}}{\pow{F}(\rcomp{R_1}{R_2})}$.
	Since $F$ and $\pow{\!}$ preserve weak pullbacks 
	and by definition of composition, we
	have the following weak pullback and (epi, mono)-factorisation:
	\begin{center}
	\begin{tikzpicture}[scale=1.5]
			
		\node (r1r2) at (2,1) {\scriptsize{$\pow{F}(R_1\star R_2)$}};
		\node (r1) at (2,0) {\scriptsize{$\pow{F}R_1$}};
		\node (r2) at (5,1) {\scriptsize{$\pow{F}R_2$}};
		\node (y) at (5,0) {\scriptsize{$\pow{F}Y$}};
		
		\path[->,font=\scriptsize]
			(r1r2) edge node[left]{$\pow{F}\mu_1$} (r1)
			(r1r2) edge node[above]{$\pow{F}\mu_2$} (r2)
			(r1) edge node[below]{$\pow{F}(\pi_2\cdot r_1)$} (y)
			(r2) edge node[right]{$\pow{F}(\pi_1\cdot r_2)$} (y);

		\draw[dashed] (2.2,0.55) -- (2.45,0.55) -- (2.45,0.8);

		\node (r1r2) at (7,1) {\scriptsize{$R_1\star R_2$}};
		\node (xz) at (11,1) {\scriptsize{$X\times Z$}};
		\node (r1sqr2) at (9,0) {\scriptsize{$\rcomp{R_1}{R_2}$}};
		
		\path[->,font=\scriptsize]
			(r1r2) edge node[above]{$\pair{\pi_1\cdot r_1 \cdot \mu_1}{\pi_2\cdot r_2 \cdot \mu_2}$} (xz);
			
		\path[->>,font=\scriptsize]
			(r1r2) edge node[below]{$e_{\rcomp{r_1}{r_2}}$} (r1sqr2);
			
		\path[>->,font=\scriptsize]
			(r1sqr2) edge node[below]{$\rcomp{r_1}{r_2}$} (xz);
				
	\end{tikzpicture}
	\end{center}
	By the universal property of weak pullbacks, 
	we have $\map{\phi}{R_1\star R_2}{\pow{F}(R_1\star R_2)}$, such that
	\begin{center}
	\begin{tikzpicture}[scale=1.5]
			
		\node (r1r2) at (3,1) {\scriptsize{$\pow{F}(R_1\star R_2)$}};
		\node (r1) at (3,0) {\scriptsize{$\pow{F}R_1$}};
		\node (r2) at (6,1) {\scriptsize{$\pow{F}R_2$}};
		\node (y) at (6,0) {\scriptsize{$\pow{F}Y$}};
		\node (z) at (1.8,1.7) {\scriptsize{$R_1\star R_2$}};
		
		\path[->,font=\scriptsize]
			(r1r2) edge node[left]{$\pow{F}\mu_1$} (r1)
			(r1r2) edge node[above]{$\pow{F}\mu_2$} (r2)
			(r1) edge node[below]{$\pow{F}(\pi_2\cdot r_1)$} (y)
			(r2) edge node[right]{$\pow{F}(\pi_1\cdot r_2)$} (y);
			
		\path[->,font=\scriptsize, bend right =30]
			(z) edge node[left]{$W_1\cdot\mu_1$} (r1);
			
		\path[->,font=\scriptsize, bend left =20]
			(z) edge node[right]{$W_2\cdot\mu_2$} (r2);
			
		\path[->,font=\scriptsize, dotted]
			(z) edge node[above]{$\phi$} (r1r2);
				
	\end{tikzpicture}
	\end{center}
	Now $W = \pow{F}(e_{\rcomp{r_1}{r_2}})\cdot\mu_{F(R_1\star R_2)}
		\cdot\pow{\phi}\cdot(e_{\rcomp{r_1}{r_2}})^{\dagger}$ 
		is the expected witness:
	\begin{center}
	\begin{tabular}{rclcr}
		& & $\pow{F\pi_1}\cdot\pow{F(\rcomp{r_1}{r_2})}\cdot W$\\ 
		& $=$ & $\pow{F(\pi_1\cdot \rcomp{r_1}{r_2}\cdot e_{\rcomp{r_1}{r_2}})}\cdot\mu_{F(R_1\star R_2)}\cdot\pow{\phi}\cdot(e_{\rcomp{r_1}{r_2}})^{\dagger}$
		& & \hfill (definition of $W$)\\
		& $=$ & $\mu_{FX}\cdot\pow{\pow{F(\pi_1\cdot \rcomp{r_1}{r_2}\cdot e_{\rcomp{r_1}{r_2}})}}\cdot\pow{\phi}\cdot(e_{\rcomp{r_1}{r_2}})^{\dagger}$
		& & \hfill (naturality of $\mu$)\\
		& $=$ & $\mu_{FX}\cdot\pow{\pow{F(\pi_1\cdot r_1\cdot\mu_1)}}\cdot\pow{\phi}\cdot(e_{\rcomp{r_1}{r_2}})^{\dagger}$
		& & \hfill (definition of $\rcomp{r_1}{r_2}$)\\
		& $=$ & $\mu_{FX}\cdot\pow{(\pow{F(\pi_1\cdot r)}\cdot W_1\cdot\mu_1)}\cdot(e_{\rcomp{r_1}{r_2}})^{\dagger}$
		& & \hfill (definition of $\phi$)\\
		& $=$ & $\mu_{FX}\cdot\pow{(\eta_{FX}\cdot\alpha\cdot\pi_1\cdot r_1\cdot\mu_1)}\cdot(e_{\rcomp{r_1}{r_2}})^{\dagger}$
		& & \hfill (assumption on $W_1$)\\
		& $=$ & $\pow{(\alpha\cdot\pi_1\cdot r_1\cdot\mu_1)}\cdot(e_{\rcomp{r_1}{r_2}})^{\dagger}$
		& & \hfill (unit coherence axiom)\\
		& $=$ & $\pow{(\alpha\cdot\pi_1\cdot \rcomp{r_1}{r_2}\cdot e_{\rcomp{r_1}{r_2}})}\cdot(e_{\rcomp{r_1}{r_2}})^{\dagger}$
		& & \hfill (definition of $\rcomp{r_1}{r_2}$)\\
		& $=$ & $\pow{(\alpha\cdot\pi_1\cdot \rcomp{r_1}{r_2})}\cdot\eta_{\rcomp{R_1}{R_2}}$
		& & \hfill ($e_{\rcomp{r_1}{r_2}}$ is epi)\\
		& $=$ & $\eta_{FX}\cdot\alpha\cdot\pi_1\cdot \rcomp{r_1}{r_2}$
		& & \hfill (naturality of $\eta$)
	\end{tabular}
	\end{center}
	
	Similarly, we can prove that $\pow{F\pi_2}\cdot\pow{F(\rcomp{r_1}{r_2})}\cdot W = \eta_{FZ}\cdot\gamma\cdot\pi_2\cdot\rcomp{r_1}{r_2}$, which completes the proof.
	\end{proof}

Obviously, the category of maps of the I-category of toposal 
bisimulations is then not isomorphic to $\Coal{F}$, but to the category of 
$F$-coalgebras with $\pow{F}$-coalgebra homomorphisms 
between them. Then tabularity can be formulated as follows:
\begin{prop}
If $U$ is a 
toposal bisimulation from the $F$-coalgebra $\alpha$ to the 
$F$-coalgebra $\beta$, 
and if $\map{f}{Z}{X}$, $\map{g}{Z}{Y}$ is a tabulation of $U$, 
then there is a $\pow{F}$-coalgebra structure $\gamma$ on $Z$ 
such that $f$ is a 
$\pow{F}$-coalgebra homomorphism from $\gamma$ to $\eta_X\cdot\alpha$ and 
$g$ is a $\pow{F}$-coalgebra homomorphism from $\gamma$ to 
$\eta_Y\cdot\beta$.
\end{prop}

\begin{cor}
\label{coro:tobi-spans}
Assume given two coalgebras $\map{\alpha}{X}{F(X)}$ and $\map{\beta}{Y}{F(Y)}$, 
and two points $\map{p}{\ast}{X}$ and $\map{q}{\ast}{Y}$. 
the following 
two statements are equivalent:
\begin{enumerate}
	\item 
	There is a toposal bisimulation $\mono{r}{R}{X\times Y}$ from 
	$\alpha$ to $\beta$, and a point $\map{c}{\ast}{R}$ such that 
	$r\cdot c = \pair{p}{q}$ if and only if
	\item 
	there is a span $X\,\xleftarrow{~f~}\,Z\,\xrightarrow{~g~}\,Y$, 
	a $\pow{F}$-coalgebra structure 
	$\gamma$ on $Z$, 
	and a point $\map{w}{\ast}{Z}$ such that
	$f$ is a $\pow{F}$-coalgebra homomorphism
	from $\gamma$ to $\eta_X\cdot\alpha$, $g$ from $\gamma$
	to $\eta_Y\cdot\beta$, $f\cdot w = p$, and 
	$g\cdot w = q$.
\end{enumerate}
\end{cor}

\section{From Bisimulations to Simulations}
\label{sec:simulations}

In this section, we would like to extend the analysis of the previous sections
to deal with \emph{simulations}. Classically, simulations for coalgebras 
require a notion of order on morphisms of the form 
$X\,\longrightarrow\,FY$, to allow one to define that there is fewer transitions 
coming out of a state than another. This allows one to easily modify the 
definition of AM-bisimulations to obtain \emph{AM-simulations}. 
We will show that toposal bisimulations can also 
be extended to simulations in a nice way to mitigate these issues.
The only reason we chose to stay in a topos and not in a general regular category 
is because theorems have a nicer formulation there, but most of the discussion here 
can be done in a regular category.

\subsection{Order-Structure on Functors, and Lax Coalgebra Homomorphisms}
\label{sec:good-order}

We want to be able to 
compare two morphisms of the form $X\,\longrightarrow\,FY$. So, 
assuming a preorder $\leq$ on each Hom-set $\CC(X,FY)$, we can 
define \emph{lax homomorphisms of coalgebras}, as follows:
\begin{defi}
A lax homomorphism of coalgebras from $\map{\alpha}{X}{FX}$ to 
$\map{\beta}{Y}{FY}$
is a morphism $\map{f}{X}{Y}$ of $\CC$ such that 
the following 
diagram laxly commutes,
\begin{center}
\begin{tikzpicture}[scale=1.5]
		
	\node (r1r2) at (3,1) {\scriptsize{$X$}};
	\node (r1) at (8,1) {\scriptsize{$Y$}};
	\node (r2) at (3,0) {\scriptsize{$FX$}};
	\node (y) at (8,0) {\scriptsize{$FY$}};
	\node (ord) at (5.5,0.5) {\scriptsize{$\leq$}};
	
	\path[->,font=\scriptsize]
		(r1r2) edge node[above]{$f$} (r1)
		(r1r2) edge node[left]{$\alpha$} (r2)
		(r1) edge node[right]{$\beta$} (y)
		(r2) edge node[below]{$Ff$} (y);
			
\end{tikzpicture}
\end{center}
meaning that 
$Ff\cdot\alpha\leq\beta\cdot f$ in $\CC(X,FY)$.
\end{defi}
Unfortunately, coalgebras and lax homomorphisms of coalgebras do not form a 
category in general, and some axioms are required for the 
interaction of $\leq$ with the composition.
\begin{defi}
A \emph{good order structure on $F$} is a preorder $\leq$ on each 
Hom-set 
of the form 
$\CC(X,FY)$ such that:
\begin{enumerate}
	\item 
	if $\alpha \leq \beta$ in $\CC(X,FY)$, $\map{f}{X'}{X}$, and
	$\map{g}{Y}{Y'}$, then $Fg\cdot\alpha\cdot f \leq Fg\cdot\beta\cdot f$ 
	in $\CC(X',FY')$;
	\item 
	if $\map{h}{X}{FZ}$, $\map{k}{X}{FY}$, $\map{g}{Y}{Z}$, and
	$h \leq Fg\cdot k$ in $\CC(X,FZ)$, then there is $\map{k'}{X}{FY}$ 
	such that $k' \leq k$ in $\CC(X,FY)$ and $h = Fg\cdot k'$.
\end{enumerate}
\end{defi}
\begin{lem}
When $\leq$ is a good order structure on $F$, then coalgebras and lax 
homomorphisms of coalgebras form a category, denoted by $\laxCoal{F}$.
\end{lem}
\begin{exa}
When $F$ is the functor modelling non-deterministic labelled transition systems and 
$\leq$ is given by point-wise inclusion,
lax homomorphisms of coalgebras are exactly morphisms 
of systems in the sense of \cite{joyal96}. 
Those 
morphisms are intuitively morphisms whose graphs are simulations.
More generally, we will see that lax homomorphisms are simulation maps.
In this picture, it can be proved in some cases that coalgebra homomorphisms
are precisely open maps, that is, lax homomorphisms with 
some lifting properties (see~\cite{wissmann19}, from which the notion
of good order is adapted).
\end{exa}

\subsection{AM-Simulations}
\label{sec:am-sim}

\begin{defi}
We say that a relation is an \emph{AM-simulation} from the coalgebra 
$\map{\alpha}{X}{FX}$ to $\map{\beta}{Y}{FY}$, if for any mono 
$\mono{r}{R}{X\times Y}$ representing it, there is a morphism 
$\map{W}{R}{FR}$ such that:
\begin{center}
\begin{tikzpicture}[scale=2]
		
	\node (r) at (0,0) {\scriptsize{$R$}};
	\node (xy) at (3,0.5) {\scriptsize{$X\times Y$}};
	\node (cr) at (3,-0.5) {\scriptsize{$FR$}};
	\node (fxfy) at (6,0.5) {\scriptsize{$F(X)\times F(Y)$}};
	\node (fxy) at (6,-0.5) {\scriptsize{$F(X\times Y)$}};
	\node (ord) at (4.5,0) {\scriptsize{$\leq\times\geq$}};
	
	\path[->,font=\scriptsize]
		(xy) edge node[above]{$\product{\alpha}{\beta}$} (fxfy)
		(r) edge node[below]{$W$} (cr)
		(fxy) edge node[right]{$\splitf$} (fxfy)
		(r) edge node[above]{$r$} (xy)
		(cr) edge node[below]{$Fr$} (fxy);
			
\end{tikzpicture}
\end{center}
meaning that 
$
	\alpha\cdot\pi_1\cdot r \leq F\pi_1\cdot Fr\cdot W$  
	and 
	$\beta\cdot\pi_2\cdot r \geq F\pi_2\cdot Fr\cdot W.
$
\end{defi}
The definition can be simplified:
\begin{prop}
\label{prop:ord-sim}
When $\leq$ is a good order structure, it is equivalent to require that 
the left inequality is actually an equality 
$\alpha\cdot\pi_1\cdot r = F\pi_1\cdot Fr\cdot W.$
\end{prop}

\begin{proof}
We start with $W$ such that 
\[\alpha\cdot\pi_1\cdot r \leq F\pi_1\cdot Fr\cdot W  
	\quad\text{and}\quad 
	\beta\cdot\pi_2\cdot r \geq F\pi_2\cdot Fr\cdot W.\]
Use the second assumption of a good order structure with 
$h = \alpha\cdot\pi_1\cdot r$, $g = \pi_1\cdot r$ and $k = W$. 
We then obtain $W' \leq W$ with 
\[\alpha\cdot\pi_1\cdot r = F\pi_1\cdot Fr\cdot W'.\] 
Then since composition is monotone, 
\[\beta\cdot\pi_2\cdot r \geq F\pi_2\cdot Fr\cdot W \geq F\pi_2\cdot Fr\cdot W'.\qedhere\]
\end{proof}

\begin{exa}
When $F~:~X\,\mapsto\,\pow{(\Sigma\times X)}$, AM-simulations 
correspond to strong simulations. The left part of the commutativity means
that for every $(x,y) \in R$ and $(a,x') \in \alpha(x)$, there is $y'$ such that 
$(a, (x',y')) \in W(x,y)$. The right part then implies that necessarily 
$(a,y') \in \beta(y)$.
\end{exa}

Much as in the case of AM-bisimulations, diagonals 
(and actually all AM-bisimulations) are AM-simulations, and 
AM-simulations are closed under composition only 
under some conditions. However, they are not closed under converse. 
These observations can be encompassed as follows:
\begin{prop}
\label{prop:am-2-cat}
When $\CC$ has the regular axiom of choice and $F$ preserves weak pullbacks, then the following is a locally ordered 2-category:
\begin{itemize}
	\item 
	objects are $F$-coalgebras,
	\item 
	morphisms are AM-simulations,
	\item 
	identitites, compositions, and $\sqsubseteq$ are given 
	by $\Rel{\CC}$.
\end{itemize}
We denote this category by $\Sim{F}$.
\end{prop}
We can formalise the relationship between lax coalgebra homomorphisms 
and simulation maps:
\begin{thm}
Maps in $\Rel{\CC}$ that are AM-simulations are precisely lax 
homomorphisms of coalgebras.
\end{thm}
\noindent 
Note that this theorem cannot have a form as nice as Theorem~\ref{th:maps-coalgebra} 
because AM-simulations are not closed under converse, and the 
right adjoint of a map has to be its converse.
At this point, we can also describe the tabulations of AM-simulations:
\begin{prop}
If $U$ is an AM-simulation from $\alpha$ to $\beta$, and if 
$\map{f}{Z}{X}$, $\map{g}{Z}{Y}$ is a tabulation of $U$ then, 
there is a coalgebra structure $\gamma$ on $Z$ such that $f$ is a 
coalgebra homomorphism from $\gamma$ to $\alpha$ and 
$g$ is a lax coalgebra homomorphism from $\gamma$ to $\beta$.
\end{prop}
\begin{cor}
Assume $\CC$ has the regular axiom of choice.
Assume given two coalgebras $\map{\alpha}{X}{F(X)}$ and $\map{\beta}{Y}{F(Y)}$, 
and two points $\map{p}{\ast}{X}$ and $\map{q}{\ast}{Y}$. 
The following 
two statements are equivalent:
\begin{enumerate}
	\item 
	There is an AM-simulation $\mono{r}{R}{X\times Y}$ from 
	$\alpha$ to $\beta$, and a point $\map{c}{\ast}{R}$ with 
	$r\cdot c = \pair{p}{q}$.
	\item 
	There is a span $X\,\xleftarrow{~f~}\,Z\,\xrightarrow{~g~}\,Y$, 
	an $F$-coalgebra structure $\gamma$ 
	on $Z$ such that $f$ is a coalgebra homomorphism from $\gamma$ to 
	$\alpha$ and $g$ is a lax coalgebra homomorphism from 
	$\gamma$ to $\beta$, 
	and a point $\map{w}{\ast}{Z}$ such that $f\cdot w = p$ and 
	$g\cdot w = q$.
\end{enumerate}
\end{cor}
\noindent 
This formalises some observations that simulations are spans of a 
bisimulations map and a simulation map (see \cite{tabuada04} for examples 
of this fact in the context of open maps).

\subsection{Extending the Order-Structure}

In Section~\ref{sec:good-order}, we started by assuming a relation $\leq$ on the 
Hom-sets of the form $\CC(X,FY)$ satisfying some properties. This 
good order structure was necessary to prove the properties of 
Section~\ref{sec:am-sim}. In the coming section, we will pass again from plain to 
toposal, by considering $F$-coalgebras as $\pow{F}$-coalgebras. It is then 
necessary to extend good order structures on $F$ to good order structures 
on $\pow{F}$. 

Assume a relation $\leq$ is given on all Hom-sets of the form $\CC(X,FY)$.
We define $\leqp$ on $\CC(X,\pow{F}Y)$ as follows. A 
morphism $\map{f}{X}{\pow{F}Y}$ uniquely (up to isos) corresponds to a 
mono of the form $\map{m_f}{U_f}{FY\times X}$ by definition of $\pow{\!}$.
Then, given two morphisms $\map{f,g}{X}{\pow{F}Y}$, $f \leqp g$ if 
there exist a morphism $\map{u}{Z}{U_g}$ and an epi $\epi{e}{Z}{U_f}$ 
such that:
$
	\pi_1\cdot m_f \cdot e \leq \pi_1\cdot m_g \cdot u$ 
	and
	$\pi_2\cdot m_f \cdot e = \pi_2\cdot m_g \cdot u.
$
\begin{exa}
The order $\leqp$ might appear complicated, but it can be interpreted easily in 
$\Set$, especially when the order structure on $\CC(X,FY)$ is a point-wise order, 
assuming that $FY$ itself is preordered. 
Indeed, given two functions $\map{f,g}{X}{\pow{F}Y}$, $f\leqp g$ if and 
only if for every $x \in X$, and every $a \in f(x) \subseteq FY$ there is 
$b \in g(x)$ such that $a \leq b$ in $F(Y)$.
\end{exa}

To make it consistent with the previous section, we show that this 
preserves goodness:
\begin{prop}
\label{prop:ord-pow}
$\leqp$ is a good order structure if $\leq$ is.
\end{prop}

\begin{proof}
Let us prove that $\leqp$ is a good order structure on $\pow{F}$.
\begin{itemize}
	\item \textbf{$\leqp$ is a preorder.}
	\begin{itemize}
		\item \textbf{reflexivity.} To prove $f \leqp f$, 
		take $u = e = \id$.
		\item \textbf{transitivity}: Assume $f\leqp g \leqp h$. 
		So there are a morphism $\map{u}{Z}{U_g}$, 
		and an epi $\epi{e}{Z}{U_f}$ such that:
\begin{itemize}
	\item $\pi_1\cdot m_f \cdot e \leq \pi_1\cdot m_g \cdot u$,
	\item $\pi_2\cdot m_f \cdot e = \pi_2\cdot m_g \cdot u$.
\end{itemize}
and there a morphism $\map{u'}{Z'}{U_h}$ and an epi 
$\epi{e'}{Z'}{U_g}$ such that:
\begin{itemize}
	\item $\pi_1\cdot m_g \cdot e' \leq \pi_1\cdot m_h \cdot u'$,
	\item $\pi_2\cdot m_g \cdot e' = \pi_2\cdot m_h \cdot u'$.
\end{itemize}
Form the following pullback:
\begin{center}
\begin{tikzpicture}[scale=1.5]
		
	\node (r1r2) at (3,1) {\scriptsize{$Z''$}};
	\node (r1) at (3,0) {\scriptsize{$Z$}};
	\node (r2) at (6,1) {\scriptsize{$Z'$}};
	\node (y) at (6,0) {\scriptsize{$U_g$}};
	
	\path[->,font=\scriptsize]
		(r1r2) edge node[left]{$e''$} (r1)
		(r1r2) edge node[above]{$u''$} (r2)
		(r1) edge node[below]{$u$} (y)
		(r2) edge node[right]{$e'$} (y);
		
	\draw (3.2,0.55) -- (3.7,0.55) -- (3.7,0.8);
			
\end{tikzpicture}
\end{center}
Since in a topos, epis are closed under pullbacks, $e''$ is an epi, 
and so is $e\cdot e''$. So then, $Z''$, $u'\cdot u''$ and $e\cdot e''$ 
witness the fact that $f \leqp h$.
	\end{itemize} 
	\item \textbf{Composition is monotone.} 
	Assume $\map{f \leqp f'\!\!}{\!X}{\pow{F(Y)}}$,
	with witnesses $Z$, $u$ and $e$
	\begin{itemize}
		\item \textbf{composition to the left}: 
		Assume $\map{h}{Y}{Y'}$. 
		Then $\pow{F(h)}\cdot f$ corresponds to the mono part
		of the following (epi, mono)-factorisation:
\begin{center}
\begin{tikzpicture}[scale=1.5]
		
	\node (r1r2) at (0,0) {\scriptsize{$U_f$}};
	\node (xz) at (5,0) {\scriptsize{$FY'\times X$}};
	\node (r1sqr2) at (2.5,-0.75) {\scriptsize{$S_f$}};
	
	\path[->,font=\scriptsize]
		(r1r2) edge node[above]{$\product{Fh}{\id}\cdot m_f$} (xz);
		
	\path[->>,font=\scriptsize]
		(r1r2) edge node[below]{$e_f$} (r1sqr2);
		
	\path[>->,font=\scriptsize]
		(r1sqr2) edge node[below]{\quad\quad$m_{\pow{F(h)}\cdot f}$} (xz);
			
\end{tikzpicture}
\end{center}
Same for $\pow{F(h)}\cdot f'$. Then $Z$, $e_{f'}\cdot u$ and 
$e_f\cdot e$ is a witness of the fact that 
$\pow{F(h)}\cdot f \leqp \pow{F(h)}\cdot f'$.
		\item \textbf{composition to the right}: 
		Assume given $\map{g}{X'}{X}$. 
		Then $f\cdot g$ corresponds to the relation represented by 
		$m_{f\cdot g}$, given by the following composition of pullbacks:
		\begin{center}
			\begin{tikzpicture}[scale=1.5]
		
				\node (r1r2) at (3,1) {\scriptsize{$U_{f\cdot g}$}};
				\node (r1) at (3,0) {\scriptsize{$FY\times X'$}};
				\node (r2) at (6,1) {\scriptsize{$U_f$}};
				\node (y) at (6,0) {\scriptsize{$FY\times X$}};
				\node (r2p) at (9,1) {\scriptsize{$E_{FY}$}};
				\node (yp) at (9,0) {\scriptsize{$FY\times \pow{FY}$}};
	
				\path[->,font=\scriptsize]
					(r1r2) edge node[left]{$m_{f\cdot g}$} (r1)
					(r1r2) edge node[above]{$\tau_f$} (r2)
					(r1) edge node[below]{$\product{\id}{g}$} (y)
					(r2) edge node[left]{$m_f$} (y)
					(r2) edge node[above]{$\theta_f$} (r2p)
					(y) edge node[below]{$\product{\id}{f}$} (yp)
					(r2p) edge node[right]{$\belong{FY}$} (yp);
		
				\draw (3.2,0.55) -- (3.7,0.55) -- (3.7,0.8);
				\draw (6.2,0.55) -- (6.7,0.55) -- (6.7,0.8);
			
			\end{tikzpicture}
		\end{center}
		for some $\tau_f$, $\theta_f$. Same for $m_{f'\cdot g}$. 
		Form the following pullback:
		\begin{center}
			\begin{tikzpicture}[scale=1.5]
		
				\node (r1r2) at (3,1) {\scriptsize{$Z_g$}};
				\node (r1) at (3,0) {\scriptsize{$X'$}};
				\node (r2) at (6,1) {\scriptsize{$Z$}};
				\node (y) at (6,0) {\scriptsize{$X$}};
	
				\path[->,font=\scriptsize]
					(r1r2) edge node[left]{$\alpha$} (r1)
					(r1r2) edge node[above]{$\beta$} (r2)
					(r1) edge node[below]{$g$} (y)
					(r2) edge node[right]{$\pi_2\cdot m_f\cdot e = \pi_2\cdot m_{f'}\cdot u$} (y);
					
				\draw (3.2,0.55) -- (3.7,0.55) -- (3.7,0.8);
			
			\end{tikzpicture}
		\end{center}
		So then, we have that 
		$m_f\cdot e\cdot\beta = \product{\id}{g}\cdot\pair{\pi_1\cdot m_f\cdot e\cdot\beta}{\alpha}$, 
		and by the universal property of $U_{f\cdot g}$, there is a unique
		morphism $\map{e'}{Z_g}{U_{f\cdot g}}$ such that 
		\[\tau_f \cdot e' = e\cdot\beta 
			\quad\text{and}\quad 
			m_{f\cdot g}\cdot e' = \pair{\pi_1\cdot m_f\cdot e\cdot\beta}{\alpha}.\]
		Similarly, there is $\map{u'}{Z_g}{U_{f'\cdot g}}$ such that 
		\[\tau_{f'} \cdot u' = u\cdot\beta 
		\quad\text{and}\quad 
		m_{f\cdot g}\cdot u' = \pair{\pi_1\cdot m_{f'}\cdot u\cdot\beta}{\alpha}.\]
		The only interesting part in proving that $Z_g$, $u'$ and 
		$e'$ is a witness of the fact that $f\cdot g \leqp f'\cdot g$ 
		is the fact that $e'$ is an epi. 
		For that, it is enough to observe that:
		\begin{center}
			\begin{tikzpicture}[scale=1.5]
		
				\node (r1r2) at (3,1) {\scriptsize{$Z_g$}};
				\node (r1) at (3,0) {\scriptsize{$U_{f\cdot g}$}};
				\node (r2) at (6,1) {\scriptsize{$Z$}};
				\node (y) at (6,0) {\scriptsize{$U_f$}};
	
				\path[->,font=\scriptsize]
					(r1r2) edge node[left]{$e'$} (r1)
					(r1r2) edge node[above]{$\beta$} (r2)
					(r1) edge node[below]{$\tau_f$} (y)
					(r2) edge node[right]{$e$} (y);
			
			\end{tikzpicture}
		\end{center}
		is a pullback square, and to use the fact that in a topos, 
		epis are closed under pullback.
	\end{itemize}
	\item \textbf{Last axiom of good order structure.} 
	Assume $h \leqp \pow{F(g)}\cdot k$, with 
	$\map{h}{X}{\pow{FZ}}$, $\map{k}{X}{\pow{FY}}$ and 
	$\map{g}{Y}{Z}$. So we have a morphism 
	$\map{u}{S}{U_{\pow{F(h)}\cdot k}}$ and an epi 
	$\map{e}{S}{U_h}$ such that 
	\[\pi_1\cdot m_h \cdot e \leq \pi_1\cdot m_{\pow{F(h)}\cdot k} \cdot u
	\quad\text{and}\quad 
	\pi_2\cdot m_h \cdot e = \pi_2\cdot m_{\pow{F(h)}\cdot k} \cdot u.\]
	$m_{\pow{F(h)}\cdot k}$ is obtained by the following factorisation:
	\begin{center}
		\begin{tikzpicture}[scale=1.5]
		
			\node (r1r2) at (0,0) {\scriptsize{$U_k$}};
			\node (xz) at (5,0) {\scriptsize{$FZ\times X$}};
			\node (r1sqr2) at (2.5,-0.74) {\scriptsize{$U_{\pow{F(h)}\cdot k}$}};
	
			\path[->,font=\scriptsize]
				(r1r2) edge node[above]{$\product{Fg}{\id}\cdot m_k$} (xz);
						
			\path[->>,font=\scriptsize]
				(r1r2) edge node[below]{$e'$} (r1sqr2);
		
			\path[>->,font=\scriptsize]
				(r1sqr2) edge node[below]{$m_{\pow{F(h)}\cdot k}$} (xz);
			
		\end{tikzpicture}
	\end{center}
	Form the following pullback:
	\begin{center}
		\begin{tikzpicture}[scale=1.5]
		
			\node (r1r2) at (3,1) {\scriptsize{$T$}};
			\node (r1) at (3,0) {\scriptsize{$S$}};
			\node (r2) at (6,1) {\scriptsize{$U_k$}};
			\node (y) at (6,0) {\scriptsize{$U_{\pow{F(h)}\cdot k}$}};
	
			\path[->,font=\scriptsize]
				(r1r2) edge node[left]{$e''$} (r1)
				(r1r2) edge node[above]{$v$} (r2)
				(r1) edge node[below]{$u$} (y)
				(r2) edge node[right]{$e'$} (y);
					
			\draw (3.2,0.55) -- (3.7,0.55) -- (3.7,0.8);
			
		\end{tikzpicture}
	\end{center}
	Since epis are closed under pullbacks in a topos, $e''$ is an epi, 
	and we have 
	\[\pi_1\cdot m_h\cdot e\cdot e'' \leq Fg\cdot\pi_1\cdot m_k\cdot v.\]
	By using the fact that $\leq$ is a good order structure, 
	we obtain $w \leq \pi_1\cdot m_k\cdot v$, such that 
	$Fg \cdot w = \pi_1\cdot m_h\cdot e\cdot e''$, and then 
	$m_h\cdot e\cdot e' = \product{Fg}{\id}\cdot\pair{w}{\pi_2\cdot m_h\cdot e\cdot e''}$. 
	Consider the following factorisation:
	\begin{center}
		\begin{tikzpicture}[scale=1.5]
	
			\node (r1r2) at (0,0) {\scriptsize{$T$}};
			\node (xz) at (5,0) {\scriptsize{$FY\times X$}};
			\node (r1sqr2) at (2.5,-0.75) {\scriptsize{$U$}};
	
			\path[->,font=\scriptsize]
				(r1r2) edge node[above]{$\pair{w}{\pi_2\cdot m_h\cdot e\cdot e''}$} (xz);
						
			\path[->>,font=\scriptsize]
				(r1r2) edge node[below]{$\rho$} (r1sqr2);
		
			\path[>->,font=\scriptsize]
				(r1sqr2) edge node[below]{$m$} (xz);
			
		\end{tikzpicture}
	\end{center}
	and define $k'$ as $\xi_m$. Then $m_h$, 
	which is the monic part of $m_h\cdot\rho\cdot\rho'$, is also the monic
	part of $\product{Fg}{\id}\cdot m$. But $m_{\pow{Fg}\cdot k'}$ is also 
	the monic part of $\product{Fg}{\id}\cdot m$, so by unicity of the 
	(epi, mono)-factorisation, $m_{\pow{Fg}\cdot k'} \equiv m_h$, which 
	means that $\pow{Fg}\cdot k' = h$. 
	Furthermore, $T$, $v$ and $\rho$ is a witness of $k' \leqp k$.\qedhere
\end{itemize}
\end{proof}

\subsection{Toposal AM-Simulations}

With all those ingredients, we can easily deduce the right notion of 
\emph{AM toposal-simulations}:
\begin{defi}
We say that a relation is a \emph{toposal AM-simulation} from the coalgebra 
$\map{\alpha}{X}{FX}$ to $\map{\beta}{Y}{FY}$, if for any mono 
$\mono{r}{R}{X\times Y}$ representing it, there is a morphism 
$\map{W}{R}{\pow{FR}}$ such that:
\begin{center}
\begin{tikzpicture}[scale=2]
		
	\node (r) at (-1,0) {\scriptsize{$R$}};
	\node (xy) at (0.5,0.7) {\scriptsize{$X\times Y$}};
	\node (pfr) at (0.5,-0.7) {\scriptsize{$\pow{FR}$}};
	\node (fxfy) at (3.5,0.7) {\scriptsize{$F(X)\times F(Y)$}};
	\node (pfxy) at (3.5,-0.7) {\scriptsize{$\pow{F(X\times Y)}$}};
	\node (pfxfy) at (5,0) {\scriptsize{$\pow{F(X)}\times \pow{F(Y)}$}};
	\node (oi) at (5.2,-0.5) {\scriptsize{$\splitpf$}};
	\node (boi) at (5.2,0.5) {\scriptsize{$\product{\eta_{F(X)}}{\eta_{F(Y)}}$}};
	\node (ord) at (2,0) {\scriptsize{$\leqp\times\geqp$}};
	
	\path[->,font=\scriptsize]
		(xy) edge node[above]{$\product{\alpha}{\beta}$} (fxfy)
		(r) edge node[below]{$W$} (pfr)
		(pfxy) edge (pfxfy)
		(fxfy) edge (pfxfy)
		(r) edge node[above]{$r$} (xy)
		(pfr) edge node[below]{$\pow{Fr}$} (pfxy);
			
\end{tikzpicture}
\end{center}
\end{defi}
Plain and toposal AM-simulations also coincide 
under the axiom of choice:
\begin{prop}
\label{prop:losi-tosi}
Assume that $\CC$ has the regular axiom of choice. 
Then for every relation $U$ from 
$X$ to $Y$, 
every coalgebra $\map{\alpha}{X}{FX}$ and $\map{\beta}{Y}{FY}$, 
$U$ is an AM-simulation from $\alpha$ to $\beta$ if and only 
if it is a toposal AM-simulation between them. 
\end{prop}
\noindent 
The proof of this Proposition 
relies on the following lemma, relating 
the regular axiom of choice and picking elements in a power-object:
\begin{lem}
\label{lem:choice-pow}
Assume that every epi is split and assume given $F$ with a good order 
structure $\leq$.
Assume also given a square:
\begin{center}
\begin{tikzpicture}[scale=1.5]
		
	\node (r1r2) at (3,1) {\scriptsize{$X$}};
	\node (r1) at (3,0) {\scriptsize{$\pow{Y}$}};
	\node (r2) at (8,1) {\scriptsize{$\prod\limits_{i\in I} Z_i$}};
	\node (y) at (8,0) {\scriptsize{$\prod\limits_{i\in I} \pow{Z_i}$}};
	
	\path[->,font=\scriptsize]
		(r1r2) edge node[left]{$f$} (r1)
		(r1r2) edge node[above]{$\langle h_i \mid i\in I\rangle$} (r2)
		(r1) edge node[below]{$\langle\pow{g_i} \mid i\in I\rangle$} (y)
		(r2) edge node[right]{$[\eta_{Z_i} \mid i \in I]$} (y);
			
\end{tikzpicture}
\end{center}
such that $I$ is finite and:
\begin{itemize}
	\item for all $i \in I$, either:
		\begin{itemize}
			\item $\pow{g_i}\cdot f = \eta_{Z_i}\cdot h_i$, or
			\item $g_i$ is of the form $Fg_i'$ and 
			$\pow{g_i}\cdot f \leqp \eta_{Z_i}\cdot h_i$,
		\end{itemize}
	\item there is $i_0 \in I$ satisfying the first case of the first point.
\end{itemize} 
Then there is $\map{f'}{X}{Y}$ such that:
\begin{itemize}
	\item if $\mono{m}{R}{Y\times X}$ is the relation corresponding to 
	$f$, then 
	$\pair{f'}{\id} \sqsubseteq m$,
	\item for all $i$ satisfying the first case, 
	$g_i \cdot f' = h_i$, and
	\item for all $i$ satisfying the second case, 
	$g_i \cdot f' \leq h_i$.
\end{itemize}
\end{lem}
\begin{proof}
The conclusion $\pair{f'}{\id} \sqsubseteq m$ means that we are looking 
for $f'$ such that there is $\map{u}{X}{R}$ such that 
$\pair{f'}{\id} = m\cdot u$, that is:
\[
	f' = \pi_1\cdot m\cdot u \quad\text{and}\quad \id = \pi_2\cdot m\cdot u.
\]
To obtain this $u$ is then enough to prove that $\pi_2\cdot m$ is an epi and 
using the regular axiom of choice. Observe that we have the following:
\begin{center}
\begin{tabular}{rclcr}
    $\xi_{\pair{h_i}{\id}}$ & $=$ & $\eta_{Z_i}\cdot h_i$
    & & \hfill (definition of $\xi_{\pair{h_i}{\id}}$)\\
    & $\simeq_i$ & $\pow{g_i}\cdot f$
    & & \hfill (assumption)\\
    & $=$ & $\mu_{Z_i}\cdot\pow{\eta_{Z_i}}\cdot\pow{g_i}\cdot f$
    & & \hfill (unit coherence axiom)\\
    & $=$ & $\mu_{Z_i}\cdot\pow{\eta_{Z_i}\cdot g_i}\cdot f$
    & & \hfill ($\text{Pow}$ is a functor)\\
    & $=$ & $\mu_{Z_i}\cdot\pow{\xi_{\pair{g_i}{\id}}}\cdot \xi_{m}$
    & & \hfill (definition of $m$ and $\xi_{\pair{g_i}{\id}}$)\\
    & $=$ & $\xi_{\rcomp{\pair{g_i}{\id}}{m}}$
    & & \hfill (Proposition~\ref{prop:composition})
\end{tabular}
\end{center}
where $\simeq_i$ is either $=$ if $i$ satisfies the first case, or $\geqp$ 
otherwise.

This means that for all $i$ satisfying the first case (and there is at least 
one), $\pair{h_i}{\id} \equiv \rcomp{\pair{g_i}{\id}}{m}$, which means that 
there is an iso $v_i$ such that 
$\rcomp{\pair{g_i}{\id}}{m} = \pair{h_i}{\id}\cdot v_i$. Unfolding the 
definition of the composition, there is an epi $e_i$ such that:
\[
	\rcomp{\pair{g_i}{\id}}{m}\cdot e_i = \pair{g\cdot\pi_1\cdot m}{\pi_2\cdot r}.
\]
In total, $\pi_2\cdot m = v_i\cdot e_i$ which is an epi.

It then remains to prove that $f' = \pi_1\cdot m\cdot u$ satisfies all the statements in the
conclusion. The first one is by construction.
Now, for $i$ satisfying the first case, 
\begin{center}
\begin{tabular}{rclcr}
    $g_i \cdot \pi_1\cdot m\cdot u$ & $=$ & $\pi_1\cdot(\rcomp{\pair{g_i}{\id}}{m})\cdot e_i\cdot u$
    & & \hfill (definition of $\rcomp{\pair{g_i}{\id}}{m}$)\\
    & $=$ & $\pi_1\cdot\pair{h_i}{\id}\cdot v_i\cdot e_i\cdot u$
    & & \hfill (definition of $v_i$)\\
    & $=$ & $h_i\cdot v_i\cdot e_i\cdot u$
    & & \hfill (computation on products)\\
    & $=$ & $h_i\cdot \pi_2\cdot m\cdot u$
    & & \hfill (see previously)\\
    & $=$ & $h_i$
    & & \hfill (definition of $u$)
\end{tabular}
\end{center}
Now, for $i$ in the second case, we have proved that 
$\xi_{\rcomp{\pair{g_i}{\id}}{m}} \leqp \xi_{\pair{h_i}{\id}}$. By definition, 
this means that there is an epi $\rho_i$ and a morphism $u_i$ such that:
\[
	\pi_1\cdot\rcomp{\pair{g_i}{\id}}{m}\cdot\rho \leq 
		h_i\cdot u_i \quad\text{and}\quad
	\pi_2\cdot\rcomp{\pair{g_i}{\id}}{m}\cdot\rho = u_i.
\]
Since every is split and $\leq$ is a good order structure, this implies that:
\[
	\pi_1\cdot\rcomp{\pair{g_i}{\id}}{m} \leq 
		h_i\cdot\pi_2\cdot\rcomp{\pair{g_i}{\id}}{m},
\]
from which it is easy to deduce that $g_i \cdot \pi_1\cdot m\cdot u \leq h_i$.
\end{proof}
\begin{rem}
Actually, the converse of the previous lemma also holds. Assume given an 
epi $\epi{e}{X}{Y}$. Since $e$ is an epi, the following diagram 
commutes:
\begin{center}
\begin{tikzpicture}[scale=1.5]
		
	\node (r1r2) at (3,1) {\scriptsize{$Y$}};
	\node (r1) at (3,0) {\scriptsize{$\pow{X}$}};
	\node (r2) at (8,1) {\scriptsize{$\ast$}};
	\node (y) at (8,0) {\scriptsize{$\pow{\ast}$}};
	
	\path[->,font=\scriptsize]
		(r1r2) edge node[left]{$e^\dagger$} (r1)
		(r1r2) edge node[above]{$!_Y$} (r2)
		(r1) edge node[below]{$\pow{!}$} (y)
		(r2) edge node[right]{$\eta_\ast$} (y);
			
\end{tikzpicture}
\end{center}
Then using the conclusion of the previous lemma, we obtain 
$\map{s}{Y}{X}$ such that $\pair{s}{\id} \sqsubseteq \pair{\id}{e}$.
This means that there is a mono $u$ such that $s = u$ and $\id = e\cdot u$, 
and $e$ is split.
\end{rem}

An additional argument is needed, namely:
\begin{lem}
If $f \leq g$ then $\eta\cdot f \leqp \eta\cdot g$.
\end{lem}
\begin{proof}
$\eta\cdot f$ (resp. $\eta\cdot g$) corresponds to the mono $\pair{f}{\id}$ 
(resp. $\pair{g}{\id}$). So we have:
\[
	\pi_1\cdot\pair{f}{\id} = f \leq g = \pi_1\cdot\pair{g}{\id}
\]
and
\[
	\pi_2\cdot\pair{f}{\id} = \id = \pi_2\cdot\pair{g}{\id},
\]
which means that $\eta\cdot f \leqp \eta\cdot g$.
\end{proof}

\begin{proof}[Proof of Proposition~\ref{prop:losi-tosi}]
Let us prove both implications:
\begin{itemize}
	\item If $U$ is AM-simulation and $\map{W}{R}{FR}$ is 
	a witness, then $\map{\eta_{FR}\cdot W}{R}{FR}$ is 
	a toposal witness by naturality of $\eta$ and the previous Lemma.
	\item If $U$ is a toposal AM-simulation, then we obtain a  
	witness to prove it is a plain AM-simulation directly by 
	Lemma~\ref{lem:choice-pow}.\qedhere
\end{itemize}
\end{proof}
\noindent 
Finally, we can prove the closure under composition and the 
characterisation with spans without the axiom of choice:
\begin{prop}
Proposition~\ref{prop:am-2-cat} holds without regular axiom of choice 
when replacing AM-simulations by toposal AM-simulations.
\end{prop}
\begin{thm}
Assume given two coalgebras $\map{\alpha}{X}{F(X)}$ and $\map{\beta}{Y}{F(Y)}$, 
and two points $\map{p}{\ast}{X}$ and $\map{q}{\ast}{Y}$. 
	There is a toposal AM-simulation $\mono{r}{R}{X\times Y}$ from 
	$\alpha$ to $\beta$, and a point $\map{c}{\ast}{R}$ such that 
	$r\cdot c = \pair{p}{q}$ if and only if
	there is a span $X\,\xleftarrow{~f~}\,Z\,\xrightarrow{~g~}\,Y$, 
	a $\pow{F}$-coalgebra structure 
	$\gamma$ on $Z$ such that $f$ is a $\pow{F}$-coalgebra homomorphism 
	from $\gamma$ to $\eta_X\cdot\alpha$ and $g$ a lax 
	$\pow{F}$-coalgebra homomorphism from $\gamma$ to 
	$\eta_Y\cdot\beta$,
	and a point $\map{w}{\ast}{Z}$ such that $f\cdot w = p$ and 
	$g\cdot w = q$.
\end{thm}

\section{Examples}
\label{sec:examples}

In this section, let us develop some examples in different regular categories.

\subsection{Vietoris Bisimulations}

In~\cite{bezhanishvili10}, the authors study bisimulations for the Vietoris functor, which maps 
a topological space to its set of closed subspaces equipped with a suitable topology, in the 
category $\Stone$ of Stone spaces and continuous functions. 
More specifically, they show that so-called descriptive models coincide with coalgebras of the form
$X\rightarrow\viet{X}\times A$ where $\mathcal{V}$ is the Vietoris functor and $A$ is some fixed Stone space 
(such as a finite set of sets of propositions equipped with the discrete topology).
They are interested in 
describing relation liftings (similar to those defining HJ-bisimulations) that 
coincide with behavioural equivalences. They actually proved that in this 
case AM-bisimilarity does not coincide with
behavioural equivalence. The main reason for discrepancy is that the Vietoris 
functor does not preserves weak pullbacks.
In \cite{staton11}, Staton proved that the Vietoris functor is a so-called 
$\mathcal{S}$-powerset functor, and that it, in particular, covers pullbacks.
Combining this with the (well-known) fact that the category of Stone spaces 
is regular and has pushouts, Theorem~\ref{theo:reg-equivalences} 
holds in this case, and all three notions-regular AM-bisimulations, HJ-bisimulations, 
and behavioural equivalences-coincide.

\begin{center}
\begin{tikzpicture}[scale=1.05]

	\draw [dashed] (5,-1.5) -- (5,7.5);

	\node (0) at (-1.5,0) {$0$};
	\node (1) at (-1.5,1) {$1$};
	\node (2) at (-1.5,2) {$2$};
	\node [rotate=90] (i) at (-1.5,3) {$\cdots$};
	\node [rotate=90] (i11) at (0,3) {$\cdots$};
	\node (n) at (-1.5,4) {$2n$};
	\node (p) at (-1.5,5) {$2n+1$};
	\node [rotate=90] (j) at (-1.5,6) {$\cdots$};
	\node [rotate=90] (j11) at (0,6) {$\cdots$};
	
	\node (inf) at (-1.5,7) {$\infty$};
	
	\node [rotate=90] (i12) at (1.5,3) {$\cdots$};
	\node [rotate=90] (j12) at (1.5,6) {$\cdots$};
	\node [rotate=90] (i13) at (3,3) {$\cdots$};
	\node [rotate=90] (j13) at (3,6) {$\cdots$};
	
	\node (11) at (0,-1) {$1$};
	\node (12) at (1.5,-1) {$2$};
	\node (13) at (3,-1) {$3$};
	
	\node (0p) at (11.5,0) {$0$};
	\node (1p) at (11.5,1) {$1$};
	\node (2p) at (11.5,2) {$2$};
	\node [rotate=90] (ip) at (11.5,3) {$\cdots$};
	\node [rotate=90] (i21) at (7,3) {$\cdots$};
	\node (np) at (11.5,4) {$2n$};
	\node (pp) at (11.5,5) {$2n+1$};
	\node [rotate=90] (jp) at (11.5,6) {$\cdots$};
	\node [rotate=90] (j21) at (7,6) {$\cdots$};
	\node (infp) at (11.5,7) {$\infty$};
	\node [rotate=90] (i22) at (8.5,3) {$\cdots$};
	\node [rotate=90] (j22) at (8.5,6) {$\cdots$};
	\node [rotate=90] (i23) at (10,3) {$\cdots$};
	\node [rotate=90] (j23) at (10,6) {$\cdots$};
	\node (21) at (7,-1) {$1$};
	\node (22) at (8.5,-1) {$2$};
	\node (23) at (10,-1) {$3$};

	\node (011) at (0,0) {\tiny{$\emptyset$}};
	\node (111) at (0,1) {\tiny{$\emptyset$}};
	\node (211) at (0,2) {\tiny{$\emptyset$}};
	\node (n11) at (0,4) {\tiny{$\emptyset$}};
	\node (p11) at (0,5) {\tiny{$\emptyset$}};
	\node (inf11) at (0,7) {\tiny{$\emptyset$}};
	\node (inf12) at (1.5,7) {\tiny{$\emptyset$}};
	\node (inf13) at (3,7) {\tiny{$\emptyset$}};
	\node (021) at (7,0) {\tiny{$\emptyset$}};
	\node (121) at (7,1) {\tiny{$\emptyset$}};
	\node (221) at (7,2) {\tiny{$\emptyset$}};
	\node (n21) at (7,4) {\tiny{$\emptyset$}};
	\node (p21) at (7,5) {\tiny{$\emptyset$}};
	\node (inf21) at (7,7) {\tiny{$\emptyset$}};
	\node (inf22) at (8.5,7) {\tiny{$\emptyset$}};
	\node (inf23) at (10,7) {\tiny{$\emptyset$}};
	
	\node (112) at (1.5,1) {\tiny{$\{1+\}$}};
	\node (p12) at (1.5,5) {\tiny{$\{(2n+1)+\}$}};
	\node (122) at (8.5,1) {\tiny{$\{1+\}$}};
	\node (p22) at (8.5,5) {\tiny{$\{(2n+1)+\}$}};
	\node (113) at (3,1) {\tiny{$\{1-\}$}};
	\node (p13) at (3,5) {\tiny{$\{(2n+1)-\}$}};
	\node (123) at (10,1) {\tiny{$\{1-\}$}};
	\node (p23) at (10,5) {\tiny{$\{(2n+1)-\}$}};
	\path[->]
		(111) edge (112)
		(p11) edge (p12)
		(121) edge (122)
		(p21) edge (p22);
	\path[->,bend left = 15] 
		(111) edge (113)
		(p11) edge (p13)
		(121) edge (123)
		(p21) edge (p23);
	\path[->]
		(inf11) edge (inf12)
		(inf21) edge (inf22);
	\path[->,bend left = 15] 
		(inf11) edge (inf13)
		(inf21) edge (inf23);
	
	\node(012) at (1.5,0) {\tiny{$\{0+\}$}};
	\node(212) at (1.5,2) {\tiny{$\{2+\}$}};
	\node(n12) at (1.5,4) {\tiny{$\{(2n)+\}$}};
	\node(023) at (10,0) {\tiny{$\{0+\}$}};
	\node(223) at (10,2) {\tiny{$\{2+\}$}};
	\node(n23) at (10,4) {\tiny{$\{(2n)+\}$}};
	\node(013) at (3,0) {\tiny{$\{0-\}$}};
	\node(213) at (3,2) {\tiny{$\{2-\}$}};
	\node(n13) at (3,4) {\tiny{$\{(2n)-\}$}};
	\node(022) at (8.5,0) {\tiny{$\{0-\}$}};
	\node(222) at (8.5,2) {\tiny{$\{2-\}$}};
	\node(n22) at (8.5,4) {\tiny{$\{(2n)-\}$}};
	\path[->]
		(011) edge (012)
		(211) edge (212)
		(n11) edge (n12)
		(021) edge (022)
		(221) edge (222)
		(n21) edge (n22);
	\path[->,bend left = 15] 
		(011) edge (013)
		(211) edge (213)
		(n11) edge (n13)
		(021) edge (023)
		(221) edge (223)
		(n21) edge (n23);

\end{tikzpicture}
\end{center}

We now develop the counter-examples described in \cite{bezhanishvili10}.
Consider the set $\natb = \nat\cup\{\infty\}$, which is obtained as the Alexandroff-compactification 
of $\nat$ equipped with the discrete topology. Specifically, the open sets of $\natb$ are 
$
	\{U \subseteq \nat\}\cup\{U\cup\{\infty\} \mid U \subseteq \nat \wedge \exists n \in U. \forall m \geq n. m \in u\}.
$
Denote $\natb\oplus\natb\oplus\natb$, the coproduct of three copies of $\natb$, by $3\natb$.
Let us also consider $A = \pow{(\nat\times\{+,-\})}$ 
(with the product topology on $2^{\nat\times\{+,-\}}$, which is 
compact by Tychonoff's theorem).
Define the continuous function $\map{\tau}{3\natb}{\viet{3\natb}}$ as follows:
$
	\tau(i_1) = \{i_2,i_3\} ~~~ \text{and} ~~~ \tau(i_2) = \tau(i_3) = \emptyset,
$
where $i_j$ denotes the $j$-th copy of $i \in \natb$.
Define two continuous functions $\map{\lambda, \lambda'}{3\natb}{A}$
$\lambda(i_1) = \lambda'(i_1) = \emptyset$ for all $i \in \natb$; 
$\lambda(\infty_j) = \lambda'(\infty_j) = \emptyset$ for $j \in \{2,3\}$;
$\lambda(i_2) = \lambda'(i_2) = \{i+\}$, $\lambda(i_3) = \lambda'(i_3) = \{i-\}$, for $i$ odd;
$\lambda(i_2) = \lambda'(i_3) = \{i+\}$,  $\lambda(i_2) = \lambda'(i_3) = \{i-\}$ for $i$ even. 
Altogether, this defines two coalgebras $\alpha = \langle\tau,\lambda\rangle$ and 
$\beta = \langle\tau,\lambda'\rangle$.

In~\cite{bezhanishvili10}, they proved that the following relation (for Stone spaces, relations coincide 
with closed subspaces of a product):
\begin{align*}
	R = \{(i_1,i_1) \mid i \in \natb\}&\cup\{(i_2,i_2), (i_3,i_3) \mid i \in\nat \text{ odd}\}
		\cup\{(i_2,i_3), (i_3,i_2) \mid i \in\nat \text{ even}\}\\
		&\cup\{(\infty_j,\infty_k) \mid j, k\in\{2,3\}\}
\end{align*}
is a Vietoris bisimulation but not an AM-bisimulation.
We can reformulate this as:
\begin{thm}
$R$ is a regular AM-bisimulation but not an AM-bisimulation.
\end{thm}
\noindent 
For the second part of this statement, this means that there is no continuous function
$\map{W}{R}{\viet{R}\times A}$ satisfying the requirement of an AM-bisimulation.
However, there is a relation $W \subseteq (\viet{R}\times A)\times R$ that satisfies the requirement 
of a regular AM-bisimulation as:
\begin{align*}
	W  =&~ \{((\{(i_2,i_2),(i_3,i_3)\}, \emptyset), (i_1,i_1)) \mid i \in\nat \text{ odd}\}\\
		&\cup\{((\{(i_2,i_3),(i_3,i_2)\}, \emptyset), (i_1,i_1)) \mid i \in\nat \text{ even}\}\\
		&\cup\{((\{(\infty_2,\infty_2),(\infty_3,\infty_3)\}, \emptyset), (\infty_1,\infty_1)),
			((\{(\infty_2,\infty_3),(\infty_3,\infty_2)\}, \emptyset), (\infty_1,\infty_1))\}\\
		&\cup\{((\emptyset,\lambda(i_j)), (i_j,i_k)) \mid i \in\natb \wedge (i_j,i_k) \in R\}
\end{align*}
The interesting part is that $(\infty_1,\infty_1)$ is related to two elements, and that  if one of them is removed, 
then $W$ is not closed anymore, and so not a relation in $\Stone$. This explains why this relation cannot 
be restricted to the graph of a continuous function.

\subsection{Toposes for Name-Passing}

In \cite{staton11}, the author studies models of name-passing and their bisimulations.
Three toposes and functors are presented to model different parts of the theory.
The first topos is the category of name substitution, which is the category of presheaves 
over non-empty finite subsets of a fixed countable set, together with all functions between them. 
It comes with a functor combining non-determinism and name-binding.
This functor satisfies strong properties: in particular, AM-bisimulations coincide with 
HJ-bisimulations, and the largest AM-bisimulation coincide with the largest behavioural equivalence.
This framework is already nice as AM-bisimulations describe precisely open bisimulations~\cite{sangiorgi96}.

The second topos is a refinement of the first one, as the category of functors over all finite subsets of 
the given countable set, together with injections. The proposed functor in this case is not as nice:
it does not preserve weak pullbacks, and AM-bisimulations no longer coincide with HJ-bisimulations
anymore. However, it is sufficiently well-behaved in our theory: it covers pullbacks, and the category is a topos, 
thus regular and with pushouts. Consequently, HJ-bisimulations coincide with regular AM-bisimulations, and 
their existence coincides with the existence of a behavioural equivalence.

For this topos, it is noted in \cite{staton11} that if a relation is a HJ-bisimulation (so a regular/toposal 
AM-bisimulation), 
then its $\neg\neg$-completion is an AM-bisimulation. This means, in particular, that this framework for 
name-passing behaves much more nicely when restricted to $\neg\neg$-sheaves. One main reason for this is that 
the sheaf topos for the $\neg\neg$-topology satisfies the axiom of choice when the base topos 
is a presheaf topos over a poset~\cite{maclane92}, which is the case here. 

\subsection{Weighted Linear Systems}

In \cite{bonchi12}, the authors study linear weighted systems, that is, coalgebras for the 
endofunctor $X \mapsto K\times X^A$ on $\Vect{K}$, in the category of 
$K$-vector spaces, with $K$ a field, and $A$ a set. 
The following discussion can also be carried out in the category of modules over a ring.
The category $\Vect{K}$ is 
abelian, and thus regular and with pushouts. The endofunctor in question actually 
preserves pullbacks, so the three notions of bisimilarity coincide by Theorem~\ref{theo:reg-equivalences}. In this 
paper, the focus is on linear bisimulations, which coincide with behavioural 
equivalence, and so to the other two notions of bisimilarities. 

In perspective, 
usual weighted systems are described in the category $\Set$, with the functor
$X \mapsto A \Rightarrow K^{(X)}$ where $K^{(X)}$ is the set of functions from 
$X$ to $K$ that take finitely many non-zero values. In this context, this functor does 
not even cover pullbacks in general, and they actually prove that AM-bisimilarity 
(and so regular 
AM-bisimilarity since $\Set$ has the regular axiom of choice) does not 
coincide with behavioural equivalence.

\section{Conclusion}

This paper introduces some foundations for the theory of bisimulations and 
simulations in a general regular category, mitigating some known issues with 
Aczel-Mendler bisimulations. Relations and power-objects are the key ingredients 
in this mitigation: while the axiom of choice allows for picking some 
witnesses of bisimilarity, relations and power-objects enable us to collect them  
without the need to choose. This paves the way for studying such 
bisimulations in more exotic regular categories and toposes.

One direction for future work is to investigate regular AM-bisimulations for probabilistic systems, in comparison to 
what is done in~\cite{desharnais02,danos06} for behavioural equivalences. 
The main challenge lies in identifying a suitable \emph{regular} category
of ``probabilistic space'' and a ``probabilistic distribution functor'' that \emph{covers pullbacks}.
For the first property, the work on Quasi-Borel spaces~\cite{heunen17}, which yields a quasi-topos, is of interest.
For the second, one possible avenue is to explore categories of $\sigma$-frames (see for example~\cite{simpson12}), 
in which pullbacks do not coincide with those 
in the category of measurable spaces--a solution under investigation.

Another avenue would be to explore other general properties of 
bisimulations, for instance those related to the largest bisimulation or 
to up-to techniques \cite{sangiorgi92}. These approaches require considering (finite or 
infinite) unions of bisimulations, and hence of relations, 
which necessitates working within coherent categories.

\bibliographystyle{alphaurl}
\bibliography{special}

\end{document}